%% file: main.tex
\documentclass[12pt]{article}
\usepackage[utf8]{inputenc}

\usepackage{svgcolor}
\usepackage{csquotes}
\def\showauthornotes{1}
\ifnum\showauthornotes=1

\newcommand{\jnote}[1]{\textcolor{blue}{ {\textbf{(Jeff: #1)}}}}
\newcommand{\anote}[1]{\textcolor{cyan}{ {\textbf{(Aaron: #1)}}}}
\else

\newcommand{\jnote}[1]{}
\newcommand{\anote}[1]{}
\fi

\input{math_commands.tex}

\usepackage[camera]{dtrt}
\usepackage{algorithm}
\usepackage{algorithmic}
\usepackage{amsmath}
\DeclareMathOperator{\arccosh}{arccosh}

\usepackage{comment}

\newcommand{\chebyerror}{\mathsf{ChebyshevCorrection}}

\allowdisplaybreaks

\begin{document}
\title{Sharp \Lovasz-Theta Bounds on Random Graphs}
\author{
Aaron Potechin 
\thanks{{University of Chicago}. \textit{potechin@uchicago.edu}. Supported by NSF grant CCF-2611927.}
\and
Jeff Xu
\thanks{{Toyota Technological Institute at Chicago}.\textit{jeffxusichao@ttic.edu}}
}
\maketitle
\input{theta-function/abstract}

\clearpage
\newpage
\thispagestyle{empty}
\setcounter{page}{0}

\thispagestyle{empty}
\setcounter{page}{0}
\thispagestyle{empty}
\setcounter{page}{0}
\tableofcontents
\clearpage \newpage

\input{theta-function/intro.tex}

\input{theta-function/overview.tex}
\input{theta-function/alt_norm}
\input{cheby_shape.tex}

\clearpage
\newpage

\bibliography{bib}
\bibliographystyle{alpha}
\clearpage
\newpage
\appendix
\crefalias{section}{appendix}

\input{appendix.tex}

\input{def-norm.tex}
\input{polyapprox.tex}
\input{trunc-convergence.tex}
\input{previous_attempt.tex}
\input{lovaszequivalence.tex}
\end{document}

%% file: math_commands.tex
\def \N {\mathbb{N}}

\def \E {\mathbb{E}}

\def \eps {\epsilon}
\def \al {\alpha}

\newcommand{\Var}{\mathrm{Var}}

 \def\1{\bm{1}}

\newcommand{\pE}{\widetilde{\E}}

\usepackage{amsmath,amssymb,amsthm,amsfonts,latexsym,bm,bbm,xspace,graphicx,float,mathtools,mathdots,physics}
\usepackage{braket,caption,subcaption,ellipsis,xcolor,textcomp,hhline,pifont,combelow,booktabs}
\usepackage[colorlinks=true, allcolors=blue]{hyperref}
\usepackage{color}
\usepackage{times}
\usepackage{fullpage}
\usepackage{tikz-cd}
\usepackage[shortlabels]{enumitem}
\usepackage{thm-restate}
\usepackage{cleveref}
\usepackage{mdframed}

\usepackage{mathrsfs}

\newtheorem{theorem}{Theorem}[section]
\newtheorem{lemma}[theorem]{Lemma}
\newtheorem{claim}[theorem]{Claim}
\newtheorem{proposition}[theorem]{Proposition}

\newtheorem{corollary}[theorem]{Corollary}

\newtheorem{definition}[theorem]{Definition}
\newtheorem{remark}[theorem]{Remark}

\newtheorem{observation}[theorem]{Observation}
\newtheorem{example}[theorem]{Example}

\newcommand{\poly}{\operatorname{poly}}

\DeclareMathAlphabet{\mathsfit}{\encodingdefault}{\sfdefault}{m}{sl}
\SetMathAlphabet{\mathsfit}{bold}{\encodingdefault}{\sfdefault}{bx}{n}

\mathchardef\mhyphen="2D
\newcommand{\val}{\text{val}}
\newcommand{\calP}{\mathcal{P}}

\newcommand{\calB}{\mathcal{B}}

\newcommand{\mul}{\text{mul}}

\newcommand{\gam}{\gamma}

\newcommand{\cm}{\mathsf{M} }

\newcommand{\fp}{\mathsf{P}}

\newcommand{\cmp}{\widetilde{\cm}}

\newcommand{\Erdos}{Erd\H{o}s\xspace}
\newcommand{\Renyi}{R\'enyi\xspace}
\newcommand{\Lovasz}{Lov\'asz\xspace}

%% file: theta-function/abstract.tex
It is well known that the \Lovasz-Theta function of a random graph $G(n,\tfrac{1}{2})$ is $\Theta(\sqrt{n})$. More precisely, it is tightly concentrated in the interval
\(
[\sqrt{n},\, 2\sqrt{n}],
\)
where the upper bound follows from an explicit dual witness for the associated semidefinite program. Numerical evidence and heuristic arguments suggest that the true value is $(1+o(1))\sqrt{n}$. However, closing this gap has remained a longstanding challenge, resisting existing techniques even in light of recent progress on sharp algorithmic thresholds and non-asymptotic free probability. In this work, we resolve this question by proving that the \Lovasz-Theta function of $G(n,\tfrac{1}{2})$ is $(1+o_n(1))\sqrt{n}$ with high probability, determining its asymptotic value up to vanishing relative error.

To prove this result, we design a novel and highly intricate iterative process to construct an explicit solution for the SDP in terms of \emph{graph matrices}---a family of structured random matrices whose entries are polynomials of the entries of some underlying input. Our analysis boils down to understanding the precise behavior of a family of graph matrices which we call backbone-correction matrices. We show the following:
\begin{enumerate}
\item With high probability, each normalized backbone-correction matrix $\cmp_{\tau}$ behaves similarly to a Wigner matrix except that it may not be symmetric. In particular, with high probability, each such matrix has norm $2+o_n(1)$ and the empirical distribution of its singular values approaches $\frac{\sqrt{4 - x^2}}{\pi}1_{x \in [0,2]}$ (i.e., the positive part of a semicircle distribution) as $n \to \infty$.
\item Backbone-correction matrices satisfy a form of free independence. In particular, a linear combination $\sum_{\tau}{c(\tau) \cdot \cmp_{\tau}}$ of normalized backbone-correction matrices 
behaves similarly to $\sqrt{\sum_{\tau}{c(\tau)^2}}$ times a Wigner matrix.
\end{enumerate}

\paragraph{AI Disclosure}
This is the full version of the work that will appear in FOCS 2026. All ideas are our own. In the early stage of this attempt (Jan. 2026), J.X. used ChatGPT and Cursor to empirically test the spectrum of several random matrices, and later on to help polish the writing for both the conference submission and the final version.

%% file: theta-function/intro.tex
\section{Introduction}

A central theme in theoretical computer science is to understand the power and limitations of efficient algorithms for combinatorial optimization problems. A dominant approach proceeds via \emph{convex relaxations}, especially semidefinite programs (SDPs). Among the SDP relaxations within this paradigm,
the \emph{\Lovasz-Theta function\footnote{While this is the degree $2$ sum of squares relaxation for independent set rather than the standard definition of the \Lovasz-Theta function, we confirm they are equivalent in \cref{sec:lovaszequivalence}.}} $\vartheta(G)$ is perhaps the most canonical example. 

\begin{definition}[\Lovasz-Theta SDP (Primal)] \label{def:primal-sdp}
For a graph $G$ on $n$ vertices,
\begin{align*}
\vartheta(G) \coloneqq \max \sum_{i\in[n]} \langle v_i, v_0\rangle
\quad &\text{s.t.} \quad
\langle v_i,v_i\rangle=\langle v_i,v_0\rangle\ \forall i,\;
\langle v_0,v_0\rangle=1,\\
& \;\langle v_i,v_j\rangle=0\ \forall \{i,j\} \in E.
\end{align*}
\end{definition}

Introduced by Lovász~\cite{lovasz1979shannon}, the Lovász-Theta function has emerged as a unifying object across multiple areas over the years. In information theory, it upper bounds the Shannon capacity of a graph~\cite{lovasz1979shannon}; in approximation algorithms, it serves as a benchmark for the strength of SDP relaxations; in quantum information, it captures entanglement-assisted parameters~\cite{duan2013zeroerror}; and in spectral graph theory, it interpolates between eigenvalue-based bounds and global combinatorial constraints~\cite{knuth1994sandwich}. From this perspective, $\vartheta(G)$ is not just an SDP relaxation, but a canonical lens for understanding the interplay between continuous methods and discrete optimization. We refer the reader to~\cite{lovasz1979shannon} and the scribe notes~\cite{gupta_witmer2011theta} for other equivalent formulations and a comprehensive introduction to the Lovász-Theta function.

A natural and fundamental question, especially from the perspective of average-case complexity, is: \emph{how sharp is this relaxation on typical instances?}
In this work, we are primarily interested in understanding the \Lovasz-Theta function for random graphs, particularly the \Erdos-\Renyi model $G(n,1/2)$. Over more than four decades, a sequence of results has progressively sharpened our understanding: the asymptotic order $\Theta(\sqrt{n})$ was established by~\cite{juhasz1982theta}, and later refined by~\cite{Coj05}, who showed that $\vartheta(G)$ lies in the interval
\(
\left[\tfrac{1}{2}\sqrt{n},\, 2\sqrt{n}\right],
\)
i.e., within a constant factor of $4$. A related, though somewhat orthogonal, line of work by~\cite{bhaskara_theta} established strong concentration of $\vartheta(G)$ in $G(n,p)$, albeit without identifying the precise location of its expectation.

In the sparse regime, recent advances in random matrix theory and statistical inference—most notably the success of non-backtracking operators \cite{BLM} for sparse random graphs of constant average degree—have led to substantial progress. In particular,~\cite{BKM19, banks2019vectorcoloringsrandomramanujan} identify the sharp threshold for random graphs of bounded average degree. However, their techniques crucially rely on the sparsity of the graph sample and do not translate to the dense regime.

Thus, despite this long line of work, a seemingly basic question remains open:

\begin{center}
What is the \emph{exact leading constant} of $\vartheta(G)$ for $G(n,1/2)$?
\end{center}
A simple argument based on self-complementarity already narrows the bounds to within a factor of two:
\[
\sqrt{n} \;\le\; \vartheta(G) \;\le\; 2\sqrt{n}.
\]
Pinning down the correct constant has proved surprisingly elusive—even in light of numerical evidence suggesting that $\sqrt{n}$ is the correct threshold. This question has recently attracted renewed attention~\cite{lovasz_circulant_blog}, with~\cite{feige2025upperboundsthetafunction} conjecturing a “safe” upper bound of $1.55\sqrt{n}$, drawing on recent developments from non-asymptotic free probability, albeit without a rigorous proof. 

\paragraph{From Higher-Degree SoS Back to Degree-$2$ SoS}

From the perspective of statistical inference, this question fits into the broader agenda of obtaining a fine-grained understanding of higher-degree Sum-of-Squares algorithms on average-case problems. In fact, this problem can be viewed as a necessary hurdle to overcome: since the theta function corresponds to the base level of the SoS hierarchy, resolving it is a natural prerequisite for obtaining sharp phase transitions at higher degrees.

In this work, we revisit this question with the techniques developed for higher-degree Sum-of-Squares (SoS) lower bounds over the years~\cite{BHKKMP16, PR20, Pang21, JPRTX, JPRX23, KPX24, NGCA24, Xu25, PX25}, hoping they can provide additional leverage. In particular, recent progress in understanding the fine-grained behavior of correlated random matrices (\emph{graph matrices}), together with experience in constructing and analyzing PSD matrices, may open up new avenues for attacking this question.

However, at first glance, it is far from clear that such techniques can solve this problem. On the one hand, they appear to be overkill: the \Lovasz-Theta function is captured by the basic SDP, for which the heavy machinery underlying higher-degree SoS lower bounds—such as the approximate PSD factorization of the moment matrix—appears unnecessary. On the other hand, these techniques have generally not been precise enough to capture exact constants. Higher-degree SoS lower bounds often have losses which are polylogarithmic in the dimension; results that are tight up to a constant have only been achieved relatively recently and required considerable effort.

In this paper, we show that, surprisingly, graph matrices and our techniques for analyzing them can be used to give a sharp bound on the value of the \Lovasz-Theta function of a random graph.

\subsection{Our Results}
\paragraph{A Sharp Bound for the \Lovasz-Theta Function on $G(n,1/2)$}
Our primary result is for the \Lovasz-Theta function on $G(n,1/2)$.
\begin{restatable}[Main Theorem for \Lovasz-Theta on $G(n,1/2)$]{theorem}{MainTheta}
\label{thm:main-theta}
For $G \sim G(n,1/2)$, with high probability, we have
\[
\vartheta(G) \leq (1+o_n(1)) \cdot \sqrt{n} \,.
\]
\end{restatable}

Along with the previously known lower bound, our result pins down the value of the \Lovasz-Theta function to an interval of width $o_n(\sqrt{n})$.
\begin{corollary}\(
	\sqrt{n } \leq \vartheta(G) \leq (1+o_n(1)) \cdot \sqrt{n} \,.\)
\end{corollary}

Our approach follows the standard route of constructing a dual witness for the \Lovasz-Theta SDP—equivalently, a moment matrix for the degree-$2$ SoS relaxation. However, in contrast to higher-degree SoS lower bounds for average-case problems, where the pseudo-calibration technique gives a canonical candidate construction, no such candidate exists in our setting. Thus, a significant part of the challenge lies in identifying a suitable construction—moreover, one that is explicit and amenable to analysis.

A central contribution of our work is a novel iterative scheme for constructing such a witness. This scheme differs substantially from the pseudo-calibration approach.

\paragraph{Discussion on Graph Matrices and Free Probability}
As described in~\cref{sec:key-ideas} below, our construction is of the form $F(Q) + \text{small correction}$, where $Q$ is a linear combination of graph matrices of a collection of shapes that we call backbone-correction shapes (see \cref{def:backbone-correction}). A key reason why our construction works is that backbone-correction graph matrices are freely independent\footnote{Technically, $M_{\alpha}$ and $M_{\alpha^{T}}$ are not freely independent and we need an additional condition for the behavior of $M_{\alpha}$ and $\cm_{\alpha^{T}}$. See \cref{sec:trace-power-block-walks} for the mixed-moment calculation for these matrices.}.

As described in the technical overview, most of the technical work of this paper is devoted to proving this statement.

It should be noted that there are currently only a handful of results which precisely analyze restricted families of graph matrices (beyond Wigner matrices and random matrices with i.i.d. entries) or even determine the norms of these matrices up to a constant factor.
\begin{enumerate}
	\item Cai and Potechin \cite{cai2020spectrum,cai2022mixing} determined the spectrum of the singular values for a class of graph matrices called multi-Z-shaped graph matrices.
    \item Jones and Pesenti \cite{JP24} showed that a certain class of scalar graph matrices behaves like independent Gaussian random variables.
	\item When the underlying input is a matrix/tensor with Gaussian entries, several papers \cite{HKPX23, matrixchaos, KX26} have shown norm bounds which are tight up to a constant factor for restricted sets of graph matrices, though not for general graph matrices.
	\item Kothari, Potechin, and Xu \cite{KPX24} showed norm bounds which are tight up to a constant for graph matrices on sparse random graphs $G(n,\frac{d}{n})$ where $d=O(1)$ after deleting a small proportion of vertices from the graph, in particular vertices which have abnormally high degree.

\end{enumerate}
For the linear combinations of graph matrices we consider, since we allow the adjacency matrix of the underlying input to participate as a summand, there is \emph{no} entropy left in the remaining matrices in the sum:  all the remaining matrices can be uniquely determined once a single matrix is revealed. Since the summand matrices share the same underlying randomness, it is not clear whether existing black-box results from non-asymptotic free probability---even in light of the recent progress in~\cite{bandeira2024matrixconcentrationinequalitiesfree}---apply in our setting. Determining whether those techniques can nevertheless be adapted to recover our results is an interesting open question.

\paragraph{Follow-Up Work}
For the closely related ellipsoid fitting problem \cite{SPW13}, which can likewise be formulated as an SDP/degree-$2$ SoS program, the state of affairs was similar to that for the \Lovasz-Theta function. A few years ago, several works \cite{KD22, PTVW22} established bounds tight up to polylogarithmic factors using different techniques. Subsequent works improved these bounds to within constant factors, but did not determine the optimal constant. Among these previous works on the ellipsoid fitting conjecture \cite{bandeira2024fittingellipsoidquadraticnumber, TW25}, the one most closely related to the present work is~\cite{HKPX23}. Shortly after completing this work, together with other collaborators, we applied the framework developed here to resolve the ellipsoid fitting conjecture, establishing the sharp threshold at $d^2/4$ \cite{CPTX26}. Similar results for ellipsoid fitting have also been concurrently and independently obtained in \cite{MW26, KS26} using different techniques.

\paragraph{Acknowledgment}
J.X. thanks Tim Hsieh, Pravesh K. Kothari, and Sidhanth Mohanty for helpful discussions in a previous attempt that identified a construction which gives a constant improvement and is similar to the construction in Appendix E.

%% file: theta-function/overview.tex
\section{Preliminaries and Technical Overview}
\subsection{Folklore Bound of \texorpdfstring{$2\sqrt{n}$}{2 sqrt(n)}} 
To upper bound the value of the primal SDP (see \cref{def:primal-sdp}), we use the following dual SDP:
\begin{definition}[Dual SDP for \Lovasz-Theta]  \label{def:dual-sdp}
	\[ \min \lambda> 0 \quad  \text{ s.t. }  \quad \lambda\cdot I_n - M_G \succeq 0   \]
	where $M_G$ is a dual witness matrix that satisfies the non-edge constraints $M_G[i,j] = 1$ if $\{i,j\}\notin E$ or $i=j$, and $I_n$ is the $n$-dimensional identity matrix.
\end{definition}
The connection with the primal can be seen directly from weak duality.

\begin{claim}
\label{claim:dual-certifies-primal}
Any dual-feasible pair $(\lambda,M_G)$ certifies
\(
\vartheta(G)\leq \lambda.
\)
\end{claim}

\begin{proof}
Let $\{v_0,v_1,\ldots,v_n\}$ be primal-feasible, and let
\(
X[i,j]=\langle v_i,v_j\rangle,
k=\sum_{i\in[n]}\langle v_i,v_0\rangle.
\)
The primal diagonal constraints give $\Tr(X)=k$. Moreover, since $X[i,j]=0$
on edges and $M_G[i,j]=1$ on the diagonal and nonedges,
\[
\langle M_G,X\rangle
=
\sum_{i,j}X[i,j]
=
\left\|\sum_i v_i\right\|^2
\geq
\left\langle \sum_i v_i,v_0\right\rangle^2
=
k^2.
\]
On the other hand, $\lambda I_n-M_G\succeq0$ and $X\succeq0$ imply
\[
\langle M_G,X\rangle
\leq
\lambda\Tr(X)
=
\lambda k.
\]
Thus $k^2\leq\lambda k$. Since $k=\Tr(X)\geq0$, we conclude that
$k\leq\lambda$.
\end{proof}

We now describe the folklore bound of $2\sqrt{n}$ for the \Lovasz-Theta function using the $\pm{1}$-adjacency matrix of the input graph $G$.

\begin{definition}[$\pm{1}$ adjacency matrix $A_G$]
Let $A_G$ denote the $\pm 1$ adjacency matrix of a graph $G$ with zero diagonal, defined by
\[
A_G[i,j] =
\begin{cases}
\;\;\,1 & \text{if } i \neq j \text{ and } (i,j)\in E(G),\\
-1 & \text{if } i \neq j \text{ and } (i,j)\notin E(G),\\
\;\;\,0 & \text{if } i = j.
\end{cases}
\]
\end{definition}

\begin{claim}[Folklore] For $G\sim G(n,1/2)$, with high probability, we have  \[ 
(1-o_n(1))\cdot  \frac{1}{2} \sqrt{ n} \leq \vartheta(G) \leq  (1+o_n(1)) \cdot 2\sqrt{n}\,.
\]
\end{claim}

Consider the dual witness $M_G = I_n-A_G$. Note that this satisfies the constraint that $M_G[i,j] = 1$ whenever $i = j$ or $\{i,j\}\notin E(G)$.
To obtain \[
\lambda \cdot I_n \succeq M_G = I_n - A_G \,,
 \]
it suffices to set $\lambda = (1+o_n(1)) \cdot 2\sqrt{n}$ using the standard result from random matrix theory that with high probability, $\|A_G\| \leq (1+o_n(1))2\sqrt{n}$. The same matrix can also be used to establish a lower bound of value $\sqrt{n}/2$ for the primal program (see~\cref{app:dual-lowerbound}).

Note that there remains a significant gap between the bounds above and the $\sqrt{n}$ lower bound for the primal, which can be established via non-explicit arguments~\cite{mittal_theta_notes}. We highlight this gap as our goal is to obtain an explicit construction that closes it.

Additionally, it should be noted that the analysis for this choice of $M_G$ is essentially tight as it is well known that the spectrum of $A_G$ follows a semicircular law, implying that $M_G = -A_G$ exhibits large positive eigenvalues of magnitude around $2\sqrt{n}$. Thus, any improvement must arise from a different construction of $M_G$—one that preserves the non-edge constraints (i.e., $M_G[i,j] = 1$ whenever $\{i,j\} \notin E(G)$) while addressing the following question:

\begin{displayquote}
	How can we exploit the freedom of the other matrix entries to offset the large eigenvalues of $-A_G$?
\end{displayquote}
In particular, we observe that any such offset needs to be high-rank as it needs to cancel out all of the positive eigenvalues of $-A_G$ which are greater than $\sqrt{n}$, which is a significant portion of the eigenvalues of $-A_G$.

\paragraph{Our Concrete Target} 
Throughout this work, we adopt a formulation of a witness matrix equivalent to the dual SDP in \cref{def:dual-sdp} up to rescaling and a simple manipulation of the equation.

\begin{definition}[(Rescaled) Target Matrix $W$ with Value $c_k$] \label{def:concrete-target}
A symmetric matrix $W = Q(G) \in \mathbb{R}^{n\times n}$ is called a \emph{target matrix} with value $c_k$ if it can be written as
\[
W = I_n + \frac{c_k}{\sqrt{n}} \cdot A_G  + R_W,
\]
where the following hold:
\begin{enumerate}
    \item \textbf{(Non-edge constraints)} $R_W[i,j] = 0$ whenever $i = j$ or $\{i,j\} \notin E(G)$. 
    \item \textbf{(PSDness)} $W \succeq 0$.
\end{enumerate}
\end{definition}
\begin{remark}
    With a slight abuse of notation, we do not distinguish between edge and non-edge constraints, as it will be clear from context that we refer to the same type of constraints in the primal SDP.
\end{remark}
Given a target matrix $W$ with value $c_k$, we can obtain a dual witness matrix $M_G$ with value $\lambda = \frac{\sqrt{n}}{c_k} + 1 = (1+o_n(1))\frac{\sqrt{n}}{c_k}$ by taking 
\[
M_G = I_n - A_G - \frac{\sqrt{n}}{c_k}R_W = \left(\frac{\sqrt{n}}{c_k} + 1\right)I_n -\frac{\sqrt{n}}{c_k}W.
\]
 From this perspective, our goal is to construct such a matrix $W$ with the coefficient $c_k$ as large as possible. In particular, in order to show that the \Lovasz-Theta function is $(1+o_n(1))\sqrt{n} $, we need to give a construction of a target matrix such that $c_k$ approaches $1$. Note that the trivial witness $R_W = 0$ gives a baseline of $c_k = \frac{1}{2} - o_n(1)$.

\paragraph{A Snapshot of Our Construction} 
Before delving into the technical details, we briefly outline the overall construction of the target matrix. At a high level, we design a non-negative function \(F\) and an inner matrix \(Q\) such that
\[
    X
    :=
    \frac{1}{C_F}
    \left(
        F(\widetilde Q)
       +\text{(low-order correction term)}    \right)
\]
approximately satisfies the diagonal and non-edge constraints while achieving an objective value  \(\sqrt{n}\).  Moreover, the spectral estimates below show that
\[
    X \succeq - \eps I
\]
for some \(\eps=o_n(1)\).

We then define
\[
    W:=\varepsilon I+(1-\varepsilon)X 
\]
as our dual witness matrix. The main challenge, therefore, is to construct a function \(F\) whose image is non-negative and which already nearly satisfies the edge constraints. 
We now introduce the technical preliminaries needed to describe our construction and give a road map of our analysis.
\subsection{Graph Matrices: Our Trusty Hammer}
Graph matrices are a powerful tool for analyzing problems on random inputs. They have been crucial for proving SoS lower bounds for average-case problems \cite{BHKKMP16, PR20, Pang21, JPRX23, KPX24, NGCA24, Xu25, PX25}  and have also been used to analyze power-sum decompositions of polynomials \cite{BHKX22}, to analyze the ellipsoid fitting conjecture \cite{PTVW22, HKPX23}, and to analyze a class of first-order iterative algorithms including belief propagation and approximate message passing \cite{JP24}. We begin with a brief overview of graph matrices, specializing the framework to our setting.

\begin{definition}[Shapes]
A shape $\alpha$ is a tuple
\(
    \alpha = \bigl(V(\alpha), U_\alpha, V_\alpha, E(\alpha)\bigr)
\)
associated with a multigraph $\bigl(V(\alpha), E(\alpha)\bigr)$,
where \(U_\alpha, V_\alpha \subseteq V(\alpha)\). We call \(U_\alpha\)
and \(V_\alpha\) the left and right boundaries of \(\alpha\), respectively.

We say that \(\alpha\) is proper if $\bigl(V(\alpha), E(\alpha)\bigr)$ is a graph (i.e., it has no multi-edges or loops) and every isolated vertex of \(\alpha\) lies in \(U_\alpha \cup V_\alpha\).
\end{definition}

Given some underlying input on $n$ vertices, concretely, the  $\pm 1$-adjacency matrix $A_G$ of the graph $G$, we can associate with each shape a matrix whose entries are polynomials in the input entries.
\begin{definition}[Shape Transpose]\label{def:shape-transpose}
Given a shape $\al$, we define its \emph{transpose} $\al^\top$ to be the shape obtained by swapping its boundary vertices while preserving the underlying graph. Formally,
\begin{enumerate}
    \item $V(\al^\top) = V(\al)$ and $E(\al^\top) = E(\al)$;
    \item $U_{\al^\top} := V_\al$ and $V_{\al^\top} := U_\al$.
\end{enumerate}
\end{definition}
\begin{definition}[Graph Matrix of a Shape via Injective Maps] \label{def:graph-matrix-for-shape}
Given a shape $\alpha$, its associated graph matrix $M_{\alpha}$ is indexed by boundary labelings $(S,T)$ (as ordered tuples), with entries defined by
\[
\cm_\al[S,T] = \sum_{\substack{\sigma: V(\al) \to [n]: \ \sigma \text{ is injective,} \\ \sigma(U_{\alpha})=S,\ \sigma(V_{\alpha})=T \\ }}
\;\prod_{e=\{a,b\}\in E({\alpha})} A_G\bigl[\sigma(a),\sigma(b) \bigr].
\]

In this work, we reserve the notation $\cm$ to refer to graph matrices.
\end{definition}
Prior works established rough norm bounds for graph matrices with underlying input from $G(n,\frac{1}{2})$. While these rough norm bounds are sufficient for many applications, since we are proving sharp bounds on $\vartheta(G)$, we will need a much more precise analysis of the graph matrices which appear in our construction.
\begin{theorem}[Rough Norm Bounds for Graph Matrices \cite{AMP20,RT23,TW25norm}]\label{thm:roughnormbounds}
For all proper shapes $\al$, with high probability, $||\cm_{\al
}||$ is $\tilde{O}\left(n^{\frac{|V(\al)|-s_{\al}}{2}}\right)$ where $s_{\al}$ is the minimum size of a vertex separator separating $U_{\al}$ and $V_{\al}$.

More generally, for all improper shapes $\al$, with high probability, $||\cm_{\al}||$ is $\tilde{O}\left(n^{\frac{|V(\al)|-s_{\al} + |Iso(\al)|}{2}}\right)$ where $s_{\al}$ is the minimum size of a vertex separator separating $U_{\al}$ and $V_{\al}$ after deleting all edges with even multiplicity and $Iso(\al)$ is the set of vertices of $\al$ outside of $U_{\al} \cup V_{\al}$ which are isolated after deleting all edges with even multiplicity.
\end{theorem}
In this work, we focus on shapes which have a particularly simple boundary structure: each shape has exactly one boundary vertex on each side, i.e., $|U_{\al}| = |V_{\al}| = 1$. This corresponds to $n \times n$ matrices whose rows and columns are indexed by vertices of $G$.
\begin{definition}
We say that a shape $\alpha$ has boundary size $1$ if $|U_{\alpha}| = |V_{\alpha}| = 1$. For shapes $\alpha$ with boundary size $1$, we define $u_{\alpha}$ to be the vertex in $U_{\alpha}$ and we define $v_{\alpha}$ to be the vertex in $V_{\alpha}$. Note that we may have $u_{\alpha} = v_{\alpha}$. 
\end{definition}

\paragraph{Path Shapes:}
An important class of shapes consists of the path shapes $\{\fp_j: j \in \mathbb{N}\}$ which are closely related to non-backtracking walks of length $j$.
\begin{definition}[Path shape $\fp_j = \fp_j(A_G)$]
For all $j \in \mathbb{N}$, we define the path shape $\fp_j$ to be the shape that is a path of length $j$ from $u_{\fp_j}$ to $v_{\fp_j}$. 
\end{definition}
\begin{figure*}[h]
	\begin{subfigure}[t]{0.48\textwidth}
    \centering
    \includegraphics[height=3cm]{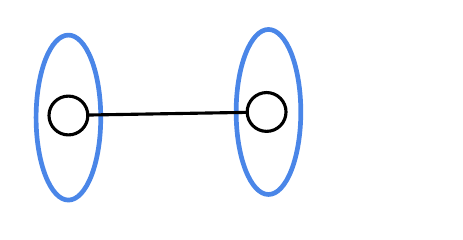}
    \caption{Adjacency Matrix ($\fp_1$)}
    \label{fig:p1}
    \caption*{{$\cm_{\fp_1}[i,j]=  A_G[i,j]$. }}
\end{subfigure}
\hfill
\begin{subfigure}[t]{0.48\textwidth}
    \centering
    \includegraphics[height=3cm]{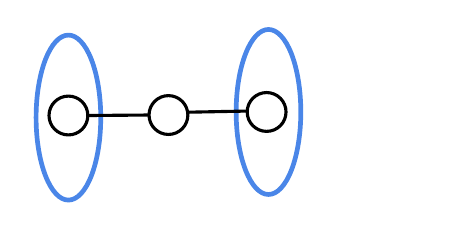}
    \caption{Length-$2$ Non-backtracking Walk ($\fp_2$)}
    \label{fig:p2}
    \caption*{{$\cm_{\fp_2}[i,j]= 1_{i \neq j}\sum_{t\notin\{i,j\}} A_G[i,t] \cdot A_G[t,j]$. }}
\end{subfigure}
\caption{Examples of path shapes}
\end{figure*}

Throughout this work, unless otherwise specified, we use $\fp_j$ as shorthand for $\fp_j(A_G)$. We will later extend this notation and define $\fp_j(M)$ for a general matrix $M$ in \cref{def: j-way-product}, once the full construction is introduced.
For our analysis, it is convenient to normalize the graph matrices we are working with so that their norm is $\Theta(1)$.
\begin{definition}[Normalized Graph Matrices] \label{def:normalized-graph-matrix}
Given a proper shape $\al$ with boundary size $1$, we define the normalized graph matrix $\cmp_{\al }$ to be $\cmp_\al  = n^{-\frac{|V(\al)| - 1_{\text{there is a path from } u_{\al} \text{ to } v_{\al} \text{ in } \al}}{2}}\cm_{\al }$.
\end{definition}
\begin{remark}
This normalization factor is chosen so that with high probabilty, $||\cmp_{\al}||$ is $\tilde{O}(1)$. Until \cref{sec:cheby-shape-approx}, we will focus on proper shapes $\al$ which have boundary size $1$ and a path between $u_{\al}$ and $v_{\al}$ so we will have that $\cmp_{\al}  = n^{-\frac{|V(\al)| - 1}{2}}\cm_{\al}$.
\end{remark}
\begin{example}
For all $j \in \mathbb{N}$, $\cmp_{\fp_j} = n^{-\frac{j}{2}} \cdot \cm_{\fp_j}$.
\end{example}
In Section \ref{sec:graphmatrixmultiplication}, we will use path shapes to illustrate key concepts for our analysis such as shape composition and intersection terms.
\paragraph{Visualizing the Non-Edge Constraints} In our setting, one benefit of working with shapes and graph matrices is that they offer an intuitive way for us to satisfy the non-edge constraints posed by the dual SDP (see \cref{def:dual-sdp}, and the $R_W$-matrix in \cref{def:concrete-target}).

To ensure the non-edge constraints are satisfied, for each off-diagonal entry $i \neq j$, we aim to express
\[
W[i,j]
=
\frac{c_k}{\sqrt{n}} \, A_G[i,j]
+
(1 + A_G[i,j]) \cdot h(i,j),
\]
for some function $h(i,j)$. The key idea is to factor out $(1 + A_G[i,j])$ from all terms beyond the leading contribution of $A_G[i,j]$.  When $(i,j)\notin E(G)$, we have $A_G[i,j] = -1$, and hence $(1 + A_G[i,j]) = 0$, which eliminates all higher-order terms. Thus, this factor ensures that the non-edge constraints are satisfied.

We can impose such an indicator by simply making sure that the shapes come in pairs $\beta$ and $\bar{\beta}$ where $\bar{\beta}$ is defined as follows:
\begin{definition}[Correction Shape $\bar{\beta}$ for $\beta$]
Given a shape $\beta$ with boundary size $1$ such that $u_{\beta} \neq v_{\beta}$ and $\{u_{\beta},v_{\beta}\} \notin E(\beta)$, we define the correction $\bar{\beta}$ for $\beta$ to be the shape obtained by adding the edge $\{u_{\beta},v_{\beta}\}$ to $\beta$. More precisely, $\bar{\beta}$ is the shape such that $V(\bar{\beta}) = V(\beta)$, $u_{\bar{\beta}} = u_{\beta}$, $v_{\bar{\beta}} = v_{\beta}$, and $E(\bar{\beta}) = E(\beta) \cup \{u_{\beta},v_{\beta}\}$.
\end{definition}
\begin{remark}
We use the letter $\beta$ in this definition as the main shapes we will correct for will be backbone shapes $\beta$ (see \cref{def:backbone-correction}).
\end{remark}
For example, while $\cm_{\fp_2}$ does not satisfy the non-edge constraints, $\cm_{\fp_2} + \cm_{\overline{{\fp}}_2}$ does satisfy the non-edge constraints as $\cm_{\fp_2}[i,j] + \cm_{\overline{\fp}_2}[i,j] = 0$ whenever $\{i,j\} \notin E(G)$.
\begin{figure}[h]
    \centering
    \includegraphics[height=3cm]{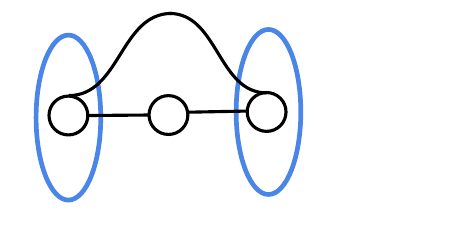}
    \caption{$\overline{\fp}_2$: $\cm_{\overline{\fp}_2}[i,j]= \sum_{t\notin\{i,j\}} A_G[i,t] \cdot A_G[t,j]\cdot A_G[i,j] $}
    \label{fig:p2bar}
        \label{fig:P2'}
\end{figure}
Thus, our task can be rephrased in the language of graph matrices as follows:
we seek a target matrix
\[
    W
    =
    I_n + \frac{c_k}{\sqrt n} A_G + R_W
\]
such that \(W \succeq 0\), where \(R_W\) is a linear combination of backbone
shapes paired with their corresponding correction terms. Concretely, we require
\[
    R_W
    =
    \sum_{\substack{
        \beta:\ |V(\beta)| \ge 2,\\
        \{u_\beta,v_\beta\} \notin E(\beta)
    }}
    c_{\beta}
    \bigl(\cm_\beta+\cm_{\bar{\beta}}\bigr),
\]
for some coefficients $\{c_\beta:\ |V(\beta)| \ge 2,\ \{u_\beta,v_\beta\}\notin E(\beta)\}$. 
The goal is to maximize the leading coefficient \(c_k\).
\subsection{Graph Matrix Multiplication: Shape Concatenation and Intersection Terms}\label{sec:graphmatrixmultiplication}
In order to describe our construction and analysis of this construction, we need to describe what happens when we multiply two or more graph matrices together.
\paragraph{Overview of Graph Matrix Multiplication}
Consider the product of two graph matrices $M_{\al_1}$ and $M_{\al_2}$. Since 
\begin{enumerate}
\item $\cm_{\al_1}[S,T] = \sum_{\substack{\sigma: V(\al_1) \to [n]: \ \sigma \text{ is injective,} \\ \sigma(U_{\al_1})=S,\ \sigma(V_{\al_1})=T \\ }}
\;\prod_{e=\{a,b\}\in E(\al_1)} A_G\bigl[\sigma(a),\sigma(b)\bigr]$
\item $\cm_{\al_2}[S',T'] = \sum_{\substack{\sigma': V(\al_2) \to [n]: \ \sigma' \text{ is injective,} \\ \sigma'(U_{\al_2})=S',\ \sigma'(V_{\al_2})=T' \\ }}
\;\prod_{e=\{a,b\}\in E(\al_2)} A_G\bigl[\sigma'(a),\sigma'(b)\bigr]$
\end{enumerate}
we have that 
\[
\cm_{\al_1}\cm_{\al_2}[S,T'] = \sum_{\substack{\sigma'': V(\al_1) \cup V(\al_2) \to [n]: \ \sigma'' \text{ is injective on } V(\al_1), \\
\sigma'' \text{ is injective on } V(\al_2), \ \sigma''(V_{\al_1}) = \sigma''(U_{\al_2}), \\ \sigma''(U_{\al_1})=S, \ \sigma''(V_{\al_2})=T'}}
\;\prod_{e=\{a,b\}\in E(\alpha) \cup E(\beta)} A_G\bigl[\sigma''(a),\sigma''(b)\bigr]
\]

In this expression, it would be nice if we had the condition that $\sigma''$ is injective (except for setting $\sigma''(V_{\al_1}) = \sigma''(U_{\al_2})$) rather than the conditions that $\sigma''$ is injective on $V(\al_1)$ and injective on $V(\al_2)$. In this idealized setting, multiplying $\cm_{\al_1}$ and $\cm_{\al_2}$ would correspond to concatenating $\al_1$ and $\al_2$ (see \cref{def:shape-concate} below). 

However, this clean picture breaks down because injectivity is enforced only within $V(\al_1)$ and $V(\al_2)$ individually. In general, vertices in $V(\al_1) \setminus V_{\al_2}$ and vertices in $V(\al_2) \setminus U_{\al_1}$ may collide. This leads to additional \emph{intersection terms}, which capture all such unintended identifications. These intersection terms are well known in prior works to be a technically challenging component of the analysis, often requiring delicate control. 

In most SoS lower bounds on average-case problems \cite{BHKKMP16, PR20, JPRTX, JPRX23, KPX24, PX25}, it turns out that intersection terms are not the dominant terms in the analysis so the focus is on the main terms where there are no unexpected intersections. For these SoS lower bounds, intersection terms are a subtle nuisance which needs to be carefully handled but which can be ignored for the purpose of giving the high-level ideas.

For our setting, intersection terms are not just a technicality but a central part of our analysis. In particular, there is a specific class of \emph{well-behaved} intersection terms that we call \emph{backtracking intersections} which give important terms in our analysis. 

\paragraph{Proper Concatenation} In light of this discussion, we first isolate the contribution corresponding to the idealized setting where no unintended vertex collisions occur. This leads to the notion of \emph{proper concatenation}, which captures the clean decomposition of shapes along their shared boundary without introducing any additional intersections. We now formalize this operation.
\begin{definition}[(Proper) Shape Concatenation]\label{def:shape-concate}
Given shapes $\al_1$ and $\al_2$ such that $|V_{\al_1}| = |U_{\al_2}|$, we take the concatenation $\al_1 \circ \al_2$ of $\al_1$ and $\al_2$ to be the shape such that:
\begin{enumerate}
		\item $V(\al_1 \circ \al_2) = V(\al_1) \cup V(\al_2)$ where the vertices in $V_{\al_1}$ are identified with the vertices in $U_{\al_2}$.
        \item $U_{\al_1 \circ \al_2} = U_{\al_1}$ and $V_{\al_1 \circ \al_2} = V_{\al_2}$.
		\item $E(\al_1 \circ \al_2) = E(\al_1) \cup E(\al_2)$. 
\end{enumerate}
In other words, $\al_1 \circ \al_2$ is the shape obtained by gluing $\al_1$ and $\al_2$ together along $V_{\al_1}$ and $U_{\al_2}$, keeping all other vertices distinct, and taking all of the edges in both $\al_1$ and $\al_2$.
	
This definition extends naturally to a sequence of $j$ shapes $\al_1,\ldots,\al_j$ such that for each $i \in [j-1]$, $|V_{\al_i}| = |U_{\al_{i+1}}|$. We can obtain the concatenation $\al_1 \circ \cdots \circ \al_j$ by concatenating these shapes together one by one (it is not hard to check that the order of the concatenations does not matter). We say that such shapes $\al_1,\ldots,\al_j$ are composable and call this a $j$-way concatenation.
\end{definition}

We give some examples of proper shape concatenations below. Throughout this work, we use blue ovals to illustrate the left and right sides of shapes and use purple dotted ovals to denote shape boundaries which were glued together during concatenations.
\begin{figure}[h]
    \centering
    \begin{subfigure}[t]{0.48\textwidth}
        \centering
        \includegraphics[width=\linewidth]{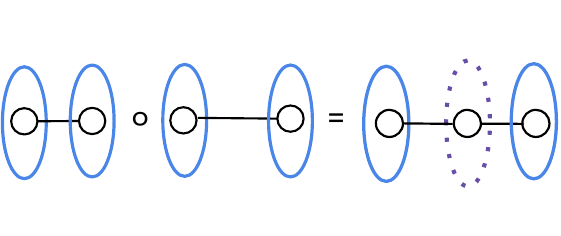}
        \caption{$\fp_1 \circ \fp_1 = \fp_2$}
        \label{fig:path-concatenation}
    \end{subfigure}
    \hfill
    \begin{subfigure}[t]{0.48\textwidth}
        \centering
        \includegraphics[width=\linewidth]{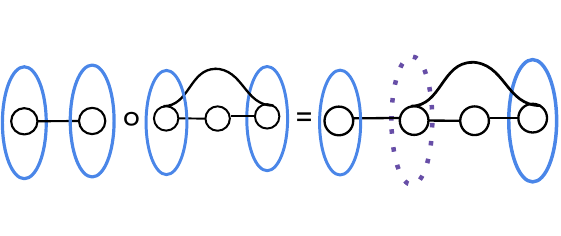}
        \caption{Concatenation of $2$ different shapes}
        \label{fig:mixed-concatenation}
    \end{subfigure}
    
      \centering
    \begin{subfigure}[t]{0.45\textwidth}
        \centering
        \includegraphics[width=\linewidth]{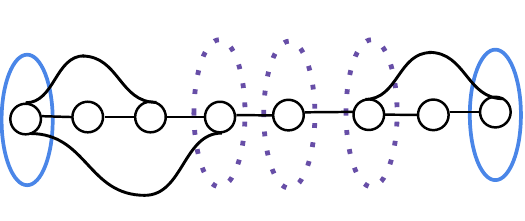}
        \caption{A 4-way concatenation}
        \label{fig:multi-concate}
    \end{subfigure}
    \hfill
    \begin{subfigure}[t]{0.45\textwidth}
        \centering
        \includegraphics[width=\linewidth]{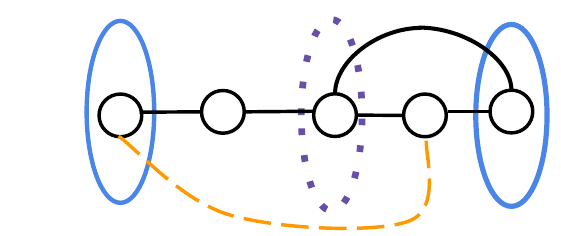}
        \caption{Intersection pattern example}
        \label{fig:eg-intersection}
    \end{subfigure}
    \caption{Examples of proper concatenations and intersections.}
\end{figure}

\paragraph{Intersection Patterns}
We now discuss the intersection terms.
\begin{definition}[Intersection patterns]
Given composable shapes $\al_1,\ldots,\al_j$, we define an intersection pattern $I$ on $\al_1,\ldots,\al_j$ to be a partition of the vertices of $\al_1 \circ \ldots \circ \al_j$ into equivalence classes such that for all $i \in [j]$, no two vertices of $\al_i$ are in the same equivalence class. 

We define $Int_{\al_1,\ldots,\al_j}$ to be the set of possible intersection patterns on $\al_1,\ldots,\al_j$.
\end{definition}
\begin{remark}
In this paper, unlike many previous papers, we consider the trivial intersection pattern where each equivalence class contains a single vertex (i.e., there are no intersections) to be a valid intersection pattern.
\end{remark}
Visually, we represent an intersection pattern $I$ by starting with $\al_1 \circ \ldots \circ \al_j$ and drawing constraint/intersection edges between vertices which are in the same equivalence class of $I$.
\begin{definition}
Given an intersection pattern $I \in Int_{\al_1,\ldots,\al_j}$,
\begin{enumerate}
\item When vertices $u,v \in V(\al_1 \circ \ldots \circ \al_j)$ are in the same equivalence class of $I$, we denote this by writing $u \sim v$.
\item For each $v \in V(\al_1 \circ \ldots \circ \al_j)$, we define $\hat{v}$ to be the equivalence class of $I$ containing $v$.
\end{enumerate}
\end{definition}
\begin{definition}
Given an intersection pattern $I \in Int_{\tau_1,\ldots,\tau_j}$, we define the resulting shape $\gamma_I$ to be the (improper) shape obtained by starting with $\al_1 \circ \ldots \circ \al_j$ and then identifying all of the vertices in each equivalence class of $I$.
\end{definition}
\begin{proposition}
For all shapes $\al_1$ and $\al_2$ such that $|V_{\al_1}| = |U_{\al_2}|$, 
${\cm}_{\al_1}{\cm}_{\al_2} = \sum_{I \in Int_{\al_1,\al_2}}{{\cm_{\gamma_I}}}$. More generally, for all composable shapes $\al_1,\ldots,\al_j$, ${\cm}_{\al_1}\ldots{\cm}_{\al_j} = \sum_{I \in Int_{\al_1,\ldots,\al_j}}{{\cm_{\gamma_I}}}$.
\end{proposition}
To gain intuition for multiplying graph matrices and our analysis, it is very useful to consider what happens when we multiply ${\cmp}_{\fp_{j}}$ and ${\cmp}_{\fp_1}$.
\begin{lemma}\label{lem:pathproduct}
For all $j \in \mathbb{N}$, ${\cmp}_{\fp_{j}}{\cmp}_{\fp_1} \approx {\cmp}_{\fp_{j+1}} + {\cmp}_{\fp_{j-1}}$. More precisely, with high probability, $||{\cmp}_{\fp_{j}}{\cmp}_{\fp_1} - ({\cmp}_{\fp_{j+1}} + {\cmp}_{\fp_{j-1}})||$ is $o_n(1)$.
\end{lemma}
\begin{proof}
Let $u_{\fp_j},w_1,\ldots,w_{j-1},v_{\fp_j}$ be the vertices of $\fp_{j}$ and let $v^{*} = v_{\fp_1}$. Observe that $Int_{\fp_j,\fp_1}$ consists of the following intersection patterns:
\begin{enumerate}
\item The null intersection pattern $I_{null}$ with no intersections. Note that $\gamma_{I_{null}} = \fp_j \circ \fp_1 = \fp_{j+1}$.
\item The intersection pattern $I_{back}$ where $v^{*} \sim w_{j-1}$. Note that $\cm_{\gamma_{I_{back}}} = (n-j)\cm_{P_{j-1}}$ as $\gamma_{I_{back}}$ is the improper shape consisting of $\fp_{j-1}$ together with a double edge between $v_{\gamma_{I_{back}}}$ (which is the equivalence class $v^{*} \sim w_{j-1}$) and $v_{\fp_j}$. Since the entries of $A_G$ are $\pm{1}$, the double edge vanishes leaving $v_{\fp_j}$ as an isolated vertex. Since there are $j$ other vertices $\{u_{\fp_j},w_1,\ldots,w_{j-2}.v_{\gamma_{I_{back}}}\}$ in $\gamma_{I_{back}}$, once these vertices have been mapped to $[n]$, there are $n-j$ choices for where $v_{\fp_j}$ is mapped to.

We call ${I_{back}}$ a backtracking intersection; see \cref{def:backtracking-intersections} for the general definition.
\item For each $j' \in \{0,1,\ldots,j-2\}$, we have an intersection pattern $I_{j'}$ where $v^{*} \sim w_{j-1}$ (where we take $w_0 = u_{\fp_j}$). Note that with high probability, $||M_{\gamma_{j'}}||$ is $\tilde{O}\left(n^{\frac{j}{2}}\right)$ as $\gamma_{I_{j'}}$ is a proper shape with $j+1$ vertices which contains a path between $u_{\gamma_{I_{j'}}}$ and $v_{\gamma_{I_{j'}}}$.
\end{enumerate}
Putting these pieces together, we have that 
\begin{align*}
{\cmp}_{\fp_{j}}{\cmp}_{\fp_1} = n^{-\frac{j+1}{2}}{\cm}_{\fp_{j}}{\cm}_{\fp_1} &= n^{-\frac{j+1}{2}}\left(\cm_{\fp_{j+1}} + (n-j)\cm_{\fp_{j-1}} + \sum_{j'=0}^{j-2}{\cm_{\gamma_{I_{j'}}}}\right)\\
&= \cmp_{\fp_{j+1}} + \cmp_{\fp_{j-1}} +  n^{-\frac{j+1}{2}}\left(-j\cm_{\fp_{j-1}} + \sum_{j'=0}^{j-2}{\cm_{\gamma_{I_{j'}}}}\right)
\end{align*}
By \cref{thm:roughnormbounds}, with high probability, $||{\cmp}_{\fp_{j}}{\cmp}_{\fp_1} - ({\cmp}_{\fp_{j+1}} + {\cmp}_{\fp_{j-1}})||$ is $o_n(1)$, as needed.
\end{proof}
\paragraph{Small powers of $\cmp_{\fp_1}$}
To gain further intuition for our analysis, it is instructive to consider powers of $\cmp_{\fp_1}$. By \cref{lem:pathproduct}, the first few powers of $\cmp_{\fp_1}$ can be approximated as follows.
\begin{enumerate}
\item $(\cmp_{\fp_1})^2 \approx \cmp_{\fp_2} + Id_n$.
\item $(\cmp_{\fp_1})^3 \approx (\cmp_{\fp_2} + Id_n)\cmp_{\fp_1} \approx \cmp_{\fp_3} + 2\cmp_{\fp_1}$.
\item $(\cmp_{\fp_1})^4 \approx (\cmp_{\fp_3} + 2\cmp_{\fp_1})\cmp_{\fp_1} \approx \cmp_{\fp_4} + 3\cmp_{\fp_2} + 2Id_n$.
\item $(\cmp_{\fp_1})^5 \approx (\cmp_{\fp_4} + 3\cmp_{\fp_2} + 2Id_n)\cmp_{\fp_1} \approx \cmp_{\fp_5} + 4\cmp_{\fp_3} + 5\cmp_{\fp_1}$.
\item $(\cmp_{\fp_1})^6 \approx (\cmp_{\fp_5} + 4\cmp_{\fp_3} + 5\cmp_{\fp_1})\cmp_{\fp_1} \approx \cmp_{\fp_6} + 5\cmp_{\fp_4} + 9\cmp_{\fp_2} + 5Id_n$.
\end{enumerate}
It is illuminating to reverse this process to write $\cmp_{\fp_j}$ as a polynomial in $\cmp_{\fp_1}$. The first few such polynomials are as follows:
\begin{enumerate}
\item $\cmp_{\fp_2} \approx (\cmp_{\fp_1})^2 - Id_n$.
\item $\cmp_{\fp_3} \approx (\cmp_{\fp_1})^3 - 2\cmp_{\fp_1}$.
\item $\cmp_{\fp_4} \approx (\cmp_{\fp_1})^4 - 3((\cmp_{\fp_1})^2 - Id_n) - 2Id_n \approx (\cmp_{\fp_1})^4 - 3(\cmp_{\fp_1})^2 + Id_n$.
\item $\cmp_{\fp_5} \approx (\cmp_{\fp_1})^5 - 4((\cmp_{\fp_1})^3 - 2\cmp_{\fp_1}) - 5\cmp_{\fp_1} \approx (\cmp_{\fp_1})^5 - 4(\cmp_{\fp_1})^3 + 3\cmp_{\fp_1}$.
\item \begin{align*}
\cmp_{\fp_6} &\approx (\cmp_{\fp_1})^6 -5\left((\cmp_{\fp_1})^4 - 3(\cmp_{\fp_1})^2 + Id_n\right) - 9\left((\cmp_{\fp_1})^2 - Id_n\right) - 5Id_n \\
&\approx (\cmp_{\fp_1})^6 -5(\cmp_{\fp_1})^4 + 6(\cmp_{\fp_1})^2 - Id_n.
\end{align*}
\end{enumerate}
Note that if we replace $(\cmp_{\fp_1})$ by $2x$, this gives the Chebyshev polynomials of the second kind! In~\cref{sec:graph-mat-cheb}, we elaborate on this connection and observe that it holds for a whole class of graph matrices, backbone-correction shapes, which is the main class of graph matrices that we analyze in this work.

\subsection{Initial Attempts with Constant Improvement}\label{sec:initialattempts}
Before discovering our construction, we made several initial attempts that gave a better bound than $2\sqrt{n}$ but fell well short of $\sqrt{n}$. We present these attempts as they give further intuition for the problem.

Our first attempt was as follows:
\paragraph{Explicit Factorization with Value 1.732} We start with
$\left(a\ \cdot Id + b \cdot \cmp_{\fp_1} \right)^2$ as an initial PSD guess and make corrections as needed. Observe that 
\[
\left(a \cdot Id + b \cdot \cmp_{\fp_1} \right)^2 \approx (a^2 + b^2)Id + 2ab \cdot \cmp_{\fp_1}  + {b^2}\cdot \cmp_{\fp_2}\,.
\]
In order to satisfy the clique constraints, we need to add ${b^2}\cdot  \cmp_{\overline{\fp}_2}$. However, the resulting matrix may not be PSD. Since $||{\cmp_{\fp_2}}|| \approx 2$, we need to add $2{b^2}Id$ to ensure that the matrix is PSD. Thus, our final candidate is approximately
\[
(a^2 + 3b^2)Id + 2ab {\cmp_{\fp_1}} + {b^2} \cmp_{\fp_2} + {b^2}\cmp_{\overline{\fp}_2}
\]
The ratio $\frac{2ab}{a^2 + 3b^2}$ is maximized when $b = \frac{a}{\sqrt{3}}$ at which it takes a value of $\frac{1}{\sqrt{3}} \approx 0.577$. Renormalization gives us an upper bound value of $\approx 1.732$.

We can improve this ratio significantly by considering more terms. For example, as we show in~\cref{sec:previous-attempt-appendix}, by considering 
$\left(aId + b\frac{\cm_{\fp_1}}{\sqrt{n}} + c {\cmp_{\fp_2}} + d \cmp_{\overline{\fp}_2}\right)^2$ and then making corrections as needed, we can obtain an upper bound value of $1.388$.

While this perspective can yield constant-factor improvements over the baseline $2\sqrt{n}$ bound, it remains far from clear how to develop a systematic construction that approaches the correct constant of $1$, or even $1+\eps$ for a small constant $\eps>0$.

\paragraph{Attempt via Polynomial Approximation}
Another approach, inspired by the line of work for the sparse regime \cite{banks2019vectorcoloringsrandomramanujan, BKM19}, is to consider functions/polynomials in the normalized adjacency matrix which are PSD and then add correction terms as needed.

For example, among nonnegative functions with second moment $2$ under the radius-$2$ semicircle distribution $SC(2)$, the following function maximizes $\E_{x\sim SC(2)}[xF(x)]$:
\[ 
F(x) = \begin{cases}
	2x \quad &\text{if } x \in [0,2]\\
	0 \quad &\text{if } x\in [-2, 0]\,.
\end{cases}
\]
To analyze $F(\cmp_{\fp_1})$, it turns out to be very useful to decompose $F$ in terms of Chebyshev polynomials. Let
\[
P_j(x) = U_j\!\left(\frac{x}{2}\right)
\]
where $U_j$ is the $j$-th Chebyshev polynomial of the second kind. As we show in \cref{sec:cheby-shape-approx}, these polynomials have the following properties:
\begin{enumerate}
\item The polynomials $\{P_j(x): j \in \mathbb{N}\}$ are an orthonormal basis for the semicircle distribution $f(x) = \frac{\sqrt{4-x^2}}{2\pi}$.
\item $P_j({\cmp_{P_1}}) \approx \cmp_{\fp_j}$.
\end{enumerate}
Expanding the above via Chebyshev polynomials (of the second kind), we have 
\[ 
F(x) = C_F + x + \sum_{j=2}^{\infty} b_j \cdot P_j(x) 
\]
where
\[ 
C_F \coloneqq \E_{x\sim SC(2)}[F(x)] = \frac{8}{3\pi}\,.
\] 
By symmetry, $\E_{x\sim SC(2)}[F(x)^2] = 2\E_{x\sim SC(2)}[x^2] = 2$. Parseval's identity, including the unit coefficient of $P_1(x)=x$, therefore gives
\[
2 = \E_{x\sim SC(2)}[F(x)^2] = C_F^2 + 1 + \sum_{j=2}^{\infty} b_j^2.
\]
Thus $\sum_{j=2}^{\infty} b_j^2 = 1 - C_F^2$, while $\operatorname{Var}_{x\sim SC(2)}(F(x)) = 2 - C_F^2$.

If we consider the matrix 
\[
F(\cmp_{\fp_1}) \approx {C_F}I_n + {\cmp_{\fp_1}} + \sum_{j=2}^{\infty}{b_j} \cmp_{\fp_j}
\]
we would need to add a correction $\sum_{j=2}^{\infty}{{b_j}\cmp_{\overline{\fp}_j}}$ in order to satisfy the clique constraints. The norm of this correction term would be $2\sqrt{\sum_{j=2}^{\infty}{b_j^2}} = 2\sqrt{1 - C_F^2}$ so after adding $2\sqrt{1 - C_F^2}Id$ to restore PSDness, we obtain a ratio of $\frac{1}{C_F + 2\sqrt{1 - C_F^2}} \approx 0.5246$ which is worse than our first attempt. That said, as we describe below, our construction is a modification of this attempt.

\subsection{Our Construction of a Witness Matrix} \label{sec:key-ideas}
Now that we have the needed background on graph matrices and have discussed some initial attempts which fell short, we are ready to describe our construction of a witness matrix showing that with high probability, the \Lovasz-Theta number of a random graph is at most $(1+o_n(1))\sqrt{n}$.

The key idea for our construction is that instead of considering $F({\cmp}_{\fp_{1}})$, we can consider $F(Q)$ for any matrix $Q$ which has the same behavior as a Wigner random matrix with spectrum contained within $[-2,2]$.  Using this freedom, we can choose $Q$ so that up to small errors which we sweep under the rug for now, no correction terms are needed because $F(Q)$ contains its own correction terms! It turns out that when we do this, the coefficient of ${\cmp}_{\fp_{1}}$ in $F(Q)$ is $C_F$ so the coefficients of $I_n$ and $\cmp_{\fp_1}$ in $F(Q)$ are both $C_F$. Taking $W \approx \frac{F(Q)}{C_F}$ gives a target matrix with value $c_k = 1$ which proves our result.

\paragraph{Notation.}
We fix some notation. Our construction of $Q$ is an iterative construction. We index the iteration levels by $i$ and we use $j$ to denote the number of shapes in a concatenation.

For a linear combination of normalized graph matrices $Q_i$, let $\calB(Q_i)$ denote the collection of shapes that appear in $Q_i$, which we call base shapes. For each base shape $\tau \in \calB(Q_i)$, we denote its (normalized) coefficient by $c({\tau})$. In other words, we have that $Q_i = \sum_{\tau \in \calB(Q_i)}{c(\tau)\cdot \cmp_{\tau}}$.

In our construction, $\calB(Q_i)$ will grow monotonically with $i$ (i.e., for all $i < i'$, $\calB(Q_i) \subseteq \calB(Q_{i'})$) and once a shape $\tau$ is added to some $\calB(Q_i)$, its coefficient $c_{\tau}$ never changes. As in our second attempt in Section \ref{sec:initialattempts}, we take 
\[ 
F(x) = \begin{cases}
	2x \quad &\text{if } x \in [0,2]\\
	0 \quad &\text{if } x\in [-2, 0]\,.
\end{cases}
\]
and we have the decomposition $F(x) = C_F + x + \sum_{j=2}^{\infty} b_j \cdot P_j(x)$ where $C_F = \frac{8  }{3 \pi }$ and $\sum_{j=2}^{\infty}{b_j^2} = 1 - C_F^2$.

\paragraph{Iterative Construction}

We start by describing our initialization and the first iteration. For clarity, we present the recursive construction in terms of normalized graph matrices.
\begin{mdframed}[linewidth=0.4pt]
\textbf{Initialization and first iteration:}

\begin{enumerate}
\item \textbf{Iteration level $i=0$.}  
Define
\[
Q_0 \coloneqq {C_F}\cdot  \cmp_{\fp_1}.
\]
The base shape collection for $Q_0$ is
\[
\calB(Q_0) = \{\fp_1\}.
\]
\item \textbf{Iteration level $i=1$.}  
When we take $F(Q_0)$, we have the terms $\sum_{j \geq 2}{ b_j {C_F^j} \cmp_{\fp_j}}$ which need to be corrected. To correct these terms, we define
\[
Q_1 \coloneqq Q_0 + \sum_{j \geq 2}{ b_j {C_F^j} \cmp_{\overline{\fp}_j}} .
\]
The base-shape collection for $Q_1$ is
\[
\calB(Q_1) = \{\fp_1\} \cup \{\overline{\fp}_j: j \geq 2\},
\]
where $c(\fp_1) = C_F$ and for all $j \geq 2$, $c({\overline{\fp}_j}) = b_j{C_F^j}$.
\end{enumerate}

\end{mdframed}
Our general recursive step is as follows:
\begin{mdframed}[linewidth=0.4pt]
\textbf{Recursive Update ($i \to i+1$):}

\begin{enumerate}
\item For each $j \geq 2$, consider the set of $j$-way concatenations of base shapes in $Q_i$:
\[
\calP_j(Q_i) :=
\{\tau_1 \circ \dots \circ \tau_j : \tau_t \in \calB(Q_i) \forall i \in [j]\}.
\]
\item Partition $\calP_j(Q_i)$ depending on whether the product uses a newly introduced base shape:
\[
\textsf{Old-}\calP_j(Q_i)
=
\{\beta \in \calP_j(Q_i) :
\tau_t \in \calB(Q_{i-1}) \text{ for all } t \in [j]\},
\]
\[
\textsf{New-}\calP_j(Q_i)
=
\{\beta \in \calP_j(Q_i) :
\exists t \text{ such that }
\tau_t \in \calB(Q_i)\setminus \calB(Q_{i-1})\}.
\]
\item For each $\beta=\tau_1\circ\dots\circ\tau_j \in \textsf{New-}\calP_j(Q_i)$, note that $\cmp_{\beta}$ appears in $F(Q_i)$ with coefficient
\[
c({\beta}) \coloneqq b_j \cdot  \prod_{t=1}^{j} c({\tau_t}),
\]
and $Q_i$ does not contain a correction for this term.

\item To correct for these terms, we take $c(\bar{\beta}) \coloneqq    c(\beta)$ and make the update 
\[
Q_{i+1}-Q_i
\coloneqq
\sum_{j\ge 2}
\sum_{\beta \in \textsf{New-}\calP_j(Q_i)}
c({\bar{\beta}}) \cdot \cmp_{\bar{\beta}} .
\]
Equivalently, the base shape collection evolves as
\[
\calB(Q_{i+1})
\coloneqq
\calB(Q_i)
\cup
\{\bar{\beta} : \beta \in \textsf{New-}\calP_j(Q_i)\},
\]
where for each $\beta \in \textsf{New-}\calP_j(Q_i)$, $c({\bar{\beta}}) = b_j\prod_{t=1}^{j} c(\tau_t)$.
\end{enumerate}

\end{mdframed}
\begin{remark}
    In our formal analysis, we will consider an inner-truncation such that for any $Q_i$ during the iterative process, we only consider $P_j(Q_i)$ for 
    $j\leq D$.
\end{remark}

For our final $Q$, we will take $Q = Q_{t}$ for some truncation parameter $t$ and then handle the remaining small correction terms directly. This is addressed in~\cref{sec:real-world-proof}. That said, for the purposes of this overview, we pretend that we take $Q$ to be the limit of $Q_i$ as $i \to \infty$.
\paragraph{Backbone and Backbone-Correction Shapes.}
Based on this construction, we define backbone shapes and backbone-correction shapes as follows:
\begin{definition}[Backbone Shapes and Backbone-Correction Shapes]  \label{def:backbone-correction} \ 
\begin{enumerate}
\item We say that $\tau$ is a backbone-correction shape if $\tau = \fp_1$ or $\tau = \bar{\beta}$ for some backbone shape $\beta$.
\item We say that $\beta$ is a backbone shape if $\beta = \tau_1 \circ \cdots \circ \tau_j$ for some backbone-correction shapes $\tau_1,\ldots,\tau_j$ where $j \geq 2$.
\end{enumerate}
Equivalently, a backbone shape is a shape $\beta$ such that 
\begin{enumerate}
\item \textbf{(Non-backtracking Path through All Vertices)} $V(\beta) = \{w_j: j \in [l] \cup \{0\}\}$ for some $l \in \mathbb{N}$, $U_{\beta} = \{w_0\}$, $V_{\beta} = \{w_l\}$, and 
$\{\{w_{i-1},w_i\}: i \in [l]\} \subseteq E(\beta)$.
\item \textbf{(No Crossings)} No two edges of $\beta$ cross. More precisely, there are no indices $i_1 < i_2 < i_3 < i_4 \in [l] \cup \{0\}$ such that $\{w_{i_1},w_{i_3}\} \in E(\beta)$ and $\{w_{i_2},w_{i_4}\} \in E(\beta)$.
\item \textbf{(No Boundary Edge)} $\{w_0,w_l\} \notin E(\beta)$.
\end{enumerate}
Backbone-correction shapes can be defined in the same way except that we have the edge $\{w_0,w_l\}$. In particular, a backbone-correction shape is a shape $\tau$ such that 
\begin{enumerate}
\item \textbf{(Non-backtracking Path through All Vertices)} $V(\tau) = \{w_j: j \in [l] \cup \{0\}\}$ for some $l \in \mathbb{N}$, $U_{\tau} = \{w_0\}$, $V_{\tau} = \{w_l\}$, and $\{\{w_{i-1},w_i\}: i \in [l]\} \subseteq E(\tau)$.
\item \textbf{(No Crossings)} No two edges of $\tau$ cross. More precisely, there are no indices $i_1 < i_2 < i_3 < i_4 \in [l] \cup \{0\}$ such that $\{w_{i_1},w_{i_3}\} \in E(\tau)$ and $\{w_{i_2},w_{i_4}\} \in E(\tau)$.
\item \textbf{(Boundary Edge)} $\{w_0,w_l\} \in E(\tau)$.
\end{enumerate}
\end{definition}
\begin{remark}
    We consider the shape $\fp_1$ corresponding to the adjacency matrix to be a degenerate case of a backbone-correction shape.
\end{remark}

\begin{figure*}[h]
	\begin{subfigure}[t]{0.5\textwidth} 
    \centering
    \includegraphics[height=3cm]{diagrams/P2p.pdf}
    \caption{$\overline{\fp}_2(A_G)$: Backbone Correction for  ($\fp_2(A_G)$)}
        \label{fig:backbone-p2-correction}
    \caption*{{$\overline{\fp}_2[i,j]= \sum_{t\notin\{i,j\}} A_G[i,t] \cdot A_G[t,j]\cdot A_G[i,j] $. } }\end{subfigure}
\hfill
\begin{subfigure}[t]{0.5\textwidth}
    \centering
    \includegraphics[height=3cm]{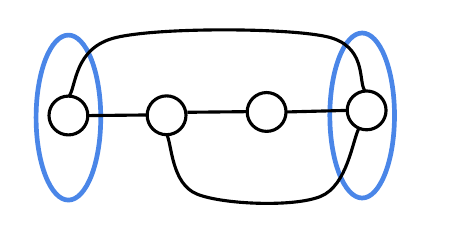}
    \caption{Example of Backbone-Correction}
    \label{fig:eg-backbone}
    \caption*{{$\cm_\tau[i,j]= \sum_{a\neq b\neq i\neq j} A_G[i,a] \cdot A_G[a,b]\cdot  A_G[b,j]\cdot A_G[i,j] \cdot A_G[a,t] $. }}
\end{subfigure}

    \centering
    \begin{subfigure}[t]{0.45\textwidth}
        \centering
        \includegraphics[width=\linewidth]{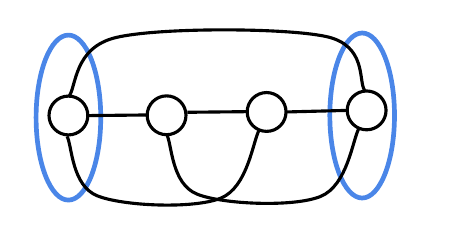}
        \caption{Illegal due to Crossing Edges}
        \label{fig:illegal-1}
    \end{subfigure}
    \hfill
    \begin{subfigure}[t]{0.45\textwidth}
        \centering
        \includegraphics[width=\linewidth]{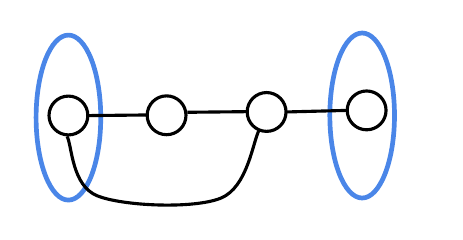}
        \caption{Illegal due to Missing the Boundary Edge}
        \label{fig:illegal-2}
    \end{subfigure}
    \caption{Examples of Legal and Illegal Backbone-Correction Shapes }
    \label{fig:backbone-correction}
\end{figure*}

\begin{figure}
    \centering
    \begin{subfigure}[t]{0.45\textwidth}
        \centering
        \includegraphics[width=\linewidth]{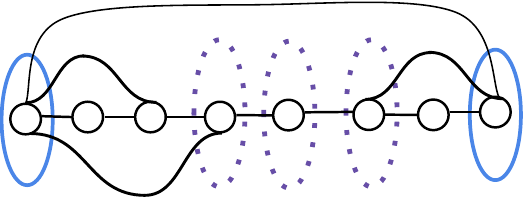}
    \end{subfigure}
        \caption{Backbone-Correction Shape (for \cref{fig:multi-concate}) }
\end{figure}
We record a structural property of the matrices that arise throughout the iterative process, namely that each $Q_i$ and $F(Q_i)$ is symmetric.
\begin{definition}[Symmetric Linear Combination] \label{def:sym-lin-comb}
Let $\calB $ be a finite collection of shapes, and let
\(
M = \sum_{\tau \in \calB} c_\tau \cdot \cm_\tau
\)
be a linear combination of the corresponding graph matrices. We say that $M$ is a \emph{symmetric linear combination} if for all $\tau \in \calB$, we have $\tau^{\top} \in \calB$ and $c_{\tau} = c_{\tau^\top}$.
\end{definition}
As a consequence, any symmetric linear combination $M$ defines a symmetric matrix, even though the individual summands $M_\tau$ need not be symmetric.

\begin{observation} In the above process, for all $i \in \mathbb{N} \cup \{0\}$, $Q_i$ and $F(Q_i)$ are symmetric linear combinations.
\end{observation}
\begin{remark}
In fact, the coefficients $c(\tau)$ for $Q$ satisfy the stronger property that $c(\tau)$ only depends on the underlying graph of $\tau$. In other words, changing the boundary vertices of $\tau$ does not affect $c(\tau)$ (as long as they are still adjacent vertices on the Hamiltonian cycle of $\tau$ so that we still have a backbone-correction shape). To see this, observe that by the nature of the iterative construction, 
\[
c(\tau) = C_F^{|V(\tau)|-1}\prod_{\text{cycles } C \text{ of } \tau \text{ with no chords}}{b_{l(C) - 1}}
\]
\end{remark}
\paragraph{Intuition for the Construction.}
The recursive construction of $Q$ is designed so that the resulting matrix satisfies three key structural properties, which together ensure that $F(Q)$ has the desired spectral behavior.

\emph{(1) Wigner-type behavior of $Q$.}
After normalization, the matrix $Q$ behaves analogously to a Wigner random matrix. Concretely, we show that $Q$ has spectral norm bounded by \[\|Q\| \leq (2+o_n(1)) \cdot \sqrt{\sum_{\tau \in \calB(Q)}{c({\tau})^2}}\]
and for all $k \in \mathbb{N}$, $\E\left[\Tr\left((QQ^T)^k\right)\right] = n \cdot \left(C_k\left(\sum_{\tau \in \calB(Q)}{c({\tau})^2}\right)^k + o_n(1)\right)$ where $C_k = \frac{1}{k+1}\binom{2k}{k}$ is the $k$th Catalan number. This implies that the limiting distribution of the spectrum of $Q$ as $n \to \infty$ is a semicircle with radius $2\sqrt{\sum_{\tau \in \calB(Q)}{c({\tau})^2}}$.

This is further discussed in~\cref{sec:norm-bound-overview}, and culminates in~\cref{cor:linear-combination-norm-bound} for matrix norm bounds that are sharp in the leading constant. The formal analysis appears in \cref{sec:trace-power-block-walks}.

\emph{(2) Equivalence with concatenation.}
For each $j \ge 1$, the matrix $P_j(Q)$ arising from the Chebyshev expansion behaves like a $j$-way concatenation of shapes from $Q$. This is a generalization of the behavior of $\fp_1$ and its powers which we discussed in \cref{sec:graphmatrixmultiplication}. In \cref{sec:graph-mat-cheb}, we discuss why we expect such a property to hold for any linear combination $Q' = \sum_{\tau}{c(\tau)\cmp_{\tau}}$ of backbone-correction shapes such that $\sum_{\tau}{c(\tau)^2} = 1$ (see \cref{eq:equivalence-dream}). We formally prove that this property holds in \cref{sec:cheby-shape-approx}.

The recursive construction ensures that whenever such concatenated shapes appear, their corresponding backbone-correction shapes are introduced at the same stage, maintaining the required structural pairing. 

\emph{(3) Variance normalization.} Our main theorem for norm bounds establishes that the spectral norm of $Q$ is determined by its variance proxy given by the coefficients. We further show that
the coefficients satisfy
\[
\sum_{\tau \in B(Q)} c(\tau)^2 = 1,
\]
so that the overall scale of $Q$ remains controlled throughout the recursion (up to truncations).

To see this, let $S_i = \sum_{\tau \in \calB(Q_i)}{c({\tau})^2}$ and observe that for all $i \in \mathbb{N} \cup \{0\}$,
\begin{align*}
S_{i+1} - S_i &= \sum_{j = 2}^{\infty}{b_j^2\sum_{\beta = \tau_1 \circ \cdots \circ \tau_j \in \textsf{New-}P_j(Q_i)} \prod_{t=1}^j c(\tau_t)^2}\\
&= \sum_{j = 2}^{\infty}{b_j^2(S_i^j - S_{i-1}^j)}
\end{align*}
where we take $S_0 = C_F^2$ and $S_{-1} = 0$. Summing this expression from $i = 0$ to $k$ gives that 
\[
S_{k+1} = S_0 + \sum_{j=2}^{\infty}{{b_j^2}S_k^j}
\]
Letting $S_{\infty} = \lim_{k \to \infty}{S_k}$ (assuming this limit exists), we have that 
$S_{\infty} = C_F^2 + \sum_{j=2}^{\infty}{{b_j^2}S_{\infty}^j}$. Since $\sum_{j=2}^{\infty}{b_j^2} = 1 - C_F^2$, $S_{\infty} = 1$ is a solution of this equation.
This is formally established in~\cref{claim:variance-evolution-recurrence}.

Taken together, these properties imply that $F(Q)$ behaves like a polynomial transformation of a Wigner-type matrix with controlled variance. In particular, this leads to a spectral norm bound of the correct order, up to lower-order error terms. The remainder of the analysis is devoted to formalizing these properties and making the above heuristic precise. 
\subsection{Graph Matrices Meet Chebyshev Polynomials}\label{sec:graph-mat-cheb}
In this section, we elaborate on the connection between the polynomials $P_j(x) = U_j\!\left(\frac{x}{2}\right)$ (which are a rescaling of the Chebyshev polynomials of the second kind) and non-backtracking walks which we previously alluded to in \cref{sec:graphmatrixmultiplication} and \cref{sec:key-ideas}. In particular, let
$
M = \sum_{\tau \in \calB(M)} c(\tau)\,\cmp_\tau
$
be a linear combination of backbone-correction matrices such that $\sum_{\tau \in \calB(M)}{c(\tau)^2} = 1$. For each $j$, let $P_j(M)$ denote the matrix obtained by applying the polynomial $P_j(x) = U_j\!\left(\frac{x}{2}\right)$ to $M$. Our goal is to show
\begin{displayquote}
$P_j(M)$ admits a \emph{simple} graph-matrix representation using shapes derived from $\calB(M)$.
\end{displayquote}
In particular,  it corresponds to a collection of shapes that can be viewed as $j$-way concatenations using shapes from the base set $\calB(M)$ (see \cref{def: j-way-product}). 

\paragraph{Backtracking Intersection Patterns} We first define backtracking intersection patterns, which are a generalization of the intersection pattern for $\cmp_{\fp_j}\cmp_{\fp_1}$ which backtracks and gives $\cmp_{\fp_{j-1}}$.
\begin{definition}[Backtracking intersections]\label{def:backtracking-intersections}
Given a shape $\alpha = \tau_1 \circ \ldots \circ \tau_j$ which is the concatenation of $j \geq 1$ backbone-correction shapes and a backbone-correction shape $\tau_{j+1}$, we say that $I \in Int_{\alpha,\tau_{j+1}}$ is a backtracking intersection pattern if the following conditions hold:
\begin{enumerate}
\item $\tau_{j+1} = \tau_{j}^{\top}$;
\item For all $v \in V(\tau_j) \setminus v_{\tau_j}$, letting $v'$ be the copy of $v$ in $\tau_{j+1}$, $v \sim v'$. Note that this implies that $v_{\tau_{j-1}} = u_{\tau_{j}} \sim v_{\tau_{j+1}}$.
\item For all vertices $u \in V(\alpha) \setminus V(\tau_j)$, there are no other vertices of $V(\alpha \circ \tau_{j+1})$ in the same equivalence class as $u$.
\end{enumerate}
Note that for each $\alpha = \tau_1 \circ \ldots \circ \tau_j$, there is a unique backbone-correction shape $\tau_{j+1}$ and a unique intersection pattern $I \in Int_{\alpha,\tau_{j+1}}$ such that $I$ is a backtracking intersection pattern.
\end{definition}
\begin{proposition}
Given a shape $\alpha = \tau_1 \circ \ldots \circ \tau_j$ which is the concatenation of $j \geq 1$ backbone-correction shapes, letting $I_{back} \in Int_{\alpha,\tau_{j}^{\top}}$ be the backtracking intersection pattern for $\alpha$, 
\[
\cm_{\gamma_{I_{back}}} = \left(\prod_{i=1}^{|V(\tau_j)| - 1}{(n - |V(\alpha) \setminus V(\tau_j)| - i)}\right)\cm_{\tau_1 \circ \ldots \circ \tau_{j-1}}
\]
Note that this implies that the intersection term for $\cmp_{\alpha}\cmp_{\tau_j^{\top}}$ resulting from $I_{back}$ is approximately $\cmp_{\tau_1 \circ \ldots \circ \tau_{j-1}}$.
\end{proposition}

 \begin{figure}[h]
     \vspace{0.5em}
    
     \begin{subfigure}[t]{0.5\textwidth}
         \centering
         \includegraphics[width=\linewidth]{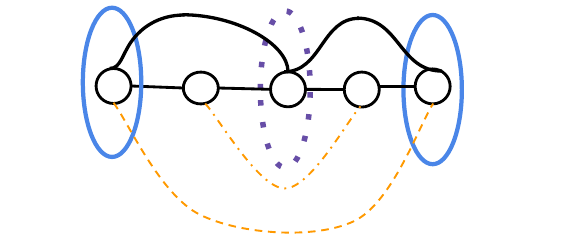}
      \end{subfigure}
     \hfill
 \begin{subfigure}[t]{0.48\textwidth}
         \centering
         \includegraphics[width=\linewidth]{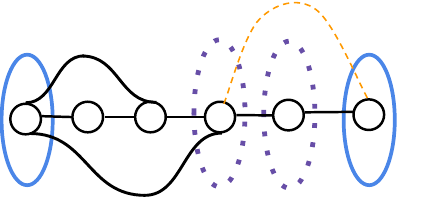}
     \end{subfigure}
     \hfill
              \caption{Examples of Backtracking Intersections}
       \label{fig:example-intersection}
       \label{fig:example-backtracking-intersection}
     \end{figure}
 We illustrate this with two examples of backtracking intersections in \cref{fig:example-backtracking-intersection}. In contrast, the reader should note that the intersection in \cref{fig:eg-intersection} is \emph{not} a backtracking intersection.

\paragraph{Products of backbone-correction matrices}
We now describe the approximate behavior of products of graph matrices. In \cref{sec:cheby-shape-approx}, we show that when we consider a product of the form $\cmp_{\tau_1 \circ \ldots \circ \tau_j}\cmp_{\tau_{j+1}}$, the main intersection patterns are the null intersection pattern with no intersections and the backtracking intersection pattern (if it exists). This leads to the following claim, which is the key claim we need for our analysis.
\begin{claim}
For all $j \in \mathbb{N}$ and all backbone-correction shapes $\tau_1,\ldots,\tau_{j+1}$, for the purposes of our analysis, $\cmp_{\tau_1 \circ \ldots \circ \tau_{j}}\cmp_{\tau_{j+1}} \approx \cmp_{\tau_1 \circ \ldots \circ \tau_{j+1}} + 1_{\tau_{j+1} = \tau_{j}^{\top}}\cmp_{\tau_1 \circ \ldots \circ \tau_{j-1}}$.
More precisely, 
\begin{align*}
\cmp_{\tau_1 \circ \ldots \circ \tau_{j}}\cmp_{\tau_{j+1}} \approx \cmp_{\tau_1 \circ \ldots \circ \tau_{j+1}} &= \cmp_{\tau_1 \circ \ldots \circ \tau_{j+1}} + 1_{\tau_{j+1} = \tau_{j}^{\top}}\cmp_{\tau_1 \circ \ldots \circ \tau_{j-1}} + (\text{terms with norm } o_n(1)) \\
&+ (\text{terms whose correction has norm } o_n(1))
\end{align*}
and when we multiply terms which have small norm or whose corrections have small norm by another backbone-correction shape, the resulting terms still either have small norm or a correction which has small norm.
\end{claim}

\paragraph{Correspondence between Chebyshev polynomials and shape concatenation}
We are now ready to describe the effect of applying the polynomial $P_j(x) = U_j(\frac{x}{2})$ to a linear combination of backbone-correction shapes. We begin by defining the collection of concatenated shapes that will ultimately appear in the resulting representation. For notational convenience, we will use $\calP_j(\calB)$ to refer to the collection of properly concatenated shapes arising from the $j$-way concatenations from a base set $\calB$ of backbone-correction shapes.

\begin{definition}[$j$-Way Concatenations from Base Shapes] \label{def: j-way-product}
Let $
M = \sum_{\tau \in \calB} c(\tau) \cdot \cm_\tau
$
be a linear combination of shapes from a base set 
\(\calB\) of backbone-correction shapes.
We define the set of its \emph{$j$-way concatenations}, denoted by 
\(\calP_j(M)\), to be the collection of shapes satisfying:

\begin{enumerate}
    \item Each shape $\al$ in \(\calP_j(\calB)\) is obtained by properly concatenating \(j\) base shapes, i.e., 
    \[
    \alpha = \tau_1 \circ \tau_2 \circ \cdots \circ \tau_j,
    \]
    where \(\tau_i \in \calB\) for every \(i \in [j]\).

    \item Each concatenated shape 
    \(\alpha = \tau_1 \circ \cdots \circ \tau_j\) 
    is assigned coefficient
    \(
    c(\al) = \prod_{t=1}^{j} c(\tau_t).
    \)
    In other words, the coefficient of a concatenated shape equals the product of the coefficients of its constituent base shapes in \(M\).
\end{enumerate}
We define the matrix $\mathsf{P}_j(M)$ corresponding to concatenating $j$ copies of $M$ to be 
\[
\mathsf{P}_j(M) 
= \sum_{\al\in\calP_j(M)}c(\al) \cdot \cmp_{\al}\,.
\]
Note that $\mathsf{P}_j(M)$ is not the same as the matrix $P_j(M)$ obtained by applying the polynomial $P_j(x) = U_j(\frac{x}{2})$ to $M$. That said, when $M$ is a properly normalized linear combination of backbone-correction shapes, $P_j(M) \approx \mathsf{P}_j(M)$ and this is the key result of our analysis.
\end{definition}
\begin{lemma}[(Informal version of \cref{thm:formal-equiv-cheby-shape})] \label{eq:equivalence-dream}
If $M = \sum_{\tau \in \calB}{c(\tau)\cmp_{\tau}}$ is a linear combination of backbone-correction matrices such that $\sum_{\tau \in \calB}{c(\tau)^2} = 1$ then $P_j(M) \approx \fp_j(M)$. More precisely, 
\[
P_j(M) = \fp_j(M) +(\text{terms with norm } o_n(1)) 
+ (\text{terms whose correction has norm } o_n(1))\,.
\]
\end{lemma}
\begin{proof}[Proof sketch]
We prove this by induction. The base cases $j = 0$ and $j = 1$ are trivial. For the inductive step, we use the recurrence relation $U_{j+1}(x) = 2xU_{j}(x) - U_{j-1}(x)$ for Chebyshev polynomials which implies that $P_{j+1}(x) = xP_{j}(x) - P_{j-1}(x)$.

Assume that for all $j' \leq j$, $P_j(M) \approx \fp_j(M)$ and consider $P_j(M)M$. Observe that 
\begin{align*}
\fp_j(M)M &= \left(\sum_{\alpha' \in \calP_{j-1}(M)}{\sum_{\tau_j \in \calB}{c(\alpha')c(\tau_j)\cmp_{\alpha' \circ \tau_j}}}\right)\left(\sum_{\tau_{j+1} \in \calB}{c(\tau_{j+1})\cmp_{\tau_{j+1}}}\right)\\
&\approx \sum_{\alpha' \in \calP_{j-1}(M)}{\sum_{\tau_j \in \calB}{\sum_{\tau_{j+1} \in \calB}{c(\alpha')c(\tau_j)c(\tau_{j+1})\left(\cmp_{\alpha' \circ \tau_j \circ \tau_{j+1}} + 1_{\tau_{j+1} = \tau_j^{\top}}\cmp_{\alpha'}\right)}}} \\
&= \fp_{j+1}(M) + \sum_{\alpha' \in \calP_{j-1}(M)}{\sum_{\tau_j \in \calB}{c(\alpha')c(\tau_j)^{2}\cmp_{\alpha'}}} = \fp_{j+1}(M) + \fp_{j-1}(M)
\end{align*}
Thus, $P_{j+1}(M) = P_{j}(M)M - P_{j-1}(M) \approx \fp_{j}(M)M - \fp_{j-1}(M) \approx \fp_{j+1}(M)$, as needed.
\end{proof}

\input{norm_bound_preview.tex}

\subsection{Putting Everything Together}\label{sec:real-world-proof}
In this section, we describe how to put everything together to prove our main result. 
\paragraph{Summary of the Argument in an Ideal World}
We start by summarizing how our argument would work if we could take the limit as the number of iterations goes to infinity when constructing $Q$ and could ignore smaller order terms.

\begin{enumerate}
\item As described in \cref{sec:key-ideas}, we construct $Q = \sum_{\tau}{c(\tau)\cmp_{\tau}}$ so that $Var(Q) := \sum_{\tau}{c(\tau)^2} = 1$.
\item By \cref{eq:equivalence-dream}, since $Var(Q) = 1$, for all $j \in \mathbb{N}$, $P_j(Q) \approx \fp_j(Q)$ where $\fp_j(Q)$ is the sum of possible concatenations of $j$ backbone-correction shapes of $Q$.
\item Letting $F(x) = C_F + x + \sum_{j=2}^{\infty}{b_jP_j(x)}$, we have that for all $x \in [-2,2]$, $F(x) = 1_{x \geq 0}(2x)$. Since we start our iteration procedure with $Q_0 = {C_F}\cmp_{\fp_1} = \frac{C_F}{\sqrt{n}}A_G$, the first two terms of $F(Q)$ are ${C_F}Id + {C_F}\cmp_{\fp_1}$. As described in \cref{sec:key-ideas}, due to the iterative procedure for constructing $Q$, $F(Q) - {C_F}Id - {C_F}\cmp_{\fp_1}$ satisfies the non-edge constraints as it contains its own correction terms.
\item We can take our target matrix to be $W = \frac{F(Q)}{C_F}$ (see \cref{def:concrete-target}). Since $||Q|| \approx 2$ and $F(x) \geq 0$ for all $x \in [-2,2]$, $W \succeq 0$. This gives an upper bound of $(1+o_n(1))\sqrt{n}$ on $\vartheta(G)$.
\end{enumerate}
However, as hinted by the title of this section, there are three issues that we have so far swept under the rug: \begin{enumerate}
	\item Ideally we want to have a finite construction, and therefore we would like to terminate the recursive process at some point.
	\item Our matrix analysis requires a bound on the sizes of the shapes involved so we would like to truncate the sizes of the shapes we consider as well.
	\item We need to handle $F(\pm (2+\delta) )$ for some tiny $\delta$ due to the deviation at the edge of the spectrum.
\end{enumerate}
Intuitively, each item should not pose a real concern as we expect certain decay phenomena that allow us to focus on the "low-degree" terms of the Chebyshev polynomial expansion - consistent with our knowledge that for average case SoS lower bounds, truncating the pseudo-expectation values given by pseudo-calibration generally works well.

\paragraph{Real-World Analysis}
We now describe how we address these concerns. First, instead of starting at $Q_0 = C_F \cmp_{\fp_1}$, we may apply the recursive process to \[ 
\widetilde{Q}_0 = (1+c_\gamma) \cdot C_F \cdot  \cmp_{\fp_1}
\]
where we will show $c_\gam$ measures the vanishing slack we have to the optimal value. Let the inner matrices built across the iterative process be $\widetilde{Q}_i$.

Next, we work with the following parameters for truncations to be chosen. 

\begin{definition}[Parameter $t^*$: Termination of the Iterative Process]
    For $t^*$ to be picked, we will let it denote the termination level for the iterative process. In other words, we will construct a sequence of inner matrices $\widetilde{Q}_1, \widetilde{Q}_2,\dots  $ up to $\widetilde{Q}_{t^*}$. 
\end{definition}

\begin{definition}[Parameter $D$: Inner Truncation for Chebyshev Polynomials] \label{def:inner-trunc}
	We define $D$ to be the truncation parameter for $P(\widetilde{Q}_i)$ at each level of $\widetilde{Q}_i$. In other words, at each level of $\widetilde{Q}_i$, we only consider $P_{j}(\widetilde{Q}_i)$ for $j\leq D$.
\end{definition}

 Let the final matrix be $\widetilde{Q}  \coloneqq \widetilde{Q}(t^*, D, c_\gamma)   $. Since we stop the process early and truncate $P_{j\leq D_j}$ at each level of $Q_i$,  the variance of the resulting matrix is not exactly $1$ but $1-o_{D,t^*}(1) $ by the following lemma. Crucially, the variance continues to go to $1$ by taking sufficiently large truncation thresholds $t^*$ and $D$, and this allows us to pick $c_\gam = c_\gam(t^*,D)$ such that the variance is $1$ for the inner matrix.

\begin{lemma}[Truncated Variance is still $1$ (Informal version of \cref{lem:truncated-var})]
There is a choice of $c_\gam$ such that the inner matrix has variance $1$; moreover, $c_\gam =o_{t^*, D}(1) $.
\end{lemma}

One important step of our work is to show the inner matrix has spectrum $[-2,2]$ (up to $o_n(1)$  deviation). Towards that end, we combine our meta norm bound from \cref{cor:linear-combination-norm-bound} and our variance calculation to obtain the following spectral estimate for the inner matrix:
\begin{restatable}{lemma}{finalnorm}
\label{lem:final-norm} (Informal version of \cref{cor:linear-combination-norm-bound})
For the inner matrix $\widetilde{Q}$, with high probability,
\[
\|\widetilde{Q} \|
\le (1+o_n(1))\cdot 2.
\]
\end{restatable}

Next, we apply outer truncation to $F(\widetilde{Q})$: take $F_D$ to be the degree-$D$ truncation of $F$, with the truncation threshold $D$ chosen according to \cref{def:inner-trunc}.

\begin{lemma}[Approximate PSDness of $F_D(\widetilde{Q})$ ]
    Suppose $\mathsf{spec}(\widetilde{Q}) \subseteq [-2-\delta,\,2+\delta]$. Then
\[
F_{D}(\widetilde{Q})
\succeq
- \tilde{O}\!\left(\frac{1}{D} + \delta \cdot \poly(D)\cdot \exp(D\sqrt{\delta})\right)\cdot I.
\]
(See also \cref{prop:FD-psd}.)
\end{lemma}
We now formally describe the construction of our witness matrix. For any parameter $\eps>0$, we define the following matrix
\[
M_{\mathrm{final}} = M_{\mathrm{final}}(\eps)
\coloneqq \eps \cdot I
+ F_D(\widetilde{Q})
+ \chebyerror + \mathsf{EarlyTerm}
\]
 with final correction terms
\begin{enumerate}
	\item \[ \chebyerror(\widetilde{Q}) \coloneqq \sum_{j=1}^D b_j\,\mathrm{CorrectionError}(j) \] uses the corrections from \cref{thm:formal-equiv-cheby-shape} for the Chebyshev expansion of $F_D$;
	\item $\mathsf{EarlyTerm}(\tilde{Q})$, whose norm is bounded in \cref{lem:corr-small}, corrects the errors from terminating at $Q_{t^*}$.
\end{enumerate}
We suppress the dependence of $M_{\mathrm{final}}$ on $\eps$, while we will complete our final construction by picking $\eps= \eps(D,t^*)$ in the following discussion. Before that, we verify the above matrix satisfies the non-edge constraints by design of the final correction terms.
\begin{restatable}[Verification of (Non-)Edge Constraints]{lemma}{EdgeConstraintsVerification}
(See also \cref{lem:edge-constraints-verification}) $M_{\mathrm{final}}$ satisfies the non-edge constraints in \cref{def:concrete-target} after normalizing its diagonal.
\end{restatable}

\paragraph{Bounds on Final Correction Terms}
Next, we establish an error bound for the nontrivial intersection terms and correction terms for edge constraints due to terminating at $Q_{t^*}$.
\begin{lemma}[Correction Bounds for Chebyshev Approximation via Graph Matrices (Informal version of \cref{thm:formal-equiv-cheby-shape})] We have
	\[ \|\chebyerror(\widetilde{Q})\|   = o_n(1) \]
\end{lemma}

\begin{lemma}[Informal version of \cref{lem:corr-small}]
    The correction terms for early termination at level $t^*$ have small norm, i.e., \[
\left\|
\mathsf{EarlyTerm}(\widetilde{Q})
\right\|
=  o_{t^*}(1)\,.\]
\end{lemma}

\paragraph{Wrapping Up the Proof}
We now show PSDness of the final matrix and verify its objective value. For given $D$ and $t^*$, we would like to pick a small $\epsilon>0$  to ensure the PSDness of the final matrix.

\begin{restatable}[Final Parameter Choice for $\eps_{\mathrm{spec}}$]{claim}{FinalParamChoiceSpec}\label{claim:final-parameter-choice-for-eps-spec}
Let $t^*,D \leq O\!\left(\frac{\log \log n}{\log \log \log n}\right)$. 
Then it suffices to take
\(
\eps_{\mathrm{spec}} = \tilde{\Theta}\!\left(\frac{1}{D}\right)
\) 
such that
\[
M_{\mathrm{final}}
\coloneqq \eps_{\mathrm{spec}} I
+ F_D(\widetilde{Q})
+ \chebyerror + \mathsf{EarlyTerm}
\;\succeq\; 0 \, .
\]
We write \(\tilde{\Theta}(\cdot)\) to suppress polylogarithmic factors in \(D\) (but not in \(n\)).
\end{restatable}
We verify this calculation in~\cref{sec:final-parameter}. This follows by combining the above results and bounding the negative spectrum of $F_D$ as a function of $D$ and $t^*$, as well as the spectral norm of the final correction term needed. Our final parameter choice sets $D=t^*$ for simplicity. Importantly, the above claim allows us to choose the truncation parameter $D$ to grow slowly with $n$, which in turn enables us to achieve a value of $(1+o_n(1))\sqrt{n}$.

By \cref{lem:edge-constraints-verification}, the corrected sum $F_D(\widetilde{Q})+\chebyerror+\mathsf{EarlyTerm}$ has diagonal $C_F$. Adding $\eps_{\mathrm{spec}}I$ gives $M_{\mathrm{final}}$ diagonal $C_F+\eps_{\mathrm{spec}}$. We therefore normalize by setting
\[
\widetilde{M}_{\mathrm{final}} \coloneqq \frac{M_{\mathrm{final}}}{C_F+\eps_{\mathrm{spec}}}\,.
\]

\begin{proposition}[Verification of Objective Value]\label{prop:value-verification}
With the final choice $t^*=D$, the matrix $\widetilde{M}_{\mathrm{final}}$ is a target matrix in \cref{def:concrete-target} with value
\[
c_k=\frac{(1+c_\gamma)C_F}{C_F+\eps_{\mathrm{spec}}}
=1-\tilde{O}\!\left(\frac{1}{D}\right).
\]
\end{proposition}
\begin{proof}
By the initialization $\widetilde{Q}_0=(1+c_\gamma)C_F\cmp_{\fp_1}$ and \cref{lem:edge-constraints-verification},
\[
F_D(\widetilde{Q})+\chebyerror+\mathsf{EarlyTerm}
=C_FI+\frac{(1+c_\gamma)C_F}{\sqrt n}A_G+R,
\]
where $R$ is symmetric and vanishes on the diagonal and nonedges. Adding $\eps_{\mathrm{spec}}I$ and dividing by $C_F+\eps_{\mathrm{spec}}$ gives
\[
\widetilde{M}_{\mathrm{final}}
=I+\frac{c_k}{\sqrt n}A_G+\frac{R}{C_F+\eps_{\mathrm{spec}}}\,,\]
for
\(
c_k=\frac{(1+c_\gamma)C_F}{C_F+\eps_{\mathrm{spec}}}.
\) For $t^*=D$, \cref{lem:truncated-var} gives $c_\gamma=O(D^{-3})$. With $\eps_{\mathrm{spec}}=\tilde{\Theta}(1/D)$ as chosen above, for all sufficiently large $D$,
\[
1-c_k
=\frac{\eps_{\mathrm{spec}}-C_Fc_\gamma}{C_F+\eps_{\mathrm{spec}}}
=\tilde{O}(1/D).
\]
Positive semidefiniteness follows from \cref{claim:final-parameter-choice-for-eps-spec} and $C_F+\eps_{\mathrm{spec}}>0$.
\end{proof}

The main theorem for the \Lovasz-Theta function then follows as a corollary of the above.
\MainTheta*
\begin{proof}
    Take $t^*=D = C_1 \cdot \frac{\log \log n}{\log \log \log n}$ for some small constant $C_1>0$, and $\eps_{\mathrm{spec}} = C_2\cdot \poly\log(D)/D$.

\paragraph{Objective Value}
By \cref{prop:value-verification}, the final matrix $\widetilde{M}_{\mathrm{final}}$ is a target matrix with value
\[
c_k=1-\tilde{O}\!\left(\frac{1}{D}\right).
\]
Using the conversion following \cref{def:concrete-target}, this yields a solution for the dual SDP in~\cref{def:dual-sdp} of value
\[
\lambda=1+\frac{\sqrt n}{c_k}
=1+
\left(1 + \tilde{O}\left(\frac{1}{D}\right)\right)\sqrt{n} = (1+o_n(1)) \cdot  \sqrt{n}\,,
\]
where we plug in our value of $D = \omega(1)$ s.t. $\tilde{O}(1/D)=o_n(1)$.

\paragraph{PSDness and Non-Edge Constraints}
PSDness follows from \cref{claim:final-parameter-choice-for-eps-spec} together with $C_F > 0$. 
The (non-)edge constraints are verified in \cref{lem:edge-constraints-verification} and are preserved under rescaling.  

Finally, the upper bound for the primal SDP (\cref{def:primal-sdp}) follows by duality. The lower bound follows by the folklore result.
\end{proof}

\begin{remark}
    Our result also gives an explicit matrix with a value up to $o_n(1)$ slack for lower bounding the primal SDP.
\end{remark}

%% file: norm_bound_preview.tex
\subsection{Norm Bounds for Linear Combinations of Backbone-Correction Shapes}
\label{sec:norm-bound-overview}
For our analysis, it is crucial to obtain precise bounds on the norms of the matrices which appear, especially
the inner matrix \(Q\). We now describe the techniques we use to obtain these tight norm bounds. As a byproduct, the same analysis also identifies the limiting spectrum of \(Q\), revealing a free-independence phenomenon among its underlying graph-matrix
constituents, i.e., backbone-correction shapes.

\paragraph{Trace Moment Method and Factor Assignment Scheme}
To bound the norms of our matrices, we use the trace moment method. For any matrix $M$, 
\[
    \|M\|
    \le
    \Tr \!\bigl( (MM^\top)^q\bigr)^{1/2q},
\]
so to bound the norm of $M$, it is sufficient to obtain bounds on the trace moments \(\Tr((MM^\top)^q)\).

For fixed natural numbers $q$, we determine the leading-order term of \(\Tr((MM^\top)^q)\). This determines the limiting distribution of the singular values of $M$ but is not sufficient to give good norm bounds on $M$, as $\Tr \!\bigl( (MM^\top)^q\bigr)^{1/2q}$ will be roughly $\sqrt[2q]{n}$ times as large as $\|M\|$. To prove our norm bounds on $M$, we consider values of $q$ much larger than $\log(n)$. For this range of $q$, our bounds on \(\Tr((MM^\top)^q)\) are not quite as accurate but are sufficient to determine the norm of $M$ up to a factor of $(1 \pm o_n(1))$.

When the trace method is applied to graph matrices, each contribution to
\(\Tr((M_{\tau}M_{\tau}^{\top})^q)\) has a concrete combinatorial interpretation:
it corresponds to a labeling of the vertices in \(q\) copies each of
\(\tau\) and \(\tau^\top\), arranged alternately, by elements of \([n]\), namely vertices of the
underlying input graph \(G\), subject to boundary consistency. In other words, we can view it as a walk of length $2q$ with each step being a \emph{block-step}: a collection of edges prescribed by the underlying shape as opposed to an ordinary edge as in the graph case. See the following figure for an example of part of a trace walk for the shape $\bar{\fp}_2$.

\begin{figure}[h]
    \centering
    \begin{subfigure}[t]{0.7\textwidth}
        \centering
        \includegraphics[width=\linewidth]{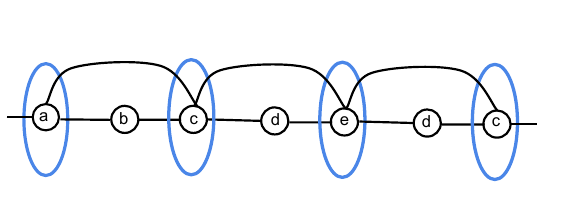}
        \end{subfigure}
        \caption{Snippet of a Trace Walk for $\bar{\fp}_2$}
\end{figure}
Next, we consider the concrete contribution of a walk to the trace.
\begin{proposition} For any $q\in \N$, 
	\[ 
	\E[\Tr(\cm_{\tau }\cdot   \cm_{\tau }^\top )^{q}]  \leq \sum_{P:\text{trace-walk of length }2q} \val(P)\,,
	\]
    where we define \[ 
    \val(P) \coloneqq \E_G \left[ \prod_{i\in [q]} \cm_{\tau} [U_{\tau_i}, V_{\tau_{i+1}}] \cdot \cm_{\tau} [U_{\tau^T_i}, V_{\tau^T_{i+1}}] \right] = \prod_{e = (I,J)\in E(P)} \E_G[(G_e)^{\mul_P(e)} ]
    \]
    where $E(P)$ denotes the set of edges (random variables) used by any step in $P$, and $\mul_P(e)$ denotes the number of times an edge $e$ appears in $P$.
\end{proposition}
An immediate observation---central to many trace method arguments for random matrices---is that only even walks contribute. In particular, any walk that traverses an edge (i.e., uses an underlying random variable) exactly once has zero contribution in expectation. Thus, it suffices to restrict attention to walks in which every edge appears an \emph{even} number of times.

\paragraph{Identifying the Dominant Walks}
To obtain our norm bound, we first isolate the walks that give the dominant
contribution to the trace moment. At a high level, these walks generalize the
tree-like walks that arise in the trace-moment calculation for Wigner matrices.
The key difference is that each ordinary edge-step is replaced by a block-step
of shape \(\tau\), namely a collection of edges corresponding to the graph
matrix under consideration. 
The guiding intuition is that the internal vertices of a backbone-correction
shape---that is, the vertices in
\(
    V(\tau)\setminus (u_\tau \cup v_\tau)
\)
---can be effectively ignored: their combinatorial contribution is exactly
balanced by the normalization factors associated with the graph-matrix entry.

With this simplification, the walk across block-steps simply
reduces to an ordinary walk across edges (as opposed to a collection of edges at each step), bringing us back to
the familiar Wigner-type picture.  The example below illustrates the intuitive simplification we would like to make.

\begin{figure}[h]
    \centering
    \begin{subfigure}[t]{0.48\textwidth}
        \centering
        \includegraphics[width=\linewidth]{diagrams/shapewalk.pdf}
    \end{subfigure}
    \hfill
    \begin{subfigure}[t]{0.48\textwidth}
        \centering
        \includegraphics[width=\linewidth]{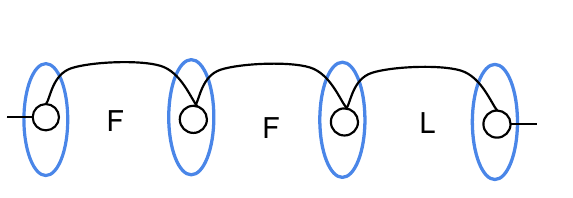}
    \end{subfigure}
    \caption{Simplified View}
\end{figure}

We next explain the \(F/L\) labels in the diagram above. For the purposes of
this overview, we give an informal definition that suppresses several technical
details needed in the full spectral-norm argument.  

\begin{definition}[\(F/L\) Walks (Informal)]
An \(F/L\) walk is a walk that starts from a root vertex \(s\) and consists of two types of steps, called \(F\)-steps and \(L\)-steps:
\begin{enumerate}
\item An \(F\)-step starts at the current vertex \(u\) and goes to a new vertex. This is analogous to going away from the root in a tree walk.
\item An \(L\)-step starts at a vertex \(v\) and returns to its parent. This is analogous to going towards the root in a tree walk.

\end{enumerate}
We consider an $F/L$ walk with $j$ steps to be specified by a sequence of $j$ letters, each of which is $F$ or $L$, as this determines the walk up to renaming the vertices.
\end{definition}

At a high level, we show that the $F/L$ walks are the only dominant terms. 
Concretely, in the normalization used for \(\cmp_\tau\), the graph matrix has
variance one, and we show that
\[
    \E\!\left[
        \Tr\!\left( \cmp_\tau \cmp_\tau^\top \right)^q
    \right]
    =
    (1+o_n(1))\, n\, \cdot C_q + o(n)
\]
for every fixed \(q=O(1)\), where \(C_q\) is the \(q\)-th Catalan number.

For readers familiar with the standard trace-moment calculation for Wigner
matrices and the semicircle law, the dominant walks identified above should look
familiar: they are precisely the same tree-like walks that give the leading
contribution in the Wigner case. Thus, this also reveals that each backbone-correction
shape exhibits a semicircular limiting spectrum, with variance determined by
the corresponding graph-matrix variance.

As discussed above, for our norm bounds, we need to bound $\E\!\left[\Tr\!\left( \cmp_\tau \cmp_\tau^\top \right)^q\right]$ when $q$ is much larger than $log(n)$ which requires more technical details. Before discussing these details, we describe why backbone-correction shapes are freely independent.

\paragraph{Free Independence from Non-Crossing Matchings of Shapes}

First, we consider generalizing trace walks to the setting where $M$ is a linear combination of backbone-correction shapes. See the figure above for an example and note that the only distinction from a single shape is that each step in the trace walk may use a different underlying shape.
\begin{figure}[h]
    \centering
    \begin{subfigure}[t]{0.6\textwidth}
        \centering
        \includegraphics[width=\linewidth]{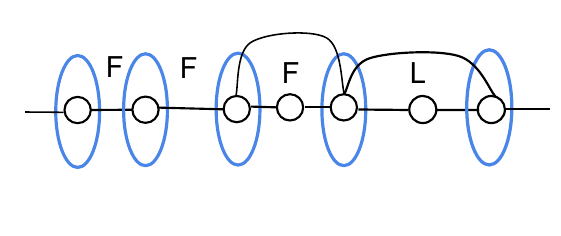}
        \end{subfigure}
        \caption{Trace Walk for a Linear Combination }  \end{figure}\label{fig:linearcombinationtracewalkexample}

In particular, we show that the intuition from the single shape applies to the linear combination as well:
\begin{displayquote}
Each step in the dominant walk is an $F/L$ step as prescribed above, so the walk can be viewed as a walk on a tree.
\end{displayquote}
Moreover, in the dominant walk, the underlying shapes used by the steps can be paired in a non-crossing manner, analogously to the pairing of $F/L$ steps in the trace calculation for a single shape. We illustrate such a matching in the following diagram.

\begin{figure}[h]
    \centering
    \begin{subfigure}[t]{0.6\textwidth}
        \centering
        \includegraphics[width=\linewidth]{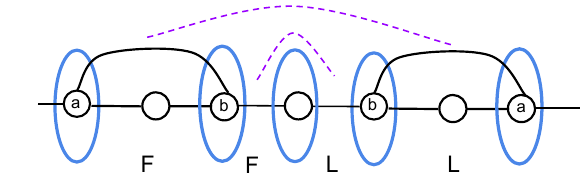}
        \end{subfigure}
        \caption{Matching of Block-Steps}
        \label{fig:match}
        \end{figure}

In particular, we prove that steps can be matched in a strong sense:
\begin{enumerate}
	\item \textbf{Non-crossing Matching of Shapes of Underlying Steps}: For an $F/L$ walk, each $L$ step is matched with the most recent $F$ step that has not yet been matched and this gives a prefect non-crossing matching between the $F$ and $L$ steps. Whenever an $F$ step is matched with an $L$ step, if the shape for the $F$ step is $\tau$ then the shape for the $L$ step must be $\tau^T$. 
	\item \textbf{Matching of Steps}: Next, as our argument for a single shape reveals, whenever an $F$ step with shape $\tau$ is matched with an $L$ step with shape $\tau^T$, the vertices of $\tau$ must be intersected with their mirror images in $\tau^T$ (i.e., each vertex of $\tau$ should be in the same equivalence class as its mirror image in $\tau^T$). Note that this implies that these steps will correspond to the same set of edges of $G$ so each edge of $G$ will appear an even number of times.
\end{enumerate}
This culminates in the following spectrum estimate for a linear combination of backbone-correction shapes. For simplicity of presentation, we state the result
in the variance-one normalization.
\begin{proposition}[Informal version of \cref{thm:shape-sequence-moments}] For a linear combination of backbone-correction shapes in $\calB(M)$,\[
M = \sum_{\tau \in \calB(M)} c(\tau) \cdot \cmp_\tau 
\]
we have
	\[
    \E\!\left[
        \Tr\!\left( M \cdot M^\top \right)^q
    \right]
    =
    (1+o_n(1))\, n\, \cdot C_q + o(n)
\]
for every fixed \(q=O(1)\), where \(C_q\) is the \(q\)-th Catalan number, provided \[ 
\Var(M) \coloneqq \sum_{\tau \in \calB(M) } c(\tau)^2 = 1\,.
\]
\end{proposition}
This implies that if $M$ is symmetric then the limiting distribution of the eigenvalues of $M$ as $n \to \infty$ is the semicircle $SC_2(x) = 1_{x \in [-2,2]}\frac{\sqrt{4-x^2}}{2\pi}$ (even if $M$ is not symmetric, we still have that the limiting distribution of the singular values of $M$ as $n \to \infty$ is the right half of the semicircle $f(x) = 1_{x \in [0,2]}\frac{\sqrt{4-x^2}}{\pi}$).  

\paragraph{Analyzing trace walks for larger $q$} 
Finally, we note that our discussion so far applies for larger $q$ as well. However, there is one additional challenge that is worth highlighting. To illustrate this challenge, it suffices to focus on the analysis for a single backbone-correction shape. The subtlety is that our guiding intuition
treats the internal vertices of such a shape as negligible, effectively replacing
a block-step by an ordinary edge-step:
\[
    \text{block-step of a backbone-correction shape \(\tau\) from \(a\) to \(b\)}
    \quad \approx \quad
    \text{edge-step from \(a\) to \(b\)} .
\]
 \begin{figure}[h]
    \centering
    \begin{subfigure}[t]{0.6\textwidth}
        \centering
        \includegraphics[width=0.8\linewidth]{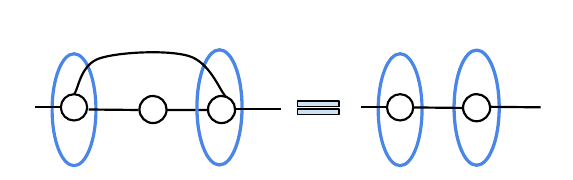}
        \end{subfigure}
        \caption{(Desired) Simplified View}
        \end{figure}    
For larger $q$, this is not the full picture. The reason is that for constant $q$, having a single extra factor of $\frac{1}{\sqrt{n}}$ is sufficient to ensure that the term is not dominant. for larger $q$, we instead need to show that each block which is badly behaved gives a factor of $\frac{1}{poly(n)}$ in order to account for the fact that there will be many such terms. When we carry out this analysis, it turns out that while we can guarantee that block steps which are matched with each other must still have the same underlying edge set, we can no longer guarantee that they have the same boundary vertices! Thus, we need to handle the following question:
\begin{displayquote}
        Given an edge set $E\subseteq \binom{n}{2}$ and a starting vertex $u$, how many possibilities are there for the next boundary vertex $v$?
\end{displayquote}
 \begin{figure}[h]
    \centering
           \begin{subfigure}[t]{0.45\textwidth}
        \centering
        \includegraphics[width=0.8\linewidth]{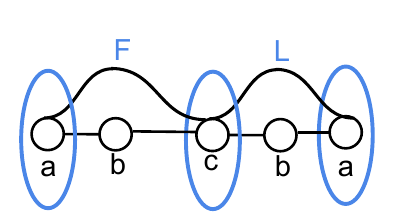}
    \end{subfigure}
    \hfill
    \begin{subfigure}[t]{0.45\textwidth}
        \centering
        \includegraphics[width=0.8\linewidth]{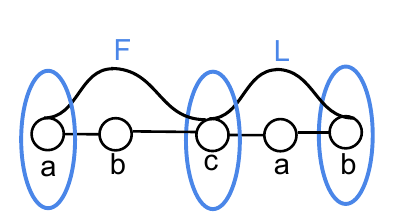}
    \end{subfigure}
    \caption{An Example of Possible Boundary Confusion}
       \label{fig:boundary-confusion}
        \end{figure}
Trivially, a bound of $|V(E)|-1$ suffices (where $V(E)$ is the set of endpoints of edges in $E$) but this is too lossy for our purposes. Using the 
structure of nontrivial backbone-correction shapes, we can show that there are at most $2$ choices for $v$ (if the block is a single edge then there is only one choice for $v$). The reason is that nontrivial backbone-correction shapes are outerplanar, with a unique Hamiltonian cycle formed by the backbone and its boundary edge (see \cref{def:backbone-correction}). Thus, $v$ must be one of the neighbors of $u$ on this Hamiltonian cycle. 

While this observation is very helpful, it is not quite sufficient to prove our norm bounds. For our norm bounds, we need a upper bound of $2$ per block step as this leads to a bound that
\[
\E\left[\Tr (\cmp_\tau \cmp_\tau^\top)^q\right] = \poly(n) \cdot 2^{2q}
\]
which is slightly looser than our bound for constant $q$ but is still sufficient to prove that with high probability, $||\cmp_{\tau}||$ is $2 \pm o_n(1)$. However, for each new block step, we need to specify the following data:
\begin{enumerate}
\item Is the block step an $F$ or $L$ step?
\item If the block step is an $L$ step, what is the next boundary vertex $v$?
\end{enumerate}
Na\"{i}vely, it may take more than a factor of $2$ to specify this data for each block.  

Our key observation is that these two choices are not independent. In particular, as we show in \cref{lem:possibleblockwalksbound}, if we are given whether the current block step is an $F$ step or an $L$ step, a factor of $2$ is sufficient to specify the boundary vertex of the current block step (if it is an $L$ step) and whether the next step is an $F$ step or an $L$ step. The idea is as follows:

\begin{enumerate}
    \item The internal vertices of an $L$-step cannot make subsequent appearances in the walk (or we can treat this as a slack step that gives a negligible contribution to the trace). This implies that if any vertex $v$ is incident to an edge $e$ which only appeared once and is not part of the curret step then $v$ must appear later to ensure that $e$ has even multiplciity so $v$ must be the boundary vertex. Thus, in this case there is no boundary vertex confusion so we just need to specify whether the next block step is an $F$ or an $R$ step.
    \item A vertex $v$ can only start an $L$ step if it is incident to an edge which has only appeared once. This implies that if the current block step is an $L$ step but there is no vertex $v$ which is incident to an edge $e$ that only appeared once and is not part of the curret step then the next step must be an $F$ step. Thus, in this case it is sufficient to specify the boundary vertex $v$ of the current block.
 \end{enumerate}
The resulting sharp norm bound is summarized in the following theorem which is proved in \cref{sec:norm}. To state this theorem, we first define the normalized variance of a linear combination of graph matrices as follows.
\begin{definition}[Normalized Variance of a Linear Combination of Graph Matrices]\label{def:normalized-variance-combination}
	For a matrix \( M= \sum_{\tau \in \calB(M)} c(\tau) \cdot  \cmp_\tau\), we define the normalized variance of $M$ as \[ 
	\Var(M) \coloneqq \sum_{\tau \in \calB(M)} c(\tau)^2\,.
	\]
\end{definition}

 \begin{theorem}[Sharp norm bounds for backbone-correction shapes; informal] \label{thm:informal-main-norm} Let \(B\) be a collection of backbone-correction shapes, and let \[ M = \sum_{\tau\in B} c(\tau)  \cmp_\tau . \]  With high probability, \[ \|M\| \le (1+o_n(1)) \left( 2\sqrt{\Var(M)} + o_n(1) \right). \] More quantitatively, the formal version in \cref{cor:linear-combination-norm-bound} includes an explicit lower-order contribution.\end{theorem}

%% file: theta-function/alt_norm.tex
\section{Spectral Norm Bounds for Backbone-Correction Shapes}
\label{sec:norm}

For our analysis, it is crucial to obtain precise control over the spectra of the matrices constructed along the way, most notably the inner matrix \(Q\). We now describe the techniques we use to obtain tight spectral control. As a byproduct, the same analysis also identifies the limiting spectrum of \(Q\), revealing a free-independence phenomenon among its underlying graph-matrix constituents, i.e., backbone-correction shapes.
\subsection{The Trace Power Method}
To analyze the spectra and norms of our matrices, we use the trace power method. In particular, we analyze $\E\left[\Tr\left(\cm_{\tau_1} \ldots \cm_{\tau_j}\right)\right]$ for arbitrary backbone-correction shapes $\tau_1,\ldots,\tau_j$. For this, it is helpful to view the terms of $\E\left[\Tr\left(\cm_{\tau_1} \ldots \cm_{\tau_j}\right)\right]$ in terms of trace walks. Throughout this section, $D_V$ is an upper bound on the number of vertices in each backbone-correction shape $\tau_i$.

Recall that $\Tr\left(\cm_{\tau_1} \ldots \cm_{\tau_j}\right) = \sum_{\text{trace walks } W \text{ for } \tau_1,\ldots,\tau_j}{\chi_{E(H_W)}(G)}$ where each trace walk $W$ assigns a label in $[n]$ to each vertex of $\tau_1 \circ \ldots \circ \tau_{j}$ and the multigraph $H_W$ consists of the vertices and edges of $\tau_1 \circ \ldots \circ \tau_{j}$ labeled according to $W$ where vertices with the same label are merged together. We first note that the only trace walks which contribute to $\E\left[\Tr\left(\cm_{\tau_1} \ldots \cm_{\tau_j}\right)\right]$ are trace walks $W$ such that every edge of $H_W$ has even multiplicity.
\begin{proposition}
Given a trace walk $W$ for $\tau_1,\ldots,\tau_j$, $\E\left[\chi_{E(H_W)}(G)\right] = 1$ if every edge of $H_W$ has even multiplicity and $\E\left[\chi_{E(H_W)}(G)\right] = 0$ otherwise.
\end{proposition}
\begin{corollary}
$\E\left[\Tr\left(\cm_{\tau_1} \ldots \cm_{\tau_j}\right)\right]$ is equal to the number of trace walks $W$ for $\tau_1,\ldots,\tau_j$ such that every edge of $H_W$ has even multiplicity.
\end{corollary}
It is helpful to partition trace walks based on their intersection patterns (i.e., which vertices of $\tau_1 \circ \ldots \circ \tau_{j}$ have the same label). Observe that if $W$ is a trace walk and $I$ is the intersection pattern for $W$ then 
\begin{enumerate}
\item Since $W$ gives $u_{\tau_1}$ and $v_{\tau_j}$ the same label, we must have that $u_{\tau_1} \sim v_{\tau_j}$.
\item $H_W$ is obtained by taking $\gamma_I$, labeling the vertices/equivalence classes of $\gamma_I$ according to $W$, and then forgetting the distinguished tuples of vertices $U_{\gamma_I}$, $V_{\gamma_I}$ so that we have a multigraph rather than a shape. Thus, $E\left[\chi_{E(H_W)}(G)\right] = 1$ if and only if every edge of $\gamma_I$ has even multiplicity.
\end{enumerate}
Based on this, we make the following definition.
\begin{definition}
Given shapes $\tau_1,\ldots,\tau_j$ with boundary size $1$, we define $Int^{Trace}_{\tau_1,\ldots,\tau_j}$ to be the set of intersection patterns $I$ for $\tau_1,\ldots,\tau_j$ such that $u_{\tau_1} \sim v_{\tau_j}$ and every edge of $\gamma_I$ has even multiplicity.
\end{definition}
\begin{proposition}
$\E\left[\Tr\left(\cm_{\tau_1} \ldots \cm_{\tau_j}\right)\right] = \sum_{I \in Int^{Trace}_{\tau_1,\ldots,\tau_j}}{\frac{n!}{(n-|V(\gamma_I)|)!}} \leq \sum_{I \in Int^{Trace}_{\tau_1,\ldots,\tau_j}}{n^{|V(\gamma_I)|}}$.
\end{proposition}
\begin{proof}
Observe that specifying a trace walk $W$ with intersection pattern $I$ is equivalent to specifying distinct labels in $[n]$ for the $|V(\gamma_I)|$ equivalence classes of $I$ and there are $\frac{n!}{(n-|V(\gamma_I)|)!}$ ways to do this. The result follows by summing over $I \in Int^{Trace}_{\tau_1,\ldots,\tau_j}$ as a trace walk $W$ contributes to $E\left[\Tr\left(\cm_{\tau_1} \ldots \cm_{\tau_j}\right)\right]$ if and only if it has an intersection pattern $I$ such that every edge of $\gamma_I$ has even multiplicity.
\end{proof}
The intersection patterns $I$ which contribute the most to $\E\left[\Tr\left(\cm_{\tau_1} \ldots \cm_{\tau_j}\right)\right]$ are the intersection patterns $I \in Int^{Trace}_{\tau_1,\ldots,\tau_j}$ such that $|V(\gamma_I)|$ is as large as possible. We now analyze the intersection patterns $I \in Int^{Trace}_{\tau_1,\ldots,\tau_j}$ such that $|V(\gamma_I)|$ is large. 
\subsection{Blocks, Vertex Assignment Factors, and Vertex Separators for Blocks}
For our analysis, it is useful to split the factor of $n^{|V(\gamma_I)|}$ among the blocks $\tau_1,\ldots,\tau_j$ based on whether the vertices in each block are appearing for the first time, appearing for the last time, or are making a middle appearance.
\begin{definition}[Vertex appearances]
Given an intersection pattern $I \in Int^{Trace}_{\tau_1,\ldots,\tau_j}$, we say that a vertex $v \in V(\tau_i)$ appears in another block $\tau_{i'}$ if there is a vertex $w \in V(\tau_{i'})$ such that $v \sim w$. More generally, we say that $v$ appears in a set $S \subseteq V(\tau_{i'})$ of vertices of $\tau_{i'}$ if there is a vertex $w \in S$ such that $v \sim w$.

Similarly, we say that a vertex/equivalence class $\hat{v} \in V(\gamma_I)$ appears in a block $\tau_i$ if there is a vertex $w \in V(\tau_{i})$ such that $w \in \hat{v}$. More generally, we say that $\hat{v}$ appears in a set $S \subseteq V(\tau_{i})$ of vertices of $\tau_{i}$ if there is a vertex $w \in S$ such that $w \in \hat{v}$.
\end{definition}
\begin{definition}[First and last vertex appearances]
Given an intersection pattern $I \in Int^{Trace}_{\tau_1,\ldots,\tau_j}$, when we consider the $i$th block $\tau_i$,
\begin{enumerate}
\item We say that a vertex $v \in V(\tau_i)$ appears earlier if there is a $i' < i \in [j]$ such that $v$ appears in $\tau_{i'}$
(note that if $i > 1$ then $u_{\tau_i}$ automatically appears earlier as $u_{\tau_i} = v_{\tau_{i-1}}$). Otherwise, we say that $v$ is appearing for the first time in $\tau_i$.
\item We say that a vertex $v \in V(\tau_i)$ appears later if there is a $i' > i \in [j]$ such that $v$ appears in $\tau_{i'}$
(note that if $i < j$ then $v_{\tau_i}$ automatically appears later as $v_{\tau_i} = u_{\tau_{i+1}}$). Otherwise, we say that $v$ is appearing for the last time in $\tau_i$.
\end{enumerate}
\end{definition}
We use the following vertex assignment scheme for splitting the factor $n^{|V(\gamma_I)|}$ among the blocks $\tau_1,\ldots,\tau_j$.
\begin{enumerate}
\item Whenever a vertex appears for the first time in a block, we assign it a factor of $\sqrt{n}$.
\item Whenever a vertex appears for the last time in a block, we assign it a factor of $\sqrt{n}$.
\end{enumerate}
To distinguish vertices making middle appearances from those making first or last appearances, we follow the terminology of prior work and call vertices making middle appearances \emph{separator vertices}. These vertices no longer contribute a polynomial factor in our factor assignment scheme.
\begin{figure}[h]
    \centering
    \begin{subfigure}[t]{0.7\textwidth}
        \centering
        \includegraphics[width=\linewidth]{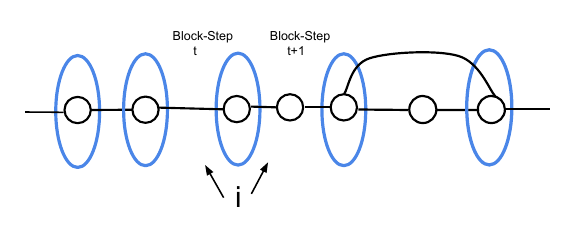}
        \end{subfigure}
        \caption{Example of Vertex Appearances in a Walk of Blocks: vertex $i$ appears in both steps $t$ and $t+1$.}
\end{figure}
As we discuss below, this vertex assignment scheme corresponds to assigning a certain vertex separator to each of the blocks $\tau_1,\ldots,\tau_j$.
\begin{definition}[Canonical vertex separator of a block]
Given an intersection pattern $I \in Int^{Trace}_{\tau_1,\ldots,\tau_j}$, for all $i \in [2,j-1]$, we define the canonical vertex separator $S_I(\tau_i)$ of the $i$th block $\tau_i$ to be the set of vertices in $\tau_i$ which appear both earlier and later. For the first and last blocks, we define $S_I(\tau_1) = U_{\tau_1} = \{u_{\tau_1}\}$ and $S_I(\tau_k) = V_{\tau_k} = \{v_{\tau_k}\}$.
\end{definition}
\begin{proposition}
For all $I \in Int^{Trace}_{\tau_1,\ldots,\tau_j}$ and all $i \in [j]$, $S_I(\tau_i)$ is a vertex separator of $\tau_i$. 
\end{proposition}
\begin{proof}
$S_I(\tau_1) = U_{\tau_1}$ and $S_I(\tau_k) = V_{\tau_k}$ are clearly vertex separators of $\tau_1$ and $\tau_k$, respectively. When $i \in [2,j-1]$, let $w$ be the first vertex on the backbone from $u_{\tau_i}$ to $v_{\tau_i}$ which appears later. If $w = u_{\tau_i}$ then $w \in S_I(\tau_i)$. Otherwise, $w$ must appear earlier as otherwise the edge leading to $w$ would only appear once in $\gamma_I$. Thus, $w \in S_I(\tau_i)$ in this case as well.
\end{proof}
We now express $|V(\gamma_I)|$ in terms of the sizes of the vertex separators $\{S_I(\tau_i): i \in [j]\}$.
\begin{proposition}
For all $I \in Int^{Trace}_{\tau_1,\ldots,\tau_j}$, for all $i \in [j]$, each vertex $v \in V(\tau_i)$ must appear earlier or later (or both).
\end{proposition}
\begin{proof}
If there is a vertex $v \in V(\tau_i)$ which does not appear earlier or later then all of the edges incident to $v$ in $\tau_i$ only appear once in $\gamma_I$ which contradicts the assumption that $I \in Int^{Trace}_{\tau_1,\ldots,\tau_j}$.
\end{proof}
\begin{corollary}
For all intersection patterns $I \in Int^{Trace}_{\tau_1,\ldots,\tau_j}$, 
\[
|V(\gamma_I)| = 1+\frac{1}{2}\sum_{i=1}^{j}{\left(|V(\tau_i)| - |S_I(\tau_i)|\right)}
\]
\end{corollary}
\begin{proof}
Observe that 
\begin{align*}
|V(\alpha_I)| &= \sum_{i=1}^{j}{\sum_{v \in V(\tau_i)}{\frac{1}{2}(1_{v \text{ appears for the first time in } \tau_i} + 1_{v \text{ appears for the last time in } \tau_i})}} \\
&= \frac{|V(\tau_1)|}{2} + \frac{|V(\tau_j)|}{2} + \frac{1}{2}\sum_{i=2}^{j-1}{(|V(\tau_i)| - |S_I(\tau_j)|)}\\
&= 1+\frac{1}{2}\sum_{i=1}^{j}{\left(|V(\tau_i)| - |S_I(\tau_i)|\right)}
\end{align*}
as all vertices in $\tau_1$ are appearing for the first time, all vertices in $\tau_k$ are appearing for the last time, $|S_I(\tau_1)| = |S_I(\tau_k)| = 1$, and for all $i \in [2,j-1]$, a vertex $v \in V(\tau_i)$ is not appearing for the first or last time in $\tau_i$ if and only if $v \in S_I(\tau_i)$.
\end{proof}
Thus, the intersection patterns which have the largest contributions are the intersection patterns  $I \in Int^{Trace}_{\tau_1,\ldots,\tau_j}$ where $|S_I(\tau_i)| = 1$ for all or almost all $i \in [j]$. 
\subsection{Forward Blocks, Return Blocks, and Slack Blocks}
\begin{proposition}\label{prop:earlierorlater}
For all intersection patterns $I \in Int^{Trace}_{\tau_1,\ldots,\tau_j}$, for each $i \in [j]$ such that $|S_I(\tau_i)| = 1$, one of the following two cases must hold:
\begin{enumerate}
\item All vertices of $\tau_i$ appear later and no vertex of $\tau_i$ appears earlier except $u_{\tau_i}$.
\item All vertices of $\tau_i$ appear earlier and no vertex of $\tau_i$ appears later except $v_{\tau_i}$.
\end{enumerate}
\end{proposition}
\begin{proof}
Since $\{u_{\tau_i},v_{\tau_i}\} \in E(\tau_i)$, the only possible vertex separators of $\tau_i$ of size $1$ are $\{u_{\tau_i}\}$ and $\{v_{\tau_i}\}$. If $S_I(\tau_i) = \{u_{\tau_i}\}$ then $u_{\tau_i}$ is the only vertex which can appear earlier. In order to avoid having an edge of $\tau_i$ which only appears once in $\gamma_I$, all vertices of $\tau_i$ must appear later.

Similarly, if $S_I(\tau_i) = \{v_{\tau_i}\}$ then $v_{\tau_i}$ is the only vertex which can appear later. In order to avoid having an edge of $\tau_i$ which only appears once in $\gamma_I$, all vertices of $\tau_i$ must appear earlier.
\end{proof}
With a more careful analysis, we can show that in order to avoid having vertex separators of size greater than $1$, each block $\tau_i$ must be paired with another block $\tau_{i'}$ which has exactly the same edges.
\begin{definition}
Given an intersection pattern $I \in Int^{Trace}_{\tau_1,\ldots,\tau_j}$, we say that a vertex $\hat{v} \in V(\gamma_I)$ is a slack witness vertex if there exists an index $i \in [j]$ such that $|S_I(\tau_i)| \geq 2$ and $\hat{v}$ appears in  $S_I(\tau_i)$.
\end{definition}
\begin{lemma}\label{lem:blockpairing}
For all intersection patterns $I \in Int^{Trace}_{\tau_1,\ldots,\tau_j}$, for all $i \in [j]$, at least one of the following cases must hold:
\begin{enumerate}
\item $|S_I(\tau_i)| \geq 2$.
\item There exists a $k \in [j] \setminus \{i\}$ such that $E(\tau_i) = E(\tau_k)$ (after taking intersections into account). Moreover, for all $k' \in [j] \setminus \{i,k\}$, $E(\tau_i) \cap E(\tau_{k'}) = \emptyset$.
\item At least two slack witness vertices of $\alpha_I$ appear in $\tau_i$.
\item There exists an index $k \in [j] \setminus \{i\}$ such that $E(\tau_i) \cap E(\tau_k) \neq \emptyset$ and $|S_I(\tau_k)| \geq 2$.
\item There exists an index $k \in [j] \setminus \{i\}$ such that $E(\tau_i) \cap E(\tau_k) \neq \emptyset$ and at least two slack witness vertices of $\gamma_I$ appear in $\tau_k$.
\end{enumerate}
\end{lemma}
\begin{remark}
While case 3 subsumes case 1 and case 5 subsumes case 4, we keep these cases separate for conceptual clarity.
\end{remark}
To prove this result, we need the following structural lemma.
\begin{lemma}[Structural Lemma]\label{lem:2-vtx-disjoint-paths}
For all backbone-correction shapes $\tau$ and all 2-colorings of the edges of $\tau$, if not all of the edges of $\tau$ have the same color then there are at least two vertices of $\tau$ which are incident to an edge of both colors.
\end{lemma}
\begin{proof}
If $\tau = \fp_1$ then the result is trivial as $\tau$ only has one edge. If $\tau \neq \fp_1$, observe that $\tau$ is $2$-connected as $\tau$ contains a Hamiltonian cycle. We have the following cases:
\begin{enumerate}
\item If every vertex of $\tau$ is incident to an edge of one color then we can choose any edge of the other color and its endpoints are incident to edges of both colors.
\item If for both colors, there is a vertex which is only incident to edges of that color then letting $W$ be the set of vertices which are incident to both colors, deleting the vertices of $W$ from $\tau$ disconnects $\tau$. Since $\tau$ is 2-connected, $W$ must have at least two vertices, as needed.
\end{enumerate}
\end{proof}
We are now ready to prove Lemma \ref{lem:blockpairing}.
\begin{proof}[Proof of Lemma \ref{lem:blockpairing}]
Assume that there is an intersection pattern $I \in Int^{Trace}_{\tau_1,\ldots,\tau_j}$ and an index $i \in [j]$ such that none of the cases of Lemma \ref{lem:blockpairing} hold. We must have that $|S_I(\tau_i)| = 1$ as otherwise the first case of Lemma \ref{lem:blockpairing} would hold. By Proposition \ref{prop:earlierorlater}, one of the following cases holds:
\begin{enumerate}
\item All vertices of $\tau_i$ appear later and no vertex of $\tau_i$ appears earlier except $u_{\tau_i}$.
\item All vertices of $\tau_i$ appear earlier and no vertex of $\tau_i$ appears later except $v_{\tau_i}$.
\end{enumerate}
Without loss of generality, assume the first case holds. Let $k \in [i+1,j]$ be the first index such that $E(\tau_{i}) \cap E(\tau_{k}) \neq \emptyset$ (after taking intersections into account). We must have that $|S_I(\tau_k)| = 1$ as otherwise the fourth case of Lemma \ref{lem:blockpairing} would hold.

If $E(\tau_i) \setminus E(\tau_{k}) \neq \emptyset$ then by Lemma \ref{lem:2-vtx-disjoint-paths}, there are two vertices $v,w \in V(\tau_i)$ which are incident to both an edge in $E(\tau_{k})$ and an edge in $E(\tau_i) \setminus E(\tau_{k})$. Observe that 
\begin{enumerate}
\item Since $v$ and $w$ are both incident to an edge in $E(\tau_k)$, $v$ and $w$ both appear in $\tau_k$.
\item Since $v$ and $w$ are both incident to an edge in $E(\tau_i) \setminus E(\tau_{k})$ and $k$ is the first index such that $E(\tau_{i}) \cap E(\tau_{k}) \neq \emptyset$, both $v$ and $w$ appear in $\tau_i$ and a block which comes later than $\tau_k$.
\end{enumerate}
Combining these observations, both $v$ and $w$ appear in $S_I(\tau_k)$ so both $\hat{v}$ and $\hat{w}$ are slack witness vertices which appear in $\tau_i$. Thus, the third case of Lemma \ref{lem:blockpairing} holds, which contradicts our assumption that none of the cases of  Lemma \ref{lem:blockpairing} hold.

If $E(\tau_i) \subsetneq E(\tau_k)$ then by Proposition \ref{prop:earlierorlater}, since $|S_I(\tau_k)| = 1$ and at least two vertices of $\tau_k$ appear earlier, all vertices of $\tau_k$ must appear earlier and no vertex of $\tau_k$ except $v_{\tau_k}$ can appear later. Let $k' \in [k-1]$ be the last index such that $k' \neq k$ and $E(\tau_{k'}) \cap E(\tau_{k}) \neq \emptyset$. We have the following cases:
\begin{enumerate}
\item If $E(\tau_{k}) \setminus E(\tau_{k'}) \neq \emptyset$ then we can use similar logic as before. By Lemma \ref{lem:2-vtx-disjoint-paths}, there are two vertices $v,w \in V(\tau_{k})$ which are incident to both an edge in $E(\tau_{k'})$ and an edge in $E(\tau_{k}) \setminus E(\tau_{k'})$. Both $v$ and $w$ must appear in $\tau_k$, $\tau_{k'}$, and a block which is earlier than $\tau_{k'}$ so both $v$ and $w$ appear in $S_I(\tau_{k'})$. Thus, $\hat{v}$ and $\hat{w}$ are both slack witness vertices which appear in $\tau_k$ so case 5 of \cref{lem:blockpairing} holds (unless $k' = i$ in which case case 3 of \cref{lem:blockpairing} holds), which contradicts our assumption that none of the cases of Lemma \ref{lem:blockpairing} hold.
\item If $E(\tau_i) \subsetneq E(\tau_{k}) \subseteq E(\tau_{k'})$ then $i < k' < k$ and for any edge $e = \{u,v\} \in E(\tau_i)$, $e$ must appear in $\tau_{k'}$ and $\tau_k$ as well. This implies that $|S_I(\tau_{k'})| \geq 2$ so case 4 of Lemma \ref{lem:blockpairing} holds, which contradicts our assumption that none of the cases of Lemma \ref{lem:blockpairing} hold.
\end{enumerate}

If $E(\tau_i) = E(\tau_k)$ then if there is an index $k' \in [j] \setminus \{i,k\}$ such that $E(\tau_i) \cap E(\tau_{k'}) \neq \emptyset$, letting $v,w \in V(\tau_k)$ be the endpoints of an edge $e = \{v,w\} \in E(\tau_k) \cap E(\tau_{k'})$, $v,w \in S_I(\tau_k)$ so case 4 of Lemma \ref{lem:blockpairing} holds, which contradicts our assumption that none of the cases of Lemma \ref{lem:blockpairing} hold. If there is no index $k' \in [j] \setminus \{i,k\}$ such that $E(\tau_i) \cap E(\tau_{k'}) \neq \emptyset$ then case 2 of Lemma \ref{lem:blockpairing} holds, which contradicts our assumption that none of the cases of Lemma \ref{lem:blockpairing} hold.
\end{proof}
Note that if only case 2 of Lemma \ref{lem:blockpairing} holds for some block $\tau_i$ then only case 2 of Lemma \ref{lem:blockpairing} holds for the block $\tau_k$ such that $E(\tau_i) = E(\tau_k)$. To see this, observe that since $\tau_i$ and $\tau_j$ have the same vertices and edges (though their boundaries could be different), 
\begin{enumerate}
\item Case 1 or case 4 holds for $\tau_i$ if and only if case 1 or case 4 holds for $\tau_k$.
\item Case 3 or case 5 holds for $\tau_i$ if and only if case 3 or case 5 holds for $\tau_k$.
\end{enumerate}
Based on Lemma \ref{lem:blockpairing}, we make the following definition.
\begin{definition}[Forward Blocks, Return Blocks, and Slack Blocks]
Given an intersection pattern $I \in Int^{Trace}_{\tau_1,\ldots,\tau_j}$, for each block $\tau_i$,
\begin{enumerate}
\item We say that $\tau_i$ is a slack block if at least one of the following conditions holds:
\begin{enumerate}
\item[1.] There is an index $i' \in [j] \setminus \{i\}$ such that $E(\tau_i) \cap E(\tau_{i'}) \neq \emptyset$ and $E(\tau_i) \neq E(\tau_{i'})$.
\item[2.] At least two slack witness vertices appear in $\tau_i$.
\end{enumerate}
\item We say that $\tau_i$ is a forward block if $\tau_i$ is not a slack block and $S_I(\tau_i) = U_{\tau_i}$.
\item We say that $\tau_i$ is a return block if $\tau_i$ is not a slack block and $S_I(\tau_i) = V_{\tau_i}$.
\end{enumerate}
\end{definition}
Lemma \ref{lem:blockpairing} implies that for each $i \in [j]$ such that $\tau_i$ is a slack block, we can extract a little bit of slack from $\tau_i$.
\begin{lemma}\label{lem:extractingslack}
For all intersection patterns $I \in Int^{Trace}_{\tau_1,\ldots,\tau_j}$, $n^{-\frac{\sum_{i \in [j]}{(|S_I(\tau_i)| - 1)}}{2}} \leq n^{-\frac{1}{2}\lceil\frac{\# \text{ of slack blocks}}{24D_V}\rceil}$.
\end{lemma}
\begin{proof}
To prove this, we need to show that $\frac{\# \text{ of slack blocks}}{24D_V} \leq \sum_{i \in [j]}{(|S_I(\tau_i)| - 1)}$. To show this, we show that there is a mapping from slack blocks to pairs $(v,k)$ such that $|S_I(\tau_k)| \geq 2$ and $v \in S_I(\tau_k)$ where at most $12D_V$ slack blocks are mapped to the same pair $(v,k)$. Since the number of such pairs is at most two times $\sum_{i \in [j]}{(|S_I(\tau_i)| - 1)}$, this implies that the number of slack blocks is at most $24D_V$ times $\sum_{i \in [j]}{(|S_I(\tau_i)| - 1)}$.

Our mapping is as follows. Given a slack block $\tau_i$,
\begin{enumerate}
\item If $|S_I(\tau_i)| \geq 2$ then we map $\tau_i$ to $(v,i)$ for some vertex $v \in S_I(\tau_i)$.
\item If $|S_I(\tau_i)| = 1$ and $\tau_i$ contains at least two slack witness vertices then $\tau_i$ must contain a slack witness vertex $\hat{v}$ which appears for the first or last time in $\tau_i$. Since $\hat{v}$ is a slack witness vertex, there is an index $k \in [j]$ such that $|S_I(\tau_k)| \geq 2$ and $\hat{v}$ appears in $S_I(\tau_k)$. In this case, we map $\tau_i$ to $(v,k)$ where $v$ is the vertex in $S_I(\tau_k)$ corresponding to $\hat{V}$ (i.e., $v \in \hat{v}$).
\item If $|S_I(\tau_i)| = 1$, $\tau_i$ contains at most one slack witness vertex, and there exists an index $k \in [j] \setminus \{i\}$ such that $E(\tau_i) \cap E(\tau_k) \neq \emptyset$ and $|S_I(\tau_k)| \geq 2$ then we map $\tau_i$ to a pair $(v,k)$ such that $v \in S_I(\tau_k)$.
\item If $|S_I(\tau_i)| = 1$, $\tau_i$ contains at most one slack witness vertex, and there exists an index $k' \in [j] \setminus \{i\}$ such that $E(\tau_i) \cap E(\tau_{k'}) \neq \emptyset$, $|S_I(\tau_{k'})| = 1$, and $\tau_{k'}$ contains at least two slack witness vertices then $\tau_{k'}$ must contain a slack witness vertex $\hat{v}$ which appears for the first or last time in $\tau_{k'}$. Since $\hat{v}$ is a slack witness vertex, there is an index $k \in [j]$ such that $|S_I(\tau_k)| \geq 2$ and $\hat{v}$ appears in $S_I(\tau_k)$. In this case, we map $\tau_i$ to $(v,k)$ where $v$ is the vertex in $S_I(\tau_k)$ corresponding to $\hat{V}$ (i.e., $v \in \hat{v}$).
\end{enumerate}
We now consider how many different slack blocks $\tau_i$ can map to a given pair $(v,k)$. We observe that 
\begin{enumerate}
\item Case 1 can only happen if $i = k$.
\item Case 2 can only happen if $\tau_i$ is the block where $v$ appears for the first or last time so there are at most two slack blocks $\tau_i$ which lead to $(v,k)$ via case 2.
\item Case 3 can only happen if there is an edge $e \in E(\tau_k)$ which appears for the first or last time in $\tau_i$. Since $|E(\tau_k)| \leq 2D_V - 3$, there are at most $2|E(\tau_k)| \leq 4D_V - 6$ slack blocks $\tau_i$ which map to $(v,k)$ via case $3$.
\item Case 4 can only happen if there is an index $k' \in [k] \setminus \{i,k\}$ such that $v$ appears for the first or last time in $\tau_{k'}$ and there is an edge $e \in E(\tau_{k'})$ which appears for the first or last time in $\tau_i$. There are at most $2$ choices for $k'$, $2D_V-3$ choices for $e$, and $2$ choices for $i$ so there are at most $8D_V - 12$ slack blocks $\tau_i$ which are mapped to $(v,k)$ via case 4.
\end{enumerate}
Putting everything together, there are at most $12D_V$ slack blocks $\tau_i$ which are mapped to $(v,k)$, as needed.
\end{proof}
Since the forward and return blocks are paired up with each other, we can treat these blocks as if they were a single edge. This gives us a walk which we call a block walk. However, we have to be careful because forward and return blocks which are paired up with each other may not have the same endpoints!
\subsection{Block Walks}
For our analysis, it is helpful to think of the blocks as steps of a walk. In particular, forward and return blocks can be thought of as forward and return steps. Slack blocks can be thought of as steps which can warp anywhere but reduce the size of the contribution from the intersection pattern $I$.
\begin{definition}[Block walks]
We define a block walk $BW$ to consist of a sequence of steps which travel on and modify an underlying hypergraph $H$ where each hyperedge $e$ represents a block with boundary vertices and a Hamiltonian cycle and each step of the walk has one of the following types:
\begin{enumerate}
\item In a forward step, we start at our current vertex $u$, go to a new vertex $v$, and draw a hyperedge $e$ with boundary vertices $u$ and $v$ and unknown vertices between $u$ and $v$ (which correspond to the backbone from $u$ to $v$ and remain unknown unless they are identified later). We call each forward step an $F$ step. 
\item In a return step, we start at our current vertex $u$, choose a hyperedge $e'$ incident to $u$, choose our destination vertex $v$ to be either the vertex after $u$ or the vertex before $u$ on the Hamiltonian cycle of $e'$, and then delete $e'$ from our hypergraph and go to $v$ (which we draw as a new vertex if it was previously unknown). We say that this return step is matched with the forward step which created $e'$ as this represents the blocks corresponding to the current step and $e'$ being paired with each other.

If the boundary vertices of $e'$ were $u$ and $v$ (in either order), then we call this return step an $R$ step. Otherwise, we call this return step an $R'$ step.
\item In a slack step, we go from our current vertex $u$ to an arbitrary vertex $w$ which represents a slack block with boundary vertices $u$ and $w$. If $w$ was previously an unknown vertex incident to some hyperedge $e$, we identify  this vertex and record the number of steps it takes to reach $w$ if we start from the left boundary of the block for $e$ and travel along the Hamiltonian cycle for $e$. We do not draw a hyperedge for a slack step because for a trace walk, slack blocks only share edges with other slack blocks and we do not track how these edges are paired up.

We call each slack step an $S$ step.
\end{enumerate}
We start each block walk from a single starting vertex.
\end{definition}
For block walks corresponding to trace intersection patterns, the block walk must return to its starting vertex at the end and have no remaining hyperedges. When analyzing products of backbone-correction shapes, we will consider block walks which do not return to the starting vertex and/or have hyperedges remaining at the end.
\begin{definition}
We say that a block walk $BW$ is a trace block walk if $BW$ ends at the vertex it starts at and has no hyperedges remaining at the end.
\end{definition}
The reason that trace block walks are useful for enumerating intersection patterns $I \in Int^{Trace}_{\tau_1,\ldots,\tau_j}$ is that they give most of the information needed to reconstruct $I$.
\begin{definition}
Given an intersection pattern $I \in Int^{Trace}_{\tau_1,\ldots,\tau_j}$, we define the block walk $BW_I$ to be the block walk obtained by going through the blocks $\tau_1,\ldots,\tau_j$ and taking the corresponding block walk steps. More precisely, for each $i \in [j]$, 
\begin{enumerate}
\item If $\tau_i$ is a forward block, we make the $i$th step of $BW_I$ a forward step which draws a hyperedge corresponding to $\tau_i$ and goes to $v_{\tau_i}$.
\item For each return block $\tau_i$, letting $i' \in [i-1]$ be the index such that $\tau_{i'}$ is the forward block which is paired with $\tau_i$, we make the $i$th step of $BW_I$ a return step which deletes the hyperedge drawn by the $i'$th step, goes to $v_{\tau_i}$, and records whether $v_{\tau_i}$ is the vertex before or after $u_{\tau_i}$ on the Hamiltonian cycle of $\tau_{i'}$.
\item If $\tau_i$ is a slack block, we make the $i$th step of $BW_I$ a slack step which goes directly to $v_{\tau_i}$ and records the needed information about $v_{\tau_i}$.
\end{enumerate}
\end{definition}
\begin{definition}
Given a block walk $BW$, we define $s_{BW}$ to be the number of slack blocks in $BW$.
\end{definition}
\begin{proposition}\label{prop:blockwalkimplications}
Given a trace block walk $BW$ with $k$ steps and backbone-correction shapes $\{\tau_i: i \in [j], \text{the } i\text{th step of } BW \text{ is not a return step of } BW\}$
for the steps of $BW$ which are not return steps, the backbone-correction shapes $\{\tau_i: i \in [j], \text{the } i\text{th step of } BW \text{ is a return step of } BW\}$ for the return steps of $BW$ are uniquely determined and there are at most $(jD_V)^{s_{BW}(D_V-1)}$ intersection patterns $I \in Int^{Trace}_{\tau_1,\ldots,\tau_j}$ such that $BW_{I} = BW$.
\end{proposition}
\begin{proof}
Consider what happens when we go through the vertices of the blocks one by one. When we consider the $i$th block,
\begin{enumerate}
\item If the $i$th block is a forward block with shape $\tau_i$, there are no choices to be made as each vertex of $\tau_i$ except $u_{\tau_i}$ is appearing for the first time. 
\item If the $i$th block is a return block with an unspecified backbone-correction shape $\tau_i$, letting $\tau_{i'}$ be the forward block it is paired with, since $\tau_{i'}$ and $\tau_i$ have the same vertices and edges and both have a unique Hamiltonian cycle, these Hamiltonian cycles must coincide. If $v_{\tau_i}$ is the vertex after $u_{\tau_i}$ on the Hamiltonian cycle of $\tau_{i'}$, the backbone of $\tau_i$ must start from $u_{\tau_i}$ and go backwards along this Hamiltonian cycle until it reaches $v_{\tau_i}$. If $v_{\tau_i}$ is the vertex before $u_{\tau_i}$ on the Hamiltonian cycle of $\tau_{i'}$, the backbone of $\tau_i$ must start from $u_{\tau_i}$ and go forwards along this Hamiltonian cycle until it reaches $v_{\tau_i}$. Note that no other cases are possible and $BW$ specifies which case holds. Since $\tau_{i'}$ and $\tau_i$ have the same vertices and edges, this uniquely determines $\tau_i$.
\item If the $i$th block is a slack block then for each vertex of $\tau_i$, we need to choose which of the existing vertices it is equal to, if any. There are at most $jD_V-1$ existing vertices so for each vertex of $\tau_i$, there are at most $jD_V$ choices. Since there are at most $D_V-1$ vertices in $V(\tau_i) \setminus \{u_{\tau_i}\}$, the total number of choices for $\tau_i$ is at most $(jD_V)^{D_V-1}$.
\end{enumerate}
\end{proof}
\subsection{Trace Power Analysis via Block Walks}
\label{sec:trace-power-block-walks}
We now use block walks to analyze trace powers. We have the following analog of free independence for backbone-correction shapes.
\begin{definition}
Given a sequence of $k$ backbone-correction shapes $\tau_1,\ldots,\tau_{j}$, we define $N_{F/R}(\tau_1,\ldots,\tau_j)$ to be the number of possible block walks such that each step of the walk is either an $F$ step or an $R$ step and for each pair of indices $i < i' \in [j]$ such that the $i$th and $i'$th steps are forward and return steps which are matched with each other, $\tau_{i'} = \tau_i^T$.

Note that if $j=2q$ is even and we have that $\tau_{2j-1} = \tau$ and $\tau_{2j} = \tau^T$ for all $j \in [q]$ for some backbone-correction shape $\tau$ then $N_{F/R}(\tau_1,\ldots,\tau_{2q}) = C_q$ where $C_q = \frac{1}{q+1}\binom{2q}{q}$ is the $q$th Catalan number. Also note that if $j$ is odd then $N_{F/R}(\tau_1,\ldots,\tau_j) = 0$.
\end{definition}
\begin{theorem}\label{thm:shape-sequence-moments}
For all $j \in \mathbb{N}$ and all backbone-correction shapes $\tau_1,\ldots,\tau_j$, 
\[
\E\left[\Tr\left(\cmp_{\tau_1} \ldots \cmp_{\tau_j}\right)\right] = N_{F/R}(\tau_1,\ldots,\tau_j) \cdot n \pm O_{k,D_V}(1).
\]
As a special case, for all backbone-correction shapes $\tau$ and all $q \in \mathbb{N}$,
\[
\E\left[\Tr\left(\left(\cmp_{\tau}\cmp_{\tau}^T\right)^{q}\right)\right] = {C_q}\cdot n \pm O_{q,D_V}(1).
\]
\end{theorem}

\begin{proof}
The key idea is that the dominant terms are the trace block walks with no slack steps and these terms contribute roughly $N_{F/R}(\tau_1,\ldots,\tau_j)n$. 

To show this, we first observe that trace block walks with no $S$ steps cannot have any $R'$ steps as after the first $R'$ step, the current vertex is disconnected from the starting vertex and there is no way to return to the connected component containing the starting vertex without an $S$ step. 

We now observe that for each trace block walk $BW$ which only has $F$ and $R$ steps, there exists a unique intersection pattern $I \in Int^{Trace}_{\tau_1,\ldots,\tau_j}$ such that $BW_I = BW$ if and only if for all indices $i < i' \in [j]$ such that steps $i$ and $i'$ of $BW$ are a forward step and a return step which are matched up, $\tau_{i'} = \tau_i^T$. Otherwise, there are no intersection patterns $I \in Int^{Trace}_{\tau_1,\ldots,\tau_j}$ such that $BW_I = BW$. This can be shown using reasoning similar to that used to prove Proposition \ref{prop:blockwalkimplications}. For each $i \in [j]$ such that step $i$ of $BW$ is a return step, letting $i'$ be the forward step which step $i$ is matched with, step $i$ of $BW$ must start at $v_{\tau_{i'}}$ and end at $u_{\tau_{i'}}$. This is possible if and only if $\tau_i = \tau_{i'}^T$ and $I$ matches each vertex in $\tau_{i'}$ to its mirror image in $\tau_i = \tau_{i'}^T$. Thus, the number of intersection patterns $I \in Int^{Trace}_{\tau_1,\ldots,\tau_j}$ which have no slack blocks is $N_{F/R}(\tau_1,\ldots,\tau_j)n$. Each such intersection pattern contributes $n^{1 - |V(\gamma_I)|}\frac{n!}{(n-|V(\gamma_I)|)!}$ to $\E\left[\Tr\left(\cmp_{\tau_1} \ldots \cmp_{\tau_j}\right)\right]$ so the total contribution from these intersection patterns is $N_{F/R}(\tau_1,\ldots,\tau_j)n - O_{j,D_V}(1)$.

For each trace block walk $BW$ which has a slack step, for all intersection patterns $I$ such that $BW_I = BW$, $\sum_{i \in [j]}{(|S_I(\tau_i)| - 1)} > 0$ so $\sum_{i \in [j]}{(|S_I(\tau_i)| - 1)} \geq 1$. The contribution from each such intersection pattern is at most $n^{1-\sum_{i \in [j]}{(|S_I(\tau_i)| - 1)}} \leq 1$. Since the number of intersection patterns $I \in Int^{Trace}_{\tau_1,\ldots,\tau_j}$ is at most $(jD_V)^{jD_V}$, the contribution from intersection patterns which have at least one slack block is $O_{j,D_V}(1)$.
\end{proof}
In order to prove probabilistic norm bounds for backbone-correction shapes, we need an upper bound for the contribution from intersection terms with at least one slack block which has a smaller dependence on $j$.
\begin{theorem}\label{thm:shapesequencetracebound}
For all $k \in \mathbb{N}$ and all backbone-correction shapes $\tau_1,\ldots,\tau_j$, 
\[
\E\left[\Tr\left(\cmp_{\tau_1} \ldots \cmp_{\tau_j}\right)\right] \leq N_{F/R}(\tau_1,\ldots,\tau_j)n + \sum_{s=1}^{j}{2^{j+2s}\binom{j}{s}\left(jD_V\right)^{sD_V}n^{1-\frac{1}{2}\lceil\frac{\# \text{ of slack blocks}}{24D_V}\rceil
}}.
\]
\end{theorem}
\begin{proof}
We prove this by bounding the number of possible trace block walks with a given number of slack steps.
\begin{lemma}\label{lem:possibleblockwalksbound}
For all $j,s \in \mathbb{N}$, the number of trace block walks with $j$ steps where $s$ of these steps are slack steps is at most $2^{j+2s}\binom{j}{s}\left(jD_V\right)^{s}$.
\end{lemma}
\begin{proof}
For the $s$ slack steps, we have the following choices:
\begin{enumerate}
\item There are $\binom{j}{s}$ choices for which of the $j$ steps are slack steps.
\item For each slack step, there are at most $jD_V$ choices for which existing hyperedge $e$ (if any) the destination vertex is on and the location of this destination vertex $v$ on the Hamiltonian cycle of $e$ relative to the existing known vertices. There are $4$ choices for whether there are unknown vertices between $v$ and the two known vertices on the Hamiltonian cycle of $e$ which are closest to $v$.
\item If the following step is not a slack step, we specify whether this step is a forward step or a return step (we explain why this is useful below).
\end{enumerate}
Thus, the total number of possibilities for the slack steps is at most $2^{3s}\binom{j}{s}(jD_V)^{s}$.

For the $j-s$ non-slack steps, as described in the proof overview, the key idea is that if we know whether the current step is a forward step or a return step, a factor of $2$ is sufficient to specify both the destination vertex of the current step and whether the next step is a forward step or a return step (if it is not a slack step). To show this, we make the following observations:
\begin{enumerate}
\item If the current step is an $F$ step, its destination vertex does not need to be specified as it is new so we just need to specify whether the next step is a forward step or a return step.
\item If the current step is an $R$ step, the current vertex $u$ is incident to exactly one hyperedge $e$ as if it were incident to more than one hyperedge $e$, $u$ would have to appear later on in which case the current step would not be a return step. We now have the following cases:
\begin{enumerate}
\item[1.] If $e$ is incident to a vertex $v$ which is incident to another hyperedge $e'$, $v$ must appear later on so $v$ must be the destination vertex of the current step as otherwise it wouldn't be a return step. In this case, we only need to specify whether the next step is a forward step or a return step.
\item[2.] If no vertex of $e$ is incident to an existing hyperedge then there are two choices for the destination vertex $v$. However, in this case the next step must be a forward step (unless it is a slack step which would have already been specified).
\end{enumerate}
Thus, the total number of choices for the non-slack steps is at most $2^{j-s}$. Combining this with the bound for the slack steps proves the lemma.
\end{enumerate}
\end{proof}
To complete the proof of Theorem \ref{thm:shapesequencetracebound}, we recall that 
\begin{enumerate}
\item By Proposition \ref{prop:blockwalkimplications}, for each trace block walk $BW$ with $s$ slack steps, there are at most $(jD_V)^{s(D_V-1)}$ intersection patterns $I \in Int^{Trace}_{\tau_1,\ldots,\tau_j}$ such that $BW_{I} = BW$.
\item By Lemma \ref{lem:extractingslack}, for each $I \in Int^{Trace}_{\tau_1,\ldots,\tau_j}$ with $s$ slack blocks, the contribution to $\E\left[\Tr\left(\cmp_{\tau_1} \ldots \cmp_{\tau_j}\right)\right]$ from $I$ is at most $n^{1-\frac{1}{2}\lceil\frac{\# \text{ of slack blocks}}{24D_V}\rceil
}$.
\end{enumerate}
The result follows from combining these bounds.
\end{proof}
We can use the same reasoning to analyze linear combinations of backbone-correction shapes.
\begin{theorem}\label{thm:linearcombinationnormbound}
If $M = \sum_{\tau \in \calB(M)}{c(\tau)\cmp_{\tau}}$ then recalling that $\Var(M) = \sqrt{\sum_{\tau \in \calB(M)}{c(\tau)^2}}$ and letting $B = \frac{1}{\Var(M)}\sum_{\tau \in \calB(M)}{|c(\tau)|}$, for all non-negative integers $q$, 
\[
\E\left[\Tr\left((MM^T)^q\right)\right] \leq \Var(M)^{2q}\left(C_{q}n + \sum_{s=1}^{2q}{B^{s}2^{2q+2s}\binom{2q}{s}\left(2qD_V\right)^{sD_V}n^{1-\frac{1}{2}\lceil\frac{s}{24D_V}\rceil}}\right).
\]
\end{theorem}
\begin{proof}
We can prove this by summing the contribution from each trace block walk and using the same reasoning that we used to prove Theorem \ref{thm:shapesequencetracebound}. The main additional observations that we need are as follows:
\begin{enumerate}
\item Whenever we have a forward step and a return step which are matched up, there is a one-to-one correspondence between the backbone-correction shape $\tau$ for the forward step and the corresponding backbone-correction shape $\tau'$ for the return step (though $\tau'$ and $\tau$ may not be equal). Letting $c'(\tau) = c(\tau')$, by Cauchy--Schwarz, we have
\[
    \left|\sum_{\tau \in \calB(M)}c(\tau)c'(\tau)\right| \leq \sqrt{\left(\sum_{\tau \in \calB(M)}c(\tau)^2\right)
    \left(\sum_{\tau \in \calB(M)}(c'(\tau))^2\right)} = {\Var(M)}^2.
\]
\item Whenever we have a slack step, we can choose any backbone-correction shape $\tau$ for this slack step so this gives a factor of $B\Var(M) = \sum_{\tau \in \calB(M)}{|c(\tau)|}$.
\end{enumerate}
\end{proof}
\begin{corollary}\label{cor:linear-combination-norm-bound}
If $D_V \geq 4$ and $n \geq B^{100D_V}(4D_V)^{100D_V^2}$ then for all $\epsilon \geq n2^{-n^{\frac{1}{200D_V^2}}}$, with probability at least $1 - \epsilon$,
$||M|| < 2\Var(M)\left(1 + \frac{\log_2\left(\frac{n}{\epsilon}\right)}{n^{\frac{1}{200D_V^2}}}\right)$.
\end{corollary}
\begin{proof}
Observe that using Markov's inequality, for all $q \in \mathbb{N}$, 
\[
P\left(||M|| \geq \sqrt[2q]{\frac{E\left[\Tr\left((MM^T)^q\right)\right]}{\epsilon}}\right) \leq P\left(\Tr\left((MM^T)^q\right) \geq \frac{E\left[\Tr\left((MM^T)^q\right)\right]}{\epsilon}\right) \leq \epsilon.
\] 
We now take $q = \lceil{\frac{1}{2}n^{\frac{1}{200{D_V^2}}}}\rceil$ and make the following observations:
\begin{enumerate}
\item For all $s \in \mathbb{N}$,
\[
\begin{aligned}
B^{s}2^{2q+2s}\binom{2q}{s}\left(2qD_V\right)^{sD_V}n^{1-\frac{1}{2}\lceil\frac{s}{24D_V}\rceil}
&\leq \frac{2^{2q}B^{s}\left(4q^{2}D_V\right)^{sD_V}}{s!n^{\frac{s}{50D_V}}}n \\
&\leq \frac{1}{s!}\left(\frac{B(2D_V)^{D_V}}{n^{\frac{1}{100D_V}}}\right)^{s}n
\leq \frac{2^{2q-s}}{s!}n.
\end{aligned}
\]
\item Since $q \geq 2$, $C_q = \frac{1}{q+1}\binom{2q}{q} < \frac{2^{2q}}{6}$ and $\sum_{s=1}^{2q}{\frac{2^{-s}}{s!}} < \frac{3}{4}$ so $\E\left[\Tr\left((MM^T)^q\right)\right] < (2\Var(M))^{2q}n$.
\item $\left(1 + \frac{\log_2\left(\frac{n}{\epsilon}\right)}{n^{\frac{1}{200D_V^2}}}\right)^{2q} \geq 2^{\log_2\left(\frac{n}{\epsilon}\right)} = \frac{n}{\epsilon}$ so $\sqrt[2q]{\frac{\E\left[\Tr\left((MM^T)^q\right)\right]}{\epsilon}} < 2\Var(M)\left(1 + \frac{\log_2\left(\frac{n}{\epsilon}\right)}{n^{\frac{1}{200D_V^2}}}\right)$. To see this, observe that $x := \frac{n^{\frac{1}{200D_V^2}}}{\log_2\left(\frac{n}{\epsilon}\right)} \geq 1$ and for all $x \geq 1$, $\left(1+\frac{1}{x}\right)^{x} \geq 2$.
\end{enumerate}
\end{proof}

%% file: cheby_shape.tex
\section{Chebyshev Polynomials via Shape Concatenation}\label{sec:cheby-shape-approx}
The goal of this section is to show that Chebyshev polynomials of a graph-matrix linear combination admit a clean approximation by proper shape concatenations. The main task is therefore to classify and bound the error terms arising from nontrivial intersections.
Throughout this section, $D_V$ bounds the total number of vertices in each shape product and correction under consideration.
\subsection{Intuition From Block Walks}
To start with, we give a precise analysis of what happens when we multiply backbone-correction shapes together. In particular, we identify the terms which have constant norm and show that all of the other terms have smaller norm.

We can gain much of the intuition for what happens when we take products of backbone-correction shapes from block walks. The intuition is that products whose block walk contains a slack block have small norm. Thus, the only block walks which can result in terms with constant norm are block walks where each step is an $F$ step, an $R$ step, or an $R'$ step.

These block walks result in shapes $\gamma_I$ which are split into connected components where each component of $\gamma_I$ is the concatenation of backbone-correction shapes, $u_{\gamma_I}$ is at the start of one of these connected components, and $v_{\gamma_I}$ is at the end of one of these connected components (possibly the same component which contains $u_{\gamma_I}$).

However, it turns out that some of these terms also have small norm. In particular, if there is a connected component of $\gamma_I$ which is not an isolated vertex and is not connected to $u_{\gamma_I}$ or $v_{\gamma_I}$, then $n^{-\frac{\sum_{i=1}^{j}{(|V(\tau_i)| - 1)}}{2}}||\cm_{\gamma_I}||$ is small. Thus, the terms which have constant norm are either compositions of backbone-correction shapes (which are the main terms in our analysis) or terms where there are two disconnected components: a composition of backbone-correction shapes starting at $u_{\gamma_I}$ and a composition of backbone-correction shapes ending at $v_{\gamma_I}$. We make this precise below.
\begin{definition}
For all intersection patterns $I \in Int_{\tau_1,\ldots,\tau_j}$, for all $i \in [2,j-1]$, we define $S_I(\tau_i)$ to be the set of vertices of $\tau_i$ which appear both earlier and later. We define $S_I(\tau_1) = \{u_{\tau_1}\}$ if $u_{\tau_1}$ appears later and $S_I(\tau_1) = \emptyset$ otherwise. Similarly, we define $S_I(\tau_j) = \{v_{\tau_j}\}$ if $v_{\tau_j}$ appears earlier and $S_I(\tau_j) = \emptyset$ otherwise.

Note that $S_I(\tau_i)$ may not be a vertex separator of $\tau_i$.
\end{definition}
\begin{definition}
We say that a vertex $\hat{v} \in V(\gamma_I)$ is a slack witness vertex if there exists an index $i \in [j]$ such that $|S_I(\tau_i)| \geq 2$ and $\hat{v}$ appears in $S_I(\tau_i)$.
\end{definition}
\begin{definition}[Block types for backbone-correction shape products]
Given an intersection pattern $I$,
\begin{enumerate}
\item We say that $\tau_i$ is a slack block if at least one of the following holds:
\begin{enumerate}
\item[1.] There are at least two slack witness vertices which appear in $\tau_i$. Note that this includes the case when $|S_I(\tau_i)| \geq 2$.
\item[2.] $\tau_i$ does not have an edge in common with any other block and there is a vertex $v \in V(\tau_i) \setminus \{v_{\tau_i}\}$ which appears later or a vertex $v \in V(\tau_i) \setminus \{v_{\tau_i}\}$ which appears earlier (or both).
\item[3.] There exists an index $i' \in [j] \setminus \{i\}$ such that $E(\tau_i) \cap E(\tau_{i'}) \neq \emptyset$ and $E(\tau_i) \neq E(\tau_{i'})$.
\end{enumerate}
\item We say that $\tau_i$ is a forward block if $\tau_i$ is not a slack block and no vertex of $\tau_i$ except $u_{\tau_i}$ appears earlier.
\item We say that $\tau_i$ is a return block if $\tau_i$ is not a slack block, no vertex of $\tau_i$ appears later except $v_{\tau_i}$, and there is a forward block $\tau_{i'}$ such that $E(\tau_{i'}) = E(\tau_{i})$. In this case, we say that the forward block $\tau_{i'}$ and the return block $\tau_i$ are paired up.
\end{enumerate}
Note that forward blocks are either paired with a return block (i.e., they have exactly the same vertices and edges, though the boundaries may be different) or do not have any edges in common with any other block.
\end{definition}
\begin{lemma}\label{lem:noslackblockpossibilities}
If $I \in Int_{\tau_1,\ldots,\tau_{j}}$ is an intersection pattern which does not have any slack blocks, then the block walk for $I$ consists of the following types of steps:
\begin{enumerate}
\item $F$ steps which start at the current vertex $u$, go to a new vertex $v$, and draw a hyperedge $e$ with boundary vertices $u$ and $v$. As long as $e$ is still present, the unknown vertices of $e$ will remain unknown, so we can think of $e$ as an edge.
\item $R$ steps which start at the current vertex $v$, delete a hyperedge $e$ which is still present and is incident to $v$ (which will be unique and will be the most recent hyperedge which is still present), and go to the other boundary vertex $u$ of $e$.
\item $R'$ steps which start at the current vertex $v$, delete a hyperedge $e$ which is still present and is incident to $v$ (which will be unique and will be the most recent hyperedge which is still present), and go to an unknown vertex of $e$ which is then identified.

Note that $R'$ steps can only be followed by $F$ steps as the newly identified vertex is not incident to any hyperedges.
\end{enumerate}
This implies that at each point of the block walk, the hyperedges which are still present either consist of a single path (which could have length $0$) from the starting vertex $u_{\tau_1}$ to the current vertex $v$ or consist of two or more disjoint paths of hyperedges where the first such path (which could have length $0$) starts from the starting vertex $u_{\tau_1}$ and the most recent such path leads to the current vertex $v$.
\end{lemma}
\begin{proof}
We prove this by induction. The base case $j = 0$ is trivial. For the inductive step, assume the result is true for all intersection patterns $I \in Int_{\tau_1,\ldots,\tau_{j-1}}$ and consider an intersection pattern $I' \in Int_{\tau_1,\ldots,\tau_{j}}$ which extends $I$.

If the $j$th step is an $F$ step, this step extends the most recent path. If the $j$th step is a return step, it must backtrack by deleting the most recent hyperedge $e$ which is still present as this is the only hyperedge which is incident to the current vertex $v$. If this step is an $R$ step, it stays on the most recent path after backtracking. If this step is an $R'$ step, it instead abandons this path and starts a new path from the previously unknown vertex $v$ it goes to.
\end{proof}
\begin{remark}
Note that since $F$ steps do not need to be paired up, when the current vertex $v$ is incident to a hyperedge $e$, we can have an $R'$ step even if another vertex $u$ of $e$ is incident to another hyperedge $e'$.
\end{remark}
\subsection{Products of Backbone-Correction Shapes}
We now give a more precise analysis of the products of backbone-correction matrices.
\begin{definition}
Given shapes $\tau_1,\ldots,\tau_j$, for all $i \in \{0,1,\ldots,j\}$, we define the shape $\tau_1 \circ \ldots \circ \tau_i | \tau_{i+1} \circ \ldots \circ \tau_j$ to be the shape $\alpha$ consisting of the following components:
\begin{enumerate}
\item $\tau_1 \circ \ldots \circ \tau_i$ (viewed as a graph). If $i = 0$ then we instead take this component to be a single vertex $u$ which is not in any of the shapes $\tau_1,\ldots,\tau_j$.
\item $\tau_{i+1} \circ \ldots \circ \tau_j$ (viewed as a graph). If $i = j$ then we instead take this component to be a single vertex $v$ which is not in any of the shapes $\tau_1,\ldots,\tau_j$.
\end{enumerate}
If $i > 0$, we take $U_{\alpha} = \{u_{\tau_1}\}$. Otherwise, we take $U_{\alpha} = \{u\}$. Similarly, if $i < j$, we take $V_{\alpha} = \{v_{\tau_j}\}$. Otherwise, we take $U_{\alpha} = \{v\}$.
\end{definition}
\begin{theorem}\label{thm:backbonecorrectionproductapproximation}
For all $j \in \mathbb{N}$ and all backbone-correction shapes $\tau_1,\ldots,\tau_{j+1}$, let ${\tau'_j}^T$ be the other backbone-correction shape which can be paired with $\tau_j$ (i.e., ${\tau'_j}^T$ has the same vertices and edges as $\tau_j^T$ but starts at $v_{\tau_j}$ and ends at the vertex preceding $v_{\tau_j}$ on the Hamiltonian cycle of $\tau_j$). Then the following approximations hold:
\begin{enumerate}
\item $\cmp_{\tau_1 \circ \ldots \circ \tau_j}\cmp_{\tau_{j+1}} \approx \cmp_{\tau_1 \circ \ldots \circ \tau_{j+1}} + 1_{\tau_{j+1} = \tau_j^T}\cmp_{\tau_1 \circ \ldots \circ \tau_{j-1}} + 1_{\tau_{j+1} = {\tau'_j}^T}\cmp_{\tau_1 \circ \ldots \circ \tau_{j-1} | \emptyset}$.
\item $\cmp_{\tau_1 \circ \ldots \circ \tau_j | \emptyset}\cmp_{\tau_{j+1}} \approx \cmp_{\tau_1 \circ \ldots \circ \tau_j | \tau_{j+1}}$.
\item $\cmp_{\tau_1 \circ \ldots \circ \tau_{j-1} | \tau_j}\cmp_{\tau_{j+1}} \approx \cmp_{\tau_1 \circ \ldots \circ \tau_{j-1} | \tau_j \circ \tau_{j+1}} + \left(1_{\tau_{j+1} = \tau_j^T} + 1_{\tau_{j+1} = {\tau'_j}^T}\right)\cmp_{\tau_1 \circ \ldots \circ \tau_{j-1} | \emptyset}$.
\item $\forall i \in [j-2]$, $\cmp_{\tau_1 \circ \ldots \circ \tau_{i} | \tau_{i+1} \circ \ldots \circ \tau_j}\cmp_{\tau_{j+1}} \approx  \cmp_{\tau_1 \circ \ldots \circ \tau_{i} | \tau_{i+1} \circ \ldots \circ \tau_{j+1}} + 1_{\tau_{j+1} = \tau_j^T}\cmp_{\tau_1 \circ \ldots \circ \tau_{i} | \tau_{i+1} \circ \ldots \circ \tau_{j-1}}$.
\end{enumerate}
In particular, if $D_V \geq 4$ and $n \geq (4D_V)^{400D_V^2}$, then for each of these cases, with probability at least $1 - n{D_V^{D_V}}2^{-n^{\frac{1}{400D_V^2}}}$, the difference between the left-hand side and the right-hand side has norm at most $2D_V^{D_V}n^{-\frac{1}{400D_V^2}}$.
\end{theorem}
\begin{proof}
We first explain the constant norm terms which we swept under the rug in our overview. 
\begin{enumerate}
\item When we multiply $\cm_{\tau_1 \circ \ldots \circ \tau_j}$ by $\cm_{{\tau'_j}^{\top}}$, the term corresponding to an $R'$ step gives
roughly $n^{|V(\tau_j)| - 2}$ times $\cm_{\tau_1 \circ \ldots \circ \tau_{j-1}|\emptyset}$ as after we delete double edges, $|V(\tau_j)|-2$ vertices of $\tau_j$ become isolated and $D_V = o(n)$ hence each isolated vertex gives a factor of $n$ effectively. Since $\cmp_{\tau_1 \circ \ldots \circ \tau_j} = n^{-\left(\frac{|V(\tau_1 \circ \ldots \circ {\tau_{j-1}})| - 1}{2} + \frac{|V(\tau_j)| - 1}{2}\right)}\cm_{\tau_1 \circ \ldots \circ \tau_j}$, $\cmp_{{\tau'_j}^{\top}} = n^{-\frac{|V(\tau_j)| - 1}{2}}\cm_{{\tau'_j}^{\top}}$,  and 
$\cmp_{\tau_1 \circ \ldots \circ \tau_{j-1}|\emptyset} = n^{-\frac{|V(\tau_1 \circ \ldots \circ {\tau_{j-1}})|+1}{2}}\cm_{\tau_1 \circ \ldots \circ \tau_{j-1}|\emptyset}$ (note that there is an extra factor of $\frac{1}{\sqrt{n}}$ as this shape is disconnected), we have that the term corresponding to an $R'$ step gives roughly 
\[
n^{|V(\tau_j)| - 2 -\left(\frac{|V(\tau_1 \circ \ldots \circ {\tau_{j-1}})| - 1}{2} + \frac{|V(\tau_j)| - 1}{2}\right)-\frac{|V(\tau_j)| - 1}{2}}\cm_{\tau_1 \circ \ldots \circ \tau_{j-1}|\emptyset} = \cmp_{\tau_1 \circ \ldots \circ \tau_{j-1}|\emptyset}.
\]
\item When we multiply $\cm_{\tau_1 \circ \ldots \circ \tau_{j-1} | \tau_j}$ by $\cm_{{\tau'_j}^{\top}}$, the term corresponding to an $R'$ step gives roughly $n^{|V(\tau_j)| - 1}$ times $\cm_{\tau_1 \circ \ldots \circ \tau_{j-1}|\emptyset}$ as after we delete double edges, $|V(\tau_j)|-1$ vertices of $\tau_j$ become isolated. Similiar to before, this factor is canceled out by the normalization factors.
\end{enumerate}
We now show that the other terms are small. In order to take advantage of the machinery we developed for analyzing expected traces of products of backbone-correction shapes, we approximate the matrices appearing on the left-hand sides of these cases by scalar multiples of $\cm_{\gamma_I}$ for appropriately chosen intersection patterns $I$. Multiplying by $\cmp_{\tau_{j+1}}$ gives intersection patterns $I'$ which extend $I$. We show that the error terms are small by applying the trace power method to these terms.
\begin{definition}
Given a natural number $j \in \mathbb{N}$ and backbone-correction shapes $\tau_1,\ldots,\tau_{j+1}$, we say that an intersection pattern $I' \in Int^{Trace}_{\tau_1,\ldots,\tau_{j+1}}$ extends an intersection pattern $I \in Int^{Trace}_{\tau_1,\ldots,\tau_{j}}$ if for all indices $i < i' \in [j]$ and pairs of vertices $u \in V(\tau_i)$, $v \in V(\tau_{i'})$, $u$ and $v$ are in the same equivalence class of $I'$ if and only if $u$ and $v$ are in the same equivalence class of $I$.
\end{definition}
\begin{proposition}
For all backbone-correction shapes $\tau_1,\ldots,\tau_{j+1}$ and all intersection patterns $I \in Int_{\tau_1,\ldots,\tau_{j}}$, 
\[
\cm_{\gamma_I}\cm_{\tau_{j+1}} = \sum_{I' \in Int_{\tau_1,\ldots,\tau_{j+1}}: I' \text{ extends } I}{\cm_{\gamma_{I'}}}.
\]
\end{proposition}

\begin{definition}
Given backbone-correction shapes $\tau_1,\ldots,\tau_{j+1}$ and a natural number $q \in \mathbb{N}$, we say that an intersection pattern $I'' \in Int^{Trace}_{\tau_1,\ldots,\tau_{j+1},\tau_{j+1}^T,\ldots,\tau_1^T,\ldots, \tau_1,\ldots,\tau_{j+1},\tau_{j+1}^T,\ldots,\tau_1^T}$ (where the shapes $\tau_1,\ldots,\tau_{j+1},\tau_{j+1}^T,\ldots,\tau_1^T$ are repeated $q$ times) is consistent with an intersection pattern $I' \in Int_{\tau_1,\ldots,\tau_{j+1}}$ if for all pairs of vertices $u,v$ in the same group of $\tau_1,\ldots,\tau_{j+1}$ or $\tau_{j+1}^T,\ldots,\tau_1^T$, $u$ and $v$ are in the same equivalence class of $I''$ if and only if they are in the same equivalence class of $I'$ applied to this group of $k+1$ backbone-correction shapes.
\end{definition}
\begin{lemma}
For all $I' \in Int_{\tau_1,\ldots,\tau_{j+1}}$ and $I'' \in Int^{Trace}_{\tau_1,\ldots,\tau_{j+1},\tau_{j+1}^T,\ldots,\tau_1^T,\ldots, \tau_1,\ldots,\tau_{j+1},\tau_{j+1}^T,\ldots,\tau_1^T}$ such that $I''$ is consistent with $I'$, if $I'$ has a slack block then $I''$ has at least $2q$ slack blocks.
\end{lemma}
\begin{proof}
We make the following observations:
\begin{enumerate}
\item If there are indices $i < i' \in [j+1]$ such that, according to $I'$, $E(\tau_i) \cap E(\tau_{i'}) \neq \emptyset$ but $E(\tau_i) \neq E(\tau_{i'})$, then this will hold for all copies of $\tau_i$ and $\tau_{i'}$, giving at least $4q$ slack blocks.
\item If there is a block $\tau_i$ such that at least two vertices of $\tau_i$ appear both earlier and later according to $I'$ then this will hold for all copies of $\tau_i$ which gives at least $2q$ slack blocks. This implies that if there are slack witness vertices in any of the blocks then there must be at least $2q$ slack blocks.
\item If according to $I'$, $\tau_i$ does not share an edge with any other block and there is a vertex $v \in V(\tau_i) \setminus \{v_{\tau_i}\}$ which appears later then there must be a slack block in each group of $\tau_1,\ldots,\tau_{j+1}$ and each group of $\tau_{j+1}^T,\ldots,\tau_1^T$ so the total number of slack blocks must be at least $2q$.

To see this, assume that there is a group of $\tau_1,\ldots,\tau_{j+1}$ which does not contain a slack block. Let $w$ be the first vertex in $\tau_{i+1} \circ \ldots \circ \tau_{j+1}$ which is in the same equivalence class of $I'$ as $v$ and let $i' \in [i+1,j+1]$ be the index such that $w \in V(\tau_{i'}) \setminus \{u_{\tau_{i'}}\}$.  We must have that $w = v_{\tau_{i'}}$ as otherwise $\tau_{i'}$ would be a slack block. 

Now consider $I''$ and let $i''$ be the first index in $[i+1,i']$ such that $u_{\tau_{i''}}$ appears in a block which comes before $\tau_{i}$. We make the following observations which give a contradiction:
\begin{enumerate}
\item[1.] To avoid having $\tau_{i'}$ be a slack block, $\tau_{i'}$ must be matched by $I''$ with a block which comes earlier. Since $v \sim w$ does not appear in blocks $\tau_{i+1},\ldots,\tau_{i'-1}$ and $\tau_i$ is not matched with $\tau_{i'}$, this block must come before $\tau_i$ so $u_{\tau_{i'}}$ must appear in a block which comes before $\tau_i$. Thus, the index $i''$ must exist.
\item[2.] In order to avoid having $\tau_{i''-1}$ be a slack block, we must have that $i'' > i+1$ and $\tau_{i''-1}$ is matched with a block $\tau_k$ which comes earlier. If $\tau_k$ comes before $\tau_i$ then $u_{\tau_{i''-1}}$ appears in a block which comes before $\tau_i$ which contradicts how we chose $i''$. If $k \in [i+1,i''-1]$ then letting $u$ be the vertex in $\tau_k$ such that $u \sim v_{\tau_{i''-1}}$, $u$ appears in a block which comes before $\tau_i$. In order to avoid having $\tau_k$ be a slack block, we must have that $u_{\tau_k} = u$. However, this contradicts how we chose $i''$.
\end{enumerate}

We can use similar logic for the case where, according to $I'$, $\tau_i$ does not share an edge with any other block and there is a vertex $v \in V(\tau_i) \setminus \{u_{\tau_i}\}$ which appears earlier.
\end{enumerate}
\end{proof}
\begin{lemma}
For all $I' \in Int_{\tau_1,\ldots,\tau_{j+1}}$ and $I'' \in Int^{Trace}_{\tau_1,\ldots,\tau_{j+1},\tau_{j+1}^T,\ldots,\tau_1^T,\ldots, \tau_1,\ldots,\tau_{j+1},\tau_{j+1}^T,\ldots,\tau_1^T}$ such that $I'$ has no slack blocks and $I''$ is consistent with $I'$, if the shape $\gamma_{I',red}$ obtained from $\gamma_{I'}$ by first deleting all doubled edges and then deleting all isolated vertices contains a connected component which is not connected to $u_{\gamma_{I'}}$ or $v_{\gamma_{I'}}$ then $I''$ has at least $2q$ slack blocks. 
\end{lemma}
\begin{proof}
Observe that in order to have such a connected component for $\alpha_{I',red}$, 
\begin{enumerate}
\item This connected component must consist of a sequence of forward steps which are not paired up with return steps.
\item The first vertex $u$ of this connected component must have first appeared as an unknown vertex in an $F$ step and then been the destination of the $R'$ step which was paired with this $F$ step.
\item The last vertex $v$ of this connected component must be followed by an $F$ step starting at $v$ and then an $R'$ step which is paired with this $F$ step. Note that this $R'$ step contains $v$.
\end{enumerate}
We now observe that 
\begin{enumerate}
\item Each group of $\tau_1,\ldots,\tau_j$ has a copy of this connected component.
\item For $I''$, each block of this connected component must be paired with a block in either an earlier group or a later group.
\item If the first block of this connected component is paired with an earlier block, the copy of the $F$ step where $u$ first appears is a slack block. 
\item If the last block of this connected component is paired with a later block, the copy of the later $R'$ step containing $v$ is a slack block.
\item If a block of the connected component is paired with a block in a later group and the next block of the connected component is paired with a block in an earlier group, both of these blocks are slack blocks.
\end{enumerate}
In all of these cases, there is at least one slack block.

Following a similar argument, each group of $\tau_j^T,\ldots,\tau_1^T$ has a copy of the transpose of this connected component which yields at least one slack block.
\end{proof}
\begin{corollary}
For each intersection pattern $I'$ which does not give one of the main terms of \cref{thm:backbonecorrectionproductapproximation} and all natural numbers $q$, letting $j' = 2q(j+1)$, we have
\[
\E\left[\Tr\left((M_{\gamma_{I'}}M_{\gamma_{T'}}^T)^q\right)\right] \leq \sum_{s=2q}^{j'}{2^{j'+2s}\binom{j'}{s}\left(j'D_V\right)^{sD_V}n^{1-\frac{1}{2}\lceil\frac{s}{24D_V}\rceil}}.
\]
\end{corollary}
\begin{proof}
This can be proved by applying \cref{thm:shapesequencetracebound} and observing that only intersection patterns $I''$ with at least $q$ slack blocks can be consistent with $I'$.
\end{proof}
\begin{corollary} \label{cor:quant-norm-bound}
For each intersection pattern $I'$ which does not give one of the main terms of \cref{thm:backbonecorrectionproductapproximation}, if $D_V \geq 4$, $j+1 \leq D_V$, and $n \geq (4D_V)^{400D_V^2}$ then with probability at least $1 - n2^{-n^{\frac{1}{400D_V^2}}}$,
$||M|| < 2n^{-\frac{1}{400D_V^2}}$.
\end{corollary}
\begin{proof}
By Markov's inequality, for all $q \in \mathbb{N}$,
\[
P\left(||M|| \geq \sqrt[2q]{\frac{E\left[\Tr\left((MM^T)^q\right)\right]}{\epsilon}}\right) \leq P\left(\Tr\left((MM^T)^q\right) \geq \frac{E\left[\Tr\left((MM^T)^q\right)\right]}{\epsilon}\right) \leq \epsilon.
\]
Take $\epsilon = n2^{-n^{\frac{1}{400D_V^2}}}$ and $q = \lceil{\frac{1}{2}\log_2(\frac{n}{\epsilon})}\rceil$. We make the following observations:
\begin{enumerate}
\item Since $j' \leq 2qD_V$ and $q \leq n^{\frac{1}{400D_V^2}}$, for all $s \geq 2q$,
\[
2^{j'+2s}\binom{j'}{s}\left(j'D_V\right)^{sD_V}n^{-\frac{1}{2}\lceil\frac{s}{24D_V}\rceil} \leq \frac{\left(16q^{2}D_V^{3}\right)^{sD_V}}{s!n^{\frac{s}{50D_V}}} \leq \frac{1}{s!n^{\frac{s}{200D_V}}}.
\]
\item $\E\left[\Tr\left((M_{\gamma_{I'}}M_{\gamma_{T'}}^T)^q\right)\right] \leq n\sum_{s=2q}^{j'}{\frac{1}{s!n^{\frac{s}{200D_V}}}} < n^{1-\frac{q}{200D_V}}$.
\item Since $\frac{n}{\epsilon} \leq 2^{2q}$, $\sqrt[2q]{\frac{E\left[\Tr\left((MM^T)^q\right)\right]}{\epsilon}} < 2n^{-\frac{1}{400D_V^2}}$.
\end{enumerate}
\end{proof}
To prove \cref{thm:backbonecorrectionproductapproximation}, we observe that in each of the cases in \cref{thm:backbonecorrectionproductapproximation}, there are at most $D_V^{D_V}$ intersection patterns $I'$ which do not give one of the main terms of \cref{thm:backbonecorrectionproductapproximation}. Using a union bound, with probability at least $1 - n{D_V^{D_V}}2^{-n^{\frac{1}{400D_V^2}}}$, the norm of the sum of these intersection terms is at most $2D_V^{D_V}n^{-\frac{1}{400D_V^2}}$.
\end{proof}
Note that in addition to the error terms addressed above, \cref{thm:backbonecorrectionproductapproximation} gives a collection of disconnected shapes $\al$ such that $\cmp_{\al}$ has constant norms and these terms are not captured by the Chebyshev polynomials. We now observe that the corrections for these terms have small norms. 
\begin{claim}[Disconnected Shapes have Small-Norm Corrections]
\label{clm:split-boundary-correction}
For all shapes $\alpha$ of the form $\alpha = \tau_1 \circ \ldots \circ \tau_i | \tau_{i+1} \circ \ldots \circ \tau_j$ such that $|V(\al)| \leq D_V$ and all $\epsilon > 0$, with probability at least $1 - \epsilon$,
\[
||\frac{1}{\sqrt{n}}\cmp_{\bar{\al}}|| < \frac{1}{\sqrt{n}}\left(10\sqrt{ln\left(\frac{n}{\epsilon}\right)}D_V\right)^{D_V}.
\]
\end{claim}
\begin{remark}
Note that the correction for $\cmp_{\alpha} = n^{-\frac{|V(\al)|}{2}}\cm_{\al}$ is $\frac{1}{\sqrt{n}}\cmp_{\bar{\al}} = n^{-\frac{|V(\al)|}{2}}\cm_{\bar{\al}}$ as there is a path from $u$ to $v$ in $\bar{\al}$ but not in $\al$.
\end{remark}
\begin{proof}
By Theorem 6.1 of \cite{AMP20}, for all $\epsilon > 0$, with probability at least $1 - \epsilon$, 
\[
||\cm_{\bar{\al}}|| < \left(10\sqrt{\frac{n}{\epsilon}}|V(\bar{\al})|\right)^{|V(\bar{\al})|}n^{\frac{|V(\bar{\al})| - 1}{2}} \leq \left(10\sqrt{\ln\left(\frac{n}{\epsilon}\right)}D_V\right)^{D_V}n^{\frac{|V(\bar{\al})| - 1}{2}}
\]
The result follows as $\tilde{M}_{\bar{\al}} = n^{-\frac{|V(\bar{\al})| - 1}{2}}M_{\bar{\al}}$.
\end{proof}

\subsection{Equivalence of Chebyshev Polynomials and Concatenated Shapes}

In this subsection, we use the properties established above for products of backbone-correction shapes to complete the proof of the equivalence between Chebyshev polynomials and concatenated shapes.

 \begin{theorem}
 [Formal Equivalence Between $P_j(M)$ and Concatenated Shapes]
 \label{thm:formal-equiv-cheby-shape}
Given a symmetric linear combination of backbone-correction shapes $M$ from~\cref{def:sym-lin-comb} with normalized variance $\Var(M)=1$,
 for every $j\geq 1$, 
 \[
 P_j(M) = \fp_j(M)+\mathrm{Error}(j),
 \]
 where $\mathrm{Error}(j)$ admits a small-norm correction. More precisely, there exists a matrix $\mathrm{CorrectionError}(j)$ such that 
 $\mathrm{CorrectionError}(j) +  \mathrm{Error}(j)$ 
  satisfies the constraints for $R_W$ in \cref{def:concrete-target} and if $D_V \geq 4$ and $n \geq (4D_V)^{400D_V^2}$ then with probability at least $1- n{D_V^{5D_V}}2^{-n^{\frac{1}{400D_V^2}}}$,
  \[
 \| \mathrm{CorrectionError}(j)\|_{sp} \leq (4B)^{j}\mathsf{Slack}(n,D_V)
 \] 
     where 
     $B\coloneqq\sum_{\tau\in\calB}|c(\tau)|$ and $\mathsf{Slack}(n,D_V) \coloneqq D_V^{D_V+1}n^{-\frac{1}{400D_V^2}}$.
  \end{theorem}
  
\begin{proof}
We use the recurrence relation $P_{j}(M) = P_{j-1}(M)M - P_{j-2}(M)$ and make the following observations:
\begin{enumerate}
\item As discussed in \cref{sec:graph-mat-cheb}, the connected constant norms which don't cancel with each other are precisely the $j$-wise concatenations of terms in $M$.
\item By \cref{thm:backbonecorrectionproductapproximation}, each constant norm term in $P_{j-1}(M)$ gives at most $3$ constant norm terms when multiplied by $M$. Using this, it is easy to show by induction that the sum of the magnitudes of the constant norm terms in $P_j(M)$ is at most $(4B)^j$.
\item By \cref{cor:linear-combination-norm-bound}, $||M||$ is $(2+o_n(1))$ with high probability. More precisely, with probability at least $1 -  n2^{-n^{\frac{1}{200D_V^2}}}$, $||M|| < 2\left(1 + n^{-\frac{1}{400D_V^2}}\right)$. As long as this bound holds, whenever we multiply a smaller order term by $M$, this multiplies its norm by at most $3$.
\end{enumerate}
We take $CorrectionError(j)$ to be the corrections for the constant norm terms in $P_j(M)$ minus the lower order terms of $M$. Observe that 
\begin{enumerate}
\item 
By \cref{clm:split-boundary-correction}, for each constant norm term $\cmp_{\al}$ where $\al$ is disconnected, with probability at least $1 - n2^{-n^{\frac{1}{400D_V^2}}}$, 
$
||\cmp_{\bar{\al}}|| < \frac{n^{\frac{1}{800D}}}{\sqrt{n}}\left(10\sqrt{ln(2)}D_V\right)^{D_V} < \frac{1}{\sqrt[4]{n}}
$.

As long as these bounds hold for all correction terms $\cmp_{\bar{\al}}$, the norm of the sum of these correction terms is at most $\frac{(4B)^j}{\sqrt[4]{n}}$.
\item By \cref{thm:backbonecorrectionproductapproximation}, whenever we multiply a constant norm term in $P_{j'-1}(M)$ by a term in $M$, with probability at least $1 - n{D_V^{D_V}}2^{-n^{\frac{1}{400D_V^2}}}$, the norm of the sum of the smaller order terms is at most $2D_V^{D_V}n^{-\frac{1}{400D_V^2}}$. If this bound holds for all such products, the norm of the sum of the error terms is at most 
$
\sum_{j'=2}^{j}{4^{j-j'}(4B)^{j'-1}}2BD_V^{D_V}n^{-\frac{1}{400D_V^2}} \leq (j-1)(4B)^{j}D_V^{D_V}n^{-\frac{1}{400D_V^2}}
$.

To see this, consider error terms which are obtained by taking lower order terms of  $P_{j'-1}(M)M$ and then potentially multiplying by $M$ some number of times later on. Observe that the sum of the magnitudes of the coeffecients of the constant norm terms in $P_{j'-1}(M)$ is at most $(4B)^{j'-1}$, the sum of the magnitudes of the coefficients of the terms in $M$ is $B$, and if we iteratively apply the recurrence relation $P_j(M) = P_{j-1}(M)M - P_{j-2}(M)$, each step either reduces $j$ by $1$ and multiplies by $M$ (which multiplies the norm of an error term by at most $3$) or reduces $j$ by $2$ so these potential later multiplications give a factor of at most $4^{j-j'}$.
\end{enumerate}
Putting everything together, if all of the above norm bounds hold then $||CorrectionError(j)|| \leq (4B)^{j}D_V^{D_V+1}n^{-\frac{1}{400D_V^2}}$. Finally, to bound the probability that one of our norm bounds fails, we bound the number of cases we need to consider.
\begin{proposition}
There are at most $(D_V)^{4D_V}$ possible pairs of the form $(\al,\tau_{j'})$ where $\al = \tau_1 \circ \ldots \circ \tau_{j'-1}$ or $\al = \tau_! \circ \ldots \circ \tau_i | \tau_{i+1} \circ \ldots \circ \tau_{j'-1}$ and $j' \leq j$.
\end{proposition}
\begin{proof}
Since $\al \circ \tau_{j'}$ is outerplanar, $\al \circ \tau_{j'}$ has at most $2|V(\al \circ  \tau_{j'})| - 3 \leq 2D_v$ edges. To specify $\al \circ \tau_{j'}$ as well as how the vertices are split up among the components $\tau_1,\ldots,\tau_{j'}$, we can specify the number of vertices in $\al \circ \tau_{j'}$, order the vertices according to the backbones of the components, specify the edges, and then go through the vertices one by one and specify whether they are a boundary vertex or not.
\end{proof}
Taking a union bound, the probability that one or more of our norm bounds fail is at most $n{D_V^{5D_V}}2^{-n^{\frac{1}{400D_V^2}}}$.
\end{proof}

%% file: appendix.tex
\section{Equivalence of Lower and Upper Bounds for the Theta Function on \texorpdfstring{$G(n,\frac{1}{2})$}{G(n,1/2)}}
\label{app:dual-lowerbound}
\begin{proposition}[From $M_G$ to $\pE$ for Maximum Clique]
\label{prop:ind-set-from-dual}
For any graph $G$, suppose $M_G$ is dual-feasible with objective value $\lambda$, namely
\[
M_G[i,j]=1 \qquad \text{for } i=j \text{ or } (i,j)\notin E(G),
\]
and
\[
\lambda I_n - M_G \succeq 0.
\]
Then there exists a valid degree-$2$ Sum-of-Squares pseudo-expectation $\pE'$ for the Maximum Clique problem on $G$ with value
$
\frac{n}{\lambda}.
$
Moreover, $\pE'$ assigns uniform weights across vertices.
\end{proposition}

Notably, the converse direction does not hold in general: it is not always the case that a valid degree-$2$ SoS moment matrix $\pE$ for Max-Clique or Max-Independent-Set can be converted into a dual witness for the $\vartheta$ function (\cref{def:dual-sdp}), while such a transformation exists for $\pE$ assigning uniform weights across vertices (e.g., $\pE[x_i]=k/n$ for all $i$ as below).

In particular, although it is known that $\vartheta(G)\ge \sqrt{n}$ via a non-explicit argument (corresponding to a degree-$2$ SoS solution of value at least $\sqrt{n}$), such a solution need not satisfy the symmetry condition above, and therefore need not yield a dual witness certifying an upper bound of $\sqrt{n}$ via this route.

From this perspective, we view these two problems as searching for equivalent matrices as we further constrain our pseudo-expectation to satisfy the extra symmetry conditions. Therefore, it is still interesting for us to obtain a lower bound of \emph{only} $(1-o_n(1))\cdot \sqrt{n}$ via an explicit construction despite the known bound of $\sqrt{n}$, since our solution crucially \emph{can} be transformed into a dual witness that additionally certifies an upper bound of $(1+o_n(1)) \cdot \sqrt{n}$.

\begin{proof}[Proof of \cref{prop:ind-set-from-dual}]
Let
\[
k \coloneqq \frac{n}{\lambda},
\qquad\text{and}\qquad
Q \coloneqq \frac{1}{\lambda}(\lambda I_n - M_G).
\]
Since $\lambda I_n-M_G \succeq 0$, we have $Q \succeq 0$.

Next, define
\[
L \coloneqq \frac{k}{n}L_{1,0} + \sqrt{\frac{k}{n}}\,L_{1,1} + L_{0,0},
\]
where these are $(n+1)\times(n+1)$ matrices with nonzero entries given by
\begin{align*}
L_{1,0}[i,\emptyset] &= 1 \qquad &&\text{for } i\in[n],\\
L_{1,1}[i,i] &= 1 \qquad &&\text{for } i\in[n],\\
L_{0,0}[\emptyset,\emptyset] &= 1.
\end{align*}
We define the degree-$2$ moment matrix by
\[
\pE' \;:=\; L
\begin{pmatrix}
1 & 0\\
0 & Q
\end{pmatrix}
L^\top .
\]
By construction, this matrix is positive semidefinite.

We now verify that $\pE'$ is a feasible degree-$2$ pseudo-expectation for Maximum Clique on $G$ and compute its value.

\paragraph{Moment entries.}
A direct expansion gives
\[
\pE'[i,\emptyset]=\pE'[\emptyset,i]=\frac{k}{n}
\qquad\text{for all } i\in[n],
\]
and
\[
\pE'[i,j]
=
\left(\frac{k}{n}\right)^2
+
\frac{k}{n}\,Q[i,j]
\qquad\text{for all } i,j\in[n].
\]
Since
\[
Q[i,j]=\delta_{ij}-\frac1\lambda M_G[i,j],
\]
we may rewrite this as
\[
\pE'[i,j]
=
\left(\frac{k}{n}\right)^2
+
\frac{k}{n}\delta_{ij}
-
\frac{k}{n}\cdot \frac1\lambda M_G[i,j].
\]

\paragraph{Booleanity and symmetry.}
Symmetry is immediate since $\pE'$ is a symmetric matrix. For the diagonal entries, using $M_G[i,i]=1$ and $k/n=1/\lambda$, we obtain
\[
\pE'[i,i]
=
\left(\frac{k}{n}\right)^2
+
\frac{k}{n}\left(1-\frac1\lambda\right)
=
\left(\frac{k}{n}\right)^2 + \frac{k}{n} - \left(\frac{k}{n}\right)^2
=
\frac{k}{n}
=
\pE'[i,\emptyset].
\]
Hence $\pE'[x_i^2]=\pE'[x_i]$ for every $i$, as required in the degree-$2$ SoS relaxation.

\paragraph{Clique constraints.}
For $i\neq j$, we have
\[
\pE'[i,j]
=
\left(\frac{k}{n}\right)^2 - \frac{k}{n}\cdot \frac1\lambda M_G[i,j].
\]
Since $k/n=1/\lambda$, this simplifies to
\[
\pE'[i,j]
=
\left(\frac{k}{n}\right)^2 \bigl(1-M_G[i,j]\bigr).
\]
If $(i,j)\notin E(G)$, then dual feasibility gives $M_G[i,j]=1$, and therefore
\[
\pE'[x_i x_j]=0.
\]
Thus the clique constraints are satisfied.

\paragraph{Objective value.}
Finally,
\[
\sum_{i=1}^n \pE'[x_i]
=
\sum_{i=1}^n \pE'[i,\emptyset]
=
n\cdot \frac{k}{n}
=
k
=
\frac{n}{\lambda}.
\]

Therefore $\pE'$ is a valid degree-$2$ Sum-of-Squares pseudo-expectation for Maximum Clique on $G$ with value $n/\lambda$.
\end{proof}

\begin{claim}[Primal Lower Bound from an Upper Bound Witness]\label{claim:lower-bound-from-upper}  Let $M_G$ be a matrix that w.h.p. satisfies the dual SDP requirements in~\cref{def:dual-sdp} with objective value $\lambda$. Then we have $\vartheta(G)\geq \frac{n}{\lambda}   $ w.h.p.  \end{claim}
\begin{proof}[Proof of \cref{claim:lower-bound-from-upper}]
	
The proof follows in two steps. We first observe that the matrix $M_G$ can be transformed into a degree-$2$ SoS lower bound for maximum clique on $G$ via \cref{prop:ind-set-from-dual}.

	Since a pseudo-expectation operator for maximum clique on $G$ is at the same time a pseudo-expectation for maximum independent set on the complement graph $\bar{G}$ with the same objective value, and $M_G$ is dual-feasible with high probability for $G\sim G(n,\frac{1}{2})$,
	$\pE[\mathsf{IndSet}(\bar{G})] \geq \frac{n}{\lambda}  $ with high probability. 
	
	Finally, since $G(n,\frac{1}{2})$ is invariant under complementation, i.e., graph samples $G$ and $\bar{G}$ appear with the same probability, this also implies the pseudo-expectation has independent-set value $\pE[\mathsf{IndSet}(G)] \geq \frac{n}{\lambda}  $
	with high probability for $G\sim G(n,\frac{1}{2})$, showing $ 
	\vartheta(G)\geq \frac{n}{\lambda}
$ via the equivalent formulation of the theta function as the degree-$2$ SoS relaxation for maximum independent set.
\end{proof}

\input{theta-function/formal_dual_veritifcation.tex}

\section{Deferred Details for Our Iterative Procedure}

\input{recurrence.tex}

%% file: theta-function/formal_dual_veritifcation.tex
\section{Formal Verification of the Dual Witness Matrix}\label{sec:final-parameter}
\begin{lemma}[Verification of Non-Edge Constraints]\label{lem:edge-constraints-verification}
For any $\eps>0$, the matrix
	    \[
M_{\mathrm{final}}
\coloneqq
\eps I+F_D(\widetilde{Q})
+\chebyerror+\mathsf{EarlyTerm}
\]
satisfies the non-edge constraints of \cref{def:concrete-target} after normalizing its diagonal.
\end{lemma}
Here the corrections are:
\begin{enumerate}
	\item $\chebyerror(\widetilde{Q})\coloneqq\sum_{j=1}^D b_j\,\mathrm{CorrectionError}(j)$, using the corrections from \cref{thm:formal-equiv-cheby-shape} for the Chebyshev expansion of $F_D$;
	\item $\mathsf{EarlyTerm}(\widetilde{Q})$, the correction for terminating at $Q_{t^*}$, whose norm is bounded in \cref{lem:corr-small}.
\end{enumerate}

\begin{proof}
	Following our notation for the iterative process, define $\textsf{Old-}\calP_j(\widetilde{Q})$ and $\textsf{New-}\calP_j(\widetilde{Q})$ to be the collections of proper $j$-way concatenations that do not involve shapes from the final update at $Q_{t^*}$, and those that do, respectively.
	\begin{align*}
F_D(\widetilde{Q})
&=C_FI+b_1\widetilde{Q}+\sum_{j=2}^D b_j\fp_j(\widetilde{Q})
  +\sum_{j=1}^D b_j\bigl(P_j(\widetilde{Q})-\fp_j(\widetilde{Q})\bigr)\\
&=\underbrace{C_FI+b_1\widetilde{Q}
  +\sum_{j=2}^D b_j\fp_j(\textsf{Old-}\widetilde{Q})}_{A}
  +\underbrace{\sum_{j=2}^D b_j\fp_j(\textsf{New-}\widetilde{Q})}_{B}\\
&\qquad+\underbrace{\sum_{j=1}^D b_j\bigl(P_j(\widetilde{Q})-\fp_j(\widetilde{Q})\bigr)}_{E}.
\end{align*}
	We discuss these three terms separately. By the construction of our iterative process, $A$ satisfies the non-edge constraint, as any term in $\tilde{Q}$ is a backbone-correction shape for some $j$-way concatenation using shapes prior to the final update.
	
	The correction $\mathsf{EarlyTerm}(\widetilde{Q})$ handles $B$; its norm is bounded in \cref{lem:corr-small} using the truncated-variance estimate of \cref{lem:truncated-var}.
	
	For the third term, \cref{thm:formal-equiv-cheby-shape} shows that $E+\chebyerror$ vanishes on the diagonal and nonedges.
Thus, using $b_1=1$ and the initialization $\widetilde{Q}_0=(1+c_\gamma)C_F\cmp_{\fp_1}$, the corrected sum has the form
\[
F_D(\widetilde{Q})+\chebyerror+\mathsf{EarlyTerm}
=C_FI+\frac{(1+c_\gamma)C_F}{\sqrt n}A_G+R,
\]
where $R$ is symmetric and vanishes on the diagonal and nonedges. Consequently, $M_{\mathrm{final}}$ has diagonal $C_F+\eps$, and division by $C_F+\eps$ gives the required non-edge constraints.
	\end{proof}
	
\begin{restatable}{lemma}{finalcorrsmall}
\label{lem:corr-small}
For $2\leq t^*,D\leq O(\log\log n/\log\log\log n)$ with a sufficiently small implied constant, the early-termination correction satisfies, with high probability,
\[
\|\mathsf{EarlyTerm}(\widetilde{Q})\| = O\!\left(D^{-3/2}+(t^*)^{-3/2}\right).
\]
In particular, this is $O(1/t^*)$ when $D=t^*$.
\end{restatable}
\begin{proof}
Write $S_i=\sum_{\tau\in\calB(\widetilde{Q}_i)}c(\tau)^2$. The exact-normalization and early-termination estimate in \cref{lem:truncated-var} gives
\[
\Var(\mathsf{EarlyTerm}(\widetilde{Q}))
=\sum_{j=2}^D b_j^2\bigl(S_{t^*}^j-S_{t^*-1}^j\bigr)
=O\!\left(D^{-3}+(t^*)^{-3}\right).
\]
Every added shape is a backbone-correction shape within the size bound, so \cref{cor:linear-combination-norm-bound} gives the asserted norm estimate.
\end{proof}

\paragraph{Picking the Final Parameters}
\FinalParamChoiceSpec*
\begin{proof}
We take $t^*=D$ and use the slack function from \cref{thm:formal-equiv-cheby-shape} with $\epsilon=n^{-1}$. Set
\[
B\coloneqq\sum_{\tau\in\calB(\widetilde{Q})}|c(\tau)|,
\qquad D_V\coloneqq2D^{D+1}+4.
\]
The truncation in \cref{def:inner-trunc} gives base shapes with at most $D^D+1$ vertices, so $D_V$ bounds all products and corrections through degree $D$. Variance normalization and \cref{lem:backbone-count} give
\[
B\leq\sqrt{|\calB(\widetilde{Q})|}\leq3^{D^D+1}.
\]
We work on the common event of the norm estimates below; their parameter conditions are verified at the end.

\paragraph{Bounding the Negative Spectrum of $F_D$}
By \cref{cor:linear-combination-norm-bound}, since $\Var(\widetilde{Q})=1$,
\[
\|\widetilde{Q}\|
\leq2+2\frac{\log_2(n^2)}{n^{1/(200D_V^2)}}
\leq2+2\mathsf{Slack}(n,D_V,n^{-1}).
\]
Taking $\delta_{sp}=2\mathsf{Slack}(n,D_V,n^{-1})$, \cref{prop:FD-psd} gives
\[
F_D(\widetilde{Q})
\succeq-\tilde{O}\!\left(\frac1D+
\delta_{sp}\poly(D)\exp(D\sqrt{\delta_{sp}})\right)I.
\]

\paragraph{Final Correction Norm}
By \cref{thm:formal-equiv-cheby-shape}, summing over $D$ levels of expansion gives us
\begin{align*}
\|\chebyerror\|
&\leq(O(B\cdot D))^D\mathsf{Slack}(n,D_V,n^{-1}),
\end{align*}
On the other hand, \cref{lem:corr-small} gives 
\[
\|\mathsf{EarlyTerm}\| = O(1/t^*)\,.
\]

\paragraph{Summary}
Choose
\[
t^*=D=\left\lfloor c\frac{\log\log n}{\log\log\log n}\right\rfloor
\]
for a fixed constant $c>0$. Then $D_V=(\log n)^{c+o(1)}$ and $\log B=O(D^D)$. The slack function satisfies
\[
\log\mathsf{Slack}(n,D_V,n^{-1})
=-\frac{\log n}{400D_V^2}
+D_V\log D_V+\log\log_2(n^2).
\]
Since $D\log(B\cdot D)\leq(\log n)^{c}$ for some constant $c>0$, it follows that
\[
(O(BD))^D\mathsf{Slack}(n,D_V,n^{-1})
=\exp\!\left(- (\log n)^{1-2c+o(1)}\right)=o(1/D).
\]
This choice satisfies the size and probability conditions of \cref{cor:linear-combination-norm-bound,thm:backbonecorrectionproductapproximation}. Moreover, $\delta_{sp}=\exp(-(\log n)^{1-2c+o(1)})$, so $D\sqrt{\delta_{sp}}=o(1)$ and
\[
\delta_{sp}\poly(D)\exp(D\sqrt{\delta_{sp}})
=\delta_{sp}\poly(D)(1+o(1))=o(1/D).
\]
Therefore,
\[
F_D(\widetilde{Q})\succeq-\tilde{O}(1/D)I,
\qquad
\|\chebyerror+\mathsf{EarlyTerm}\|=O(1/D).
\]
Consequently, for some $\eta=\tilde{O}(1/D)$,
\[
F_D(\widetilde{Q})+\chebyerror+\mathsf{EarlyTerm}\succeq-\eta I.
\]
Taking $\eps_{\mathrm{spec}}=\tilde{\Theta}(1/D)\geq\eta$ with a sufficiently large implied constant gives
\[
M_{\mathrm{final}}\succeq(\eps_{\mathrm{spec}}-\eta)I\succeq0,
\]
as required.
\end{proof}

%% file: recurrence.tex
\subsection{Evolution of the Variance of \texorpdfstring{$Q_i$}{Q\_i}}
\begin{claim}[Convergence without Truncation] \label{claim:variance-evolution-recurrence}
Define
\[
S_i \;\coloneqq\; \sum_{\tau \in \calB(Q_i)} c(\tau)^2.
\]
Then
\[
\lim_{i\to\infty} S_i = 1.
\]
\end{claim}

\begin{proof}
We prove this in two steps:\begin{enumerate}
	\item Identify a recurrence relation satisfied by $\{S_i\}$;
	\item Show the recurrence converges to $1$.
\end{enumerate}

\paragraph{Identifying the Recurrence}
By construction, when passing from $Q_i$ to $Q_{i+1}$, the newly added shapes are obtained from
$
\textsf{New-}P_j(Q_i)
$
for $j\ge 2$, i.e., the $j$-way concatenations that use at least one term from $\calB(Q_i)\setminus \calB(Q_{i-1})$. Each such term $\tau$ has coefficient
$
c(\tau) = b_j \prod_{t=1}^j c(\tau_t).
$
Therefore, the total squared mass added at level $i+1$ is
\[
S_{i+1} - S_i
=
\sum_{j\ge 2} b_j^2
\sum_{\tau\in \textsf{New-}P_j(Q_i)} \prod_{t=1}^j c(\tau_t)^2.
\]

For fixed $j$, summing over all $j$-way concatenations from $\calB(Q_i)$ gives $S_i^j$, while those using only $\calB(Q_{i-1})$ give $S_{i-1}^j$. Hence
\[
\sum_{\tau\in \textsf{New-}P_j(Q_i)} \prod_{t=1}^j c(\tau_t)^2
=
S_i^j - S_{i-1}^j.
\]
Thus
\[
S_{i+1} - S_i
=
\sum_{j\ge 2} b_j^2 \bigl(S_i^j - S_{i-1}^j\bigr).
\]
Rearranging gives the closed recursion
\[
S_{i+1}
=
C_F^2 + \sum_{j\ge 2} b_j^2 S_i^j
\;=:\;
\Phi(S_i),
\]
where $S_0 = C_F^2$.

\paragraph{Convergence of the Recurrence }
We show that $(S_i)$ is monotone and bounded by $1$. Monotonicity follows directly from the definition of $Q_i$: the ground set $\calB(Q_i)$ is increasing, and every term in $Q_i$ appears in $Q_{i+1}$ with the same coefficient, hence
\[
S_{i+1}
=
\sum_{\tau\in \calB(Q_{i+1})} c(\tau)^2
\;\ge\;
\sum_{\tau\in \calB(Q_i)} c(\tau)^2
=
S_i.
\]
For boundedness, applying Parseval to
\(
F(x)=C_F + x + \sum_{j\ge2} b_j P_j(x)
\)
under the semicircle law gives
\[
C_F^2 + \sum_{j\ge2} b_j^2 = 1.
\]
Using the recursion
\[
S_{i+1} = C_F^2 + \sum_{j\ge2} b_j^2 S_i^j,
\]
we prove by induction that $S_i \le 1$: if $S_i \le 1$, then $S_i^j \le 1$ for all $j\ge2$, and thus
\[
S_{i+1}
\le
C_F^2 + \sum_{j\ge2} b_j^2
=
1.
\]
To prove convergence to $1$, we compute a contraction constant. The coefficient formula in the proof of \cref{lem:bi-square-tail} gives $b_j=0$ for odd $j\ge3$ and
\[
b_j^2=\frac{64}{\pi^2(j-1)^2(j+3)^2}
\qquad\text{for even }j\ge2.
\]
The partial-fraction identity
\[
\frac{j}{(j-1)^2(j+3)^2}
=\frac{1}{32}\left(\frac{1}{j-1}-\frac{1}{j+3}\right)
+\frac{1}{16(j-1)^2}-\frac{3}{16(j+3)^2}
\]
and $\sum_{m\ge1}(2m-1)^{-2}=\pi^2/8$ give
\begin{equation*}\label{eq:variance-contraction-constant}
\beta\coloneqq\sum_{j\ge2}j\,b_j^2
=\frac{64}{\pi^2}
\sum_{\substack{j\ge2\\j\text{ even}}}
\frac{j}{(j-1)^2(j+3)^2}
=\frac{16}{\pi^2}-1<1.
\end{equation*}
Since $0\le S_i\le1$ and $1-x^j\le j(1-x)$ for $x\in[0,1]$, we have
\[
0\le1-S_{i+1}
=\sum_{j\ge2}b_j^2(1-S_i^j)
\le\beta(1-S_i).
\]
Iterating yields $0\le1-S_i\le\beta^i(1-C_F^2)\to0$, proving that $S_i\to1$.
\end{proof}

\subsection{Quantitative Convergence of the Variance of \texorpdfstring{$Q_i$}{Q\_i}}

\begin{lemma}[Quantitative convergence of the variance process] \label{lem:convergence_variance}
Let
\[
S_i \coloneqq \sum_{\tau\in\calB(Q_i)} c(\tau)^2,
\qquad
\Phi(x)\coloneqq S_0+\sum_{j\ge2} b_j^2 x^j,
\]
so that \(S_{i+1}=\Phi(S_i)\) and \(S_0=C_F^2\).
Define
\(
\beta \coloneqq \sum_{j\ge2} j\,b_j^2.
\)
Then for every \(t\ge0\),
\[
1-S_{t+1}
=
\sum_{j\ge2} b_j^2\bigl(1-S_t^j\bigr),
\]
and consequently
\[
(1-C_F^2)^{t+1}
\;\le\;
1-S_t
\;\le\;
\beta^t(1-C_F^2).
\]
There exists some $\gamma<1$ such that
\[
1-S_t \le \gam^t(1-C_F^2),
\]
and hence \(S_t\to 1\) exponentially fast.
\end{lemma}

\begin{proof}
Set \(E_t\coloneqq 1-S_t\). Since
\[
S_{t+1}=C_F^2+\sum_{j\ge2} b_j^2 S_t^j
\quad\text{and}\quad
C_F^2+\sum_{j\ge2} b_j^2=1,
\]
we have
\[
E_{t+1}
=
1-S_{t+1}
=
\sum_{j\ge2} b_j^2(1-S_t^j).
\]
For \(x\in[0,1]\),
\[
1-x^j=(1-x)(1+x+\cdots+x^{j-1})\le j(1-x),
\]
hence
\[
E_{t+1}
\le
\sum_{j\ge2} j\,b_j^2 (1-S_t)
=
\beta E_t.
\]
Iterating gives
\[
E_t \le \beta^t E_0=\beta^t(1-C_F^2).
\]

Similarly, since \(1-x^j\ge 1-x\) for \(x\in[0,1]\), we get
\[
E_{t+1}
\ge
\sum_{j\ge2} b_j^2(1-S_t)
=
(1-C_F^2)E_t.
\]
Iterating yields
\[
E_t \ge (1-C_F^2)^t E_0=(1-C_F^2)^{t+1}.
\]

Finally, \eqref{eq:variance-contraction-constant} gives $\beta=16/\pi^2-1<1$, so we may take $\gamma=\beta$.
This proves the claim.

\end{proof}

%% file: polyapprox.tex
\section{Deferred Details for Truncation}
\subsection{Spectral Impact of Outer Truncation of \texorpdfstring{$F_D$}{F\_D}}
Now we analyze the spectral impact of the outer truncation of $F(\widetilde{Q})$.
Throughout this section, $D\ge 1$ is an integer and
\[
F_D(x)=\sum_{j=0}^D b_j P_j(x),
\qquad P_j(x)=U_j(x/2),
\qquad b_0=C_F,\quad b_1=1.
\]
\begin{restatable}{lemma}{FDwithinsupport}
\label{lem:F_D-witin-support}
For any $x\in [-2,2]$, we have
\[
|F_D(x) - F(x)| \le \tilde{O}\!\left(\frac{1}{D}\right).
\]
\end{restatable}

\begin{restatable}{lemma}{FDedgedeviation}
\label{lem:edge-deviation}
For any $\delta \ge 0$, we have
\[
|F_D(2+\delta) - F_D(2)|
\le
O\!\big( \delta \cdot \poly(D) \cdot  \exp(D\sqrt{\delta})\big),
\]
and similarly,
\[
|F_D(-2-\delta) - F_D(-2)|
\le
O\!\big(\delta \cdot \poly(D) \exp(D\sqrt{\delta})\big).
\]
\end{restatable}

\begin{restatable}[Approximate PSDness of $F_D(\widetilde{Q})$]{proposition}{FDpsd}
\label{prop:FD-psd}
Suppose $\mathsf{spec}(\widetilde{Q}) \subseteq [-2-\delta,\,2+\delta]$. Then
\[
F_{D}(\widetilde{Q})
\succeq
- \tilde{O}\!\left(\frac{1}{D} + \delta \cdot \poly(D)\cdot \exp(D\sqrt{\delta})\right)\cdot I.
\]\end{restatable}
\begin{proof}
By~\cref{lem:F_D-witin-support},  $F_D(x)\ge -\tilde{O}(1/D)$ for all $x\in[-2,2]$. For $x$ outside this interval, \cref{lem:edge-deviation} shows that the deviation from the boundary values $F_D(\pm 2)$ is at most $O( \delta \cdot \poly(D)\cdot  \exp(D\sqrt{\delta}))$. Since $F_D(\pm 2)\ge -\tilde{O}(1/D)$, we obtain
\[
F_D(x)\ge -\tilde{O}\!\left(\frac{1}{D} + \delta \cdot \poly(D)\cdot\exp(D\sqrt{\delta})\right)
\]
for all $x\in[-2-\delta,\,2+\delta]$. The claim then follows by applying functional calculus to the spectrum of $\widetilde{Q}$.
\end{proof}

\subsection{Truncation Error within the Support}
\FDwithinsupport*

\begin{proof}[Proof of~\cref{lem:F_D-witin-support}]
We prove the stronger bound $O(1/D)$ directly in the second-kind basis.
The coefficient calculation in the proof of \cref{lem:bi-square-tail} gives, for $k\ge 1$,
\[
b_{2k}=(-1)^{k-1}a_k,\qquad
a_k=\frac{8}{\pi(2k-1)(2k+3)},\qquad b_{2k+1}=0.
\]
In particular, $a_k=O(k^{-2})$ and $a_k-a_{k+1}=O(k^{-3})$.
Writing $x=2\cos\theta$, we have
\[
\|P_{2k}\|_\infty\le 2k+1,\qquad
P_{2k+2}(2\cos\theta)-P_{2k}(2\cos\theta)
=2\cos((2k+2)\theta),
\]
where all supremum norms are on $[-2,2]$ and the identities extend to the endpoints by continuity.
Pairing consecutive nonzero terms, we obtain
\begin{align*}
\|b_{2k}P_{2k}+b_{2k+2}P_{2k+2}\|_\infty
&=\|(a_k-a_{k+1})P_{2k}
-a_{k+1}(P_{2k+2}-P_{2k})\|_\infty\\
&\le (2k+1)|a_k-a_{k+1}|+2a_{k+1}
=O(k^{-2}).
\end{align*}
The sum of these paired norms converges, and an unpaired terminal term has norm
$a_k\|P_{2k}\|_\infty=O(k^{-1})\to 0$.
Thus the ordinary partial sums converge uniformly. Their limit is $F$, since they also converge to $F$ in $L^2$ of the semicircle measure and both functions are continuous.

For $m=\lfloor D/2\rfloor$, pair the tail starting at $k=m+1$. The bound above gives
\[
\|F_D-F\|_{L^\infty([-2,2])}
=O\!\left(\sum_{r=0}^{\infty}\frac{1}{(m+1+2r)^2}\right)
=O\!\left(\frac{1}{m+1}\right)
=O\!\left(\frac{1}{D}\right).
\]
This also covers $D=1$.
\end{proof}

\subsection{Deviation at the Edge}

\begin{lemma} \label{lem:edge-deviation-detailed}
Suppose
\[
F_D(x)=\sum_{j=0}^D b_j\,U_j(x/2),
\qquad
B_D \coloneqq \sum_{j=0}^D |b_j|.
\]
Then for every $\delta\ge 0$,
\[
|F_D(2+\delta)-F_D(2)|
\le
B_D\,\delta D^2(D+1) e^{D\sqrt{\delta}},
\]
and similarly,
\[
|F_D(-2-\delta)-F_D(-2)|
\le
B_D\,\delta D^2(D+1) e^{D\sqrt{\delta}}.
\]
\end{lemma}

\begin{proof}
We prove the bound at the positive edge; the negative case is identical.

Let
\[
y = 1+\frac{\delta}{2} = \cosh t,
\qquad
t = \arccosh\!\left(1+\frac{\delta}{2}\right)\ge 0.
\]
Then
\[
U_j(y)=\sum_{r=0}^j\cosh((j-2r)t),
\qquad U_j(1)=j+1.
\]
Using $\cosh u - 1 \le u^2 e^u$ for $u\ge 0$ and $\cosh t \ge 1+t^2/2$, we get
\[
|U_j(y)-(j+1)|
\le (j+1)j^2 t^2 e^{jt}
\le (j+1)j^2 \delta\, e^{j\sqrt{\delta}}
\le (j+1)j^2 \delta\, e^{D\sqrt{\delta}}.
\]
Hence
\begin{align*}
|F_D(2+\delta)-F_D(2)|
&=
\left|
\sum_{j=0}^D b_j\bigl(U_j(1+\delta/2)-(j+1)\bigr)
\right| \\
&\le
\sum_{j=0}^D |b_j|\,(j+1)j^2 \delta\, e^{D\sqrt{\delta}} \\
&\le
B_D\,\delta D^2(D+1) e^{D\sqrt{\delta}}.
\end{align*}

For the negative edge, use $U_j(-x)=(-1)^jU_j(x)$ to reduce to the same bound.
\end{proof}

We are now ready to prove the edge-deviation bound in~\cref{lem:edge-deviation}.
\begin{proof}[Proof of \cref{lem:edge-deviation}]
Since the second-kind polynomials $P_j$ are orthonormal for the semicircle measure $\mu$, Parseval's identity gives
\[
\sum_{j\ge 0} b_j^2=\int_{-2}^2 F(x)^2\,d\mu(x)=2.
\]
By Cauchy--Schwarz,
\[
B_D=\sum_{j=0}^D |b_j|
\le \sqrt{D+1}\left(\sum_{j=0}^D b_j^2\right)^{1/2}
\le \sqrt{2(D+1)}.
\]
Substituting into \cref{lem:edge-deviation-detailed} bounds each edge deviation by
\[
O\!\left(\delta D^2(D+1)^{3/2}e^{D\sqrt{\delta}}\right)
=O\!\left(\delta\cdot\poly(D)\cdot e^{D\sqrt{\delta}}\right).
\]
\end{proof}
\subsection{Tail Bound for Correction Terms}
\begin{lemma}[Tail bound on the squared coefficients]
\label{lem:bi-square-tail}
Let
\[
F(x)=C_F + x + \sum_{j\ge 2} b_j\,P_j(x),
\qquad
P_j(x)=U_j(x/2),
\]
where
\[
F(x)=
\begin{cases}
2x,& x\in[0,2],\\
0,& x\in[-2,0].
\end{cases}
\]
Then
\[
|b_j| = O\!\left(\frac{1}{j^2}\right),
\qquad\text{and hence}\qquad
\sum_{i\ge D} b_i^2 = O\!\left(\frac{1}{D^3}\right).
\]
\end{lemma}

\begin{proof}
Since $\{P_j\}_{j\ge 0}$ forms an orthonormal basis on $[-2,2]$ with respect to the semicircle measure
\[
d\mu(x)=\frac{1}{2\pi}\sqrt{4-x^2}\,dx,
\]
we have for all $j\ge 2$,
\[
b_j=\int_{-2}^2 F(x)\,P_j(x)\,d\mu(x).
\]

Make the change of variables $x=2\cos\theta$, $\theta\in[0,\pi]$. Then
\[
dx=-2\sin\theta\,d\theta,\quad
\sqrt{4-x^2}=2\sin\theta,\quad
P_j(2\cos\theta)=\frac{\sin((j+1)\theta)}{\sin\theta}.
\]
It follows that
\[
d\mu(x)=-\frac{2}{\pi}\sin^2\theta\,d\theta,
\]
and hence
\[
b_j
=
\frac{2}{\pi}\int_0^\pi F(2\cos\theta)\sin\theta\,\sin((j+1)\theta)\,d\theta.
\]

By definition of $F$, we have
\[
F(2\cos\theta)=
\begin{cases}
4\cos\theta,& \theta\in[0,\pi/2],\\
0,& \theta\in[\pi/2,\pi].
\end{cases}
\]
Thus
\[
b_j
=
\frac{8}{\pi}\int_0^{\pi/2}\cos\theta\,\sin\theta\,\sin((j+1)\theta)\,d\theta.
\]
Using $2\sin\theta\cos\theta=\sin(2\theta)$,
\[
b_j
=
\frac{4}{\pi}\int_0^{\pi/2}\sin(2\theta)\sin((j+1)\theta)\,d\theta.
\]
Applying $\sin A\sin B=\tfrac12(\cos(A-B)-\cos(A+B))$, we obtain
\[
b_j
=
\frac{2}{\pi}\int_0^{\pi/2}
\bigl(\cos((j-1)\theta)-\cos((j+3)\theta)\bigr)\,d\theta.
\]
Evaluating the integral gives
\[
b_j
=
\frac{2}{\pi}
\left(
\frac{\sin((j-1)\pi/2)}{j-1}
-
\frac{\sin((j+3)\pi/2)}{j+3}
\right).
\]
Since\[
\sin\!\left(\frac{(j+3)\pi}{2}\right)
=
\sin\!\left(\frac{(j-1)\pi}{2}\right),
\]
we have\[
b_j
=
\frac{2}{\pi}\sin\!\left(\frac{(j-1)\pi}{2}\right)
\left(
\frac{1}{j-1}-\frac{1}{j+3}
\right)
=
\frac{8}{\pi}\,
\frac{\sin((j-1)\pi/2)}{(j-1)(j+3)}.
\]

Hence, for all $j\ge 2$,
\[
|b_j|
\le
\frac{C}{(j-1)(j+3)}
=
O\!\left(\frac{1}{j^2}\right).
\]
Combining the above gives us
\[
\sum_{i\ge D} b_i^2
\le
C\sum_{i\ge D}\frac{1}{i^4}
=
O\!\left(\frac{1}{D^3}\right)\,.
\]
for a sufficiently large constant $D$.
\end{proof}

\subsection{\texorpdfstring{$\ell_1$}{l1} Bound for Normalized Coefficients}
\begin{lemma}
	Let $c(\tau)$ be the normalized coefficient of shape $\tau$ from the iterative scheme, i.e., $c(\tau) = c_\tau \cdot \sqrt{n}^{|V(\tau)|-1}$. For inner truncation threshold $D$, we have \[
	\sum_{\tau\in \calB} |c(\tau) | \leq C^{D^D}
	\] 
	for some constant $C>1$.
\end{lemma}

\begin{proof}
	We make the following observations. Regardless of the outer iteration level, with inner truncation at degree $D$, our final shape has size bound $D_V\leq D^D$ vertices. Moreover, it suffices for us to bound the number of such shapes as $|c_\tau |\leq 1$ for all $\tau$. We apply~\cref{lem:backbone-count} to bound the number of backbone-correction shapes with at most $D^D$ vertices.
\end{proof}

\begin{remark}
	$C^{D^D} \leq n^\delta$ for any $\delta>0$ provided $D=O_{C,\delta}( \frac{\log \log n}{\log \log \log n}  )$.
\end{remark}

\begin{lemma}[Counting backbone-correction shapes] \label{lem:backbone-count}
There exists an absolute constant $C>1$ such that the number of backbone-correction shapes with at most $T$ vertices is at most $C^T$.
\end{lemma}

\begin{proof}
Fix $m\le T$. A backbone-correction shape on $m$ vertices is uniquely determined by the collection of arcs not part of the backbone path. Thus, it suffices to count noncrossing arc sets on $m$ ordered vertices. 

We next appeal to the connection between noncrossing partitions and Catalan numbers. Place the vertices on a convex $m$-gon. Any noncrossing arc set can be extended to a triangulation of the polygon, and every triangulation has exactly $m-3$ diagonals. Therefore the number $N_m$ of such arc sets satisfies
\[
N_m \le \mathrm{C}_{m-2}\,2^{m-3},
\]
where $\mathrm{C}_{m-2}$ is the $(m-2)$-nd Catalan number.
Using $\mathrm{C}_{m-2}\le 4^{m-2}$, we obtain
\[
N_m \le 2^{m-3}4^{m-2}\le 8^m.
\]
Summing over $m\le T$ gives
\[
\sum_{m=1}^T N_m \le \sum_{m=1}^T 8^m \le 9^T.
\]
Hence the total number of backbone-correction shapes on at most $T$ vertices is at most $C^T$ for $C=9$.
\end{proof}

%% file: trunc-convergence.tex
\subsection{Truncated Variance}
\label{sec:appendix-variance-bound}

\begin{lemma}[Exact normalization and early termination]
\label{lem:truncated-var}
Fix $t^*\geq 2$ and $D\geq 2$. For $c_\gamma\geq0$, run the degree-$D$
truncated recursion initialized by
\[
Q_0^{(c_\gamma)}
=
(1+c_\gamma)C_F\widetilde M_{P_1},
\]
and let
\[
S_i(c_\gamma)
\coloneqq
\sum_{\tau\in\calB(Q_i^{(c_\gamma)})}c(\tau)^2.
\]
There exists a unique $c_\gamma=c_\gamma(t^*,D)\geq0$ such that
\[
S_{t^*}(c_\gamma)=1.
\]
Moreover,
\[
c_\gamma(t^*,D)
=
O\!\left(
\frac{1}{D^3}+\frac{1}{(t^*)^3}
\right).
\]

For this choice of $c_\gamma$, the squared coefficient mass of the update
omitted by terminating at level $t^*$ satisfies
\[
\sum_{j=2}^{D}b_j^2
\left(
S_{t^*}(c_\gamma)^j-S_{t^*-1}(c_\gamma)^j
\right)
=
O\!\left(
\frac{1}{D^3}+\frac{1}{(t^*)^3}
\right).
\]
\end{lemma}

\begin{proof}
Parameterize the process by its initial variance
\[
a\coloneqq(1+c_\gamma)^2C_F^2,
\]
and write $S_i(a)$ for the resulting variance, with
$S_{-1}(a)=0$ and $S_0(a)=a$. The variance-increment recurrence,
restricted to degrees $j\leq D$, is
\[
S_{i+1}(a)-S_i(a)
=
\sum_{j=2}^{D}b_j^2
\left(S_i(a)^j-S_{i-1}(a)^j\right).
\]
Telescoping gives
\[
S_{i+1}(a)
=
a+\sum_{j=2}^{D}b_j^2S_i(a)^j.
\tag{E.1}\label{eq:truncated-variance-recurrence}
\]

First consider the original initialization $a_0=C_F^2$. Set
\[
E_i\coloneqq1-S_i(a_0),
\qquad
R_D\coloneqq\sum_{j>D}b_j^2.
\]
Since
\[
C_F^2+\sum_{j\geq2}b_j^2=1,
\]
we have
\[
E_{i+1}
=
R_D+\sum_{j=2}^{D}b_j^2
\left(1-S_i(a_0)^j\right).
\]
Using $1-x^j\leq j(1-x)$ for $x\in[0,1]$,
\[
E_{i+1}
\leq
R_D+\beta E_i,
\qquad
\beta\coloneqq\sum_{j\geq2}j\,b_j^2<1.
\]
Here $\beta<1$ follows from \eqref{eq:variance-contraction-constant}. Consequently,
\[
E_i
\leq
\beta^i(1-C_F^2)+\frac{R_D}{1-\beta}.
\]
By \cref{lem:bi-square-tail}, $R_D=O(D^{-3})$, and since
$\beta<1$, $\beta^i=O(i^{-3})$. Thus
\[
1-S_i(a_0)
=
O\!\left(
\frac{1}{D^3}+\frac{1}{i^3}
\right).
\tag{E.2}\label{eq:truncated-variance-deficit}
\]

For fixed $i$, the map $a\mapsto S_i(a)$ is continuous and strictly
increasing. Moreover, induction in
\cref{eq:truncated-variance-recurrence} gives, for $a'\geq a$,
\[
S_i(a')-S_i(a)\geq a'-a.
\]
Let
\[
\delta_{t^*,D}\coloneqq1-S_{t^*}(a_0).
\]
Then
\[
S_{t^*}(a_0)=1-\delta_{t^*,D},
\qquad
S_{t^*}(a_0+\delta_{t^*,D})\geq1.
\]
Hence there is a unique
\[
a_{t^*,D}
\in[a_0,a_0+\delta_{t^*,D}]
\]
such that $S_{t^*}(a_{t^*,D})=1$.

Define
\[
c_\gamma(t^*,D)
\coloneqq
\frac{\sqrt{a_{t^*,D}}}{C_F}-1.
\]
Then
\[
c_\gamma(t^*,D)
\leq
\frac{\delta_{t^*,D}}{2C_F^2}
=
O\!\left(
\frac{1}{D^3}+\frac{1}{(t^*)^3}
\right).
\]

It remains to bound the omitted update. Since
$a_{t^*,D}\geq a_0$, monotonicity in the initial variance gives
\[
S_{t^*-1}(a_{t^*,D})
\geq
S_{t^*-1}(a_0).
\]
Also, the variance process is nondecreasing in $i$ and
$S_{t^*}(a_{t^*,D})=1$. Therefore,
\[
0
\leq
1-S_{t^*-1}(a_{t^*,D})
\leq
1-S_{t^*-1}(a_0)
=
O\!\left(
\frac{1}{D^3}+\frac{1}{(t^*)^3}
\right).
\]
Finally,
\begin{align*}
\sum_{j=2}^{D}b_j^2
\left(S_{t^*}^j-S_{t^*-1}^j\right)
&=
\sum_{j=2}^{D}b_j^2
\left(1-S_{t^*-1}^j\right)\\
&\leq
\beta\left(1-S_{t^*-1}\right)\\
&=
O\!\left(
\frac{1}{D^3}+\frac{1}{(t^*)^3}
\right).
\end{align*}
\end{proof}

%% file: previous_attempt.tex
\section{Initial Attempt with a Constant Improvement} \label{sec:previous-attempt-appendix}
\paragraph{Explicit Factorization with Value $\approx 1.388$.}
Observe that 
\begin{align*}
\left(aId + b\cmp_{P_1} + c\cmp_{P_2} + d\cmp_{\bar{P_2}}\right)^2 &\approx (a^2 + b^2 + c^2 + d^2)Id + (2ab + 2bc)\cmp_{P_1} \\
&+(b^2 + c^2 + 2ac)\cmp_{P_2} +2ad\cmp_{\bar{P_2}} + 2bc\cmp_{P_3}\\
&+ bd\left(\cmp_{P_1 \circ \bar{P}_2} + \cmp_{\bar{P}_2 \circ P_1}\right) + cd \left(\cmp_{P_2 \circ \bar{P}_2} + \cmp_{\bar{P}_2 \circ P_2}\right)\\
&+{d^2}\frac{\cm_{\bar{P}_2 \circ \bar{P}_2}}{n^2}
\end{align*}
In order to satisfy the clique constraints, we need to add the following correction.
\begin{align*}
&(b^2 + c^2 + 2ac - 2ad)\cmp_{\bar{P}_2} + 2bc\cmp_{\bar{P}_3} + bd\left(\cmp_{\overline{P_1 \circ \bar{P}_2}} + \cmp_{\overline{\bar{P}_2 \circ P_1}}\right) \\
&+ cd \left(\cmp_{\overline{P_2 \circ \bar{P}_2}} + \cmp_{\overline{\bar{P}_2 \circ P_2}}\right)
+{d^2}\cmp_{\overline{\bar{P}_2 \circ \bar{P}_2}}
\end{align*}
Our main theorem shows that the shapes in this correction are freely independent and that the norm of this correction is approximately two times the square root of the sum of squares of the coefficients of these shapes.

Thus, to make the overall matrix PSD, we need to add approximately 
\[
2\sqrt{(b^2 + c^2 + 2ac - 2ad)^2 + 4{b^2}{c^2} + 2{b^2}{d^2} + 2{c^2}{d^2} + d^4}Id
\]
Thus, the overall ratio we obtain is approximately
\[
\frac{2ab + 2bc}{a^2 + b^2 + c^2 + d^2 2\sqrt{(b^2 + c^2 + 2ac - 2ad)^2 + 4{b^2}{c^2} + 2{b^2}{d^2} + 2{c^2}{d^2} + d^4}}
\]

This achieves a value of $0.7204$ by numerical optimization. One maximizing direction is approximately
\[
(a,b,c,d) \propto (0.84506,\;0.51997,\;-0.00614,\;0.12437).
\]

Equivalently, setting $b=1$, an optimizer is
\[
(a,b,c,d) \approx (1.6252,\;1,\;-0.0118,\;0.2392),
\]
which attains the same value.
This corresponds to an upper bound of value $\approx 1.388$.

%% file: lovaszequivalence.tex
\section{Equivalence of \Lovasz-Theta and Degree $2$ Sum of Squares for Independent Set}\label{sec:lovaszequivalence}
To confirm that $\vartheta(G)$ is the same as the value of the degree $2$ sum of squares relaxation for independent set, we use the following semidefinite program for $\vartheta(G)$ which is dual to the semidefinite program in \cref{def:dual-sdp}:

	\[ \max \tr(BJ) \quad  \text{ s.t. } B \succeq 0, \tr(B) = 1, \text{ and }B_{ij} =   0 \text{ whenever } \{i,j\} \in E(G)\]
	where $J$ is the all ones matrix.
\begin{lemma}
If $B$ is an optimal solution of the above SDP then letting $k$ be its objective value, for all $i \in [n]$, $\sum_{j=1}^{n}{B_{ij}} = kB_{ii}$.
\end{lemma}
\begin{proof}
Consider what happens when we tweak $B$ by multipyling the $i$th row and column of $B$ by $1+\epsilon$ for some small $\epsilon > 0$. In order to maintain that $tr(B) = 1$, we need to multiply the other rows and columns of $B$ by $1-\epsilon'$ where $\epsilon' \approx \frac{B_{ii}}{\sum_{j \in [n] \setminus \{i\}}{B_{jj}}} = \frac{B_{ii}}{1-B_{ii}}$.

Let $s_i = \sum_{j \in [n] \setminus \{i\}}{B_{ij}}$, and observe that after applying this tweak to $B$, the change to $tr(BJ)$ is approximately 
\[
\epsilon\left(2B_{ii} + 2\left(1 - \frac{B_{ii}}{1-B_{ii}}\right)s_i - \frac{2B_{ii}}{1-B_{ii}}(k - 2s_i - B_{ii})\right) = \frac{\epsilon}{1-B_{ii}}\left(2s_i-2(k-1)B_{ii}\right).
\]
Since $B$ is an optimal solution of the SDP, we must have that $s_i = (k-1)B_{ii}$ and thus $\sum_{j=1}^{n}{B_{ij}} = s_i + B_{ii} = kB_{ii}$.
\end{proof}
Given an optimal solution to this SDP, we can obtain a solution to the SDP given by \cref{def:primal-sdp} by taking vectors $v_1,\ldots,v_n$ such that $\langle v_i, v_j\rangle = kB_{ij}$. We then take $v_0 = \frac{1}{k}\sum_{i=1}^{n}{v_i}$.

Conversely, given vectors $v_0,v_1,\ldots,v_n$ for the SDP given by \cref{def:primal-sdp}, observe that if $\sum_{i=1}^{n}{\langle v_i, v_0\rangle} = k$ then
\[
k^2 = \left(\sum_{i=1}^{n}{\langle v_i, v_0\rangle}\right)^2 = \left(\langle{\sum_{i=1}^{n}{v_i},v_0\rangle}\right)^2 \leq \langle{\sum_{i=1}^{n}{v_i},\sum_{j=1}^{n}{v_j}}\rangle
\]
so taking $B_{ij} = \frac{1}{k}\langle v_i, v_j\rangle$, we will have that $tr(BJ) \geq k$, $B \succeq 0$, and $tr(B) = 1$, as needed.
\begin{remark}
Note that this implies that for degree $2$ sum of squares, adding in the constraint that the independent set has size exactly $k$ and considering the maximum $k$ for which the resulting program is feasible gives the same answer as just maximizing the size of the independent set.
\end{remark}